%% file: masterarbeit.tex
\documentclass[a4paper,11pt,twoside]{report} 
\pdfoutput=1
\usepackage{etoolbox} 
\newtoggle{print-version}
\togglefalse{print-version}
\nottoggle{print-version}{
  \usepackage[top=30mm, bottom=30mm, outer=40mm, inner=40mm]{geometry} 
}{
  \usepackage[top=30mm, bottom=30mm, outer=55mm, inner=25mm]{geometry} 
}
\usepackage{sty/main}

\begin{document}

\title{Distributed Graph Automata}
\author{Fabian Reiter}

\pagenumbering{alph}
\begin{titlepage}
  \nottoggle{print-version}{
    \newgeometry{top=3cm, bottom=1cm, left=3.25cm, right=3.25cm} 
  }{
    \newgeometry{top=3cm, bottom=1cm, left=3.5cm, right=3cm} 
  }
  \begin{flushleft}
    {\LARGE\fontfamily{iwona}\selectfont \hspace{-.13cm} Fabian Reiter}\\[1cm]
    {\fontfamily{iwona}\fontsize{1.4cm}{1cm}\selectfont\textbf{Distributed\\[.5cm] Graph Automata}}\\[.5cm]
    {\color{palegray} \rule{\linewidth}{.5ex}}\\[1.5cm]
  \end{flushleft}
  \begin{center}
    \input{fig/cover_illustration.tikz}
  \end{center}
  \vspace{.5cm}
  \begin{flushright}
    \advance\rightskip -1.5cm
    {\Large\fontfamily{iwona}\selectfont University of Freiburg}\hspace{2ex}\input{fig/uni_freiburg_bars.tikz}
  \end{flushright}
  \restoregeometry
\end{titlepage}

\clearpage
\thispagestyle{empty}
{\Large\fontfamily{iwona}\selectfont Master's Thesis in Computer Science}
\vfill
\begin{description}[style=nextline]
\item[Supervisors]
  \begin{tabular}{ll}
    Prof. Dr. Fabian Kuhn & (Chair of Algorithms and Complexity) \\
    Prof. Dr. Andreas Podelski & (Chair of Software Engineering)
  \end{tabular}
\item[Institution\vphantom{p}]
  \begin{tabular}{ll}
    Albert Ludwig University of Freiburg & (Germany) \\
    Faculty of Engineering \\
    Department of Computer Science
  \end{tabular}
\item[Date of Submission\quad\textnormal{(original version)}]
  \begin{tabular}{l}
    January 20\textsuperscript{th}, 2014
  \end{tabular}
\item[Date of Revision\quad\textnormal{(this version)}]
  \begin{tabular}{l}
    April 25\textsuperscript{th}, 2014
  \end{tabular}
\item[Author and Contact Address]
  \begin{tabular}{l}
    Fabian Reiter \quad \href{mailto:fabian.reiter@gmail.com}{\nolinkurl{fabian.reiter@gmail.com}}
  \end{tabular}
\end{description}
\vspace{1cm}

\clearpage
\pagenumbering{roman}

\thispagestyle{empty}
\section*{Abstract}
Inspired by distributed algorithms, we introduce a new class of finite
graph automata that recognize precisely the graph languages definable
in monadic second-order logic. For the cases of words and trees, it
has been long known that the regular languages are precisely those
definable in monadic second-order logic. In this regard, the automata
proposed in the present work can be seen, to some extent, as a
generalization of finite automata to graphs.

Furthermore, we show that, unlike for finite automata on words and
trees, the deterministic, nondeterministic and alternating variants of
our automata form a strict hierarchy with respect to their expressive
power. For the weaker variants, the emptiness problem is decidable.

\vspace{.5cm}
\section*{\it\fontfamily{iwona}\selectfont Zusammenfassung}
{\it
  Inspiriert durch verteilte Algorithmen führen wir eine neue Klasse
  von endlichen Graph-Automaten ein, die genau die Graph-Sprachen
  erkennen, die in mo\-na\-di\-scher Prädikatenlogik zweiter Stufe
  definierbar sind. Für Worte und Bäume ist seit langem bekannt, dass
  die regulären Sprachen genau jene sind, die in mo\-na\-di\-scher
  Prädikatenlogik zweiter Stufe definierbar sind. In dieser Hinsicht
  können die in vorliegender Arbeit vorgestellten Automaten gewissermaßen 
  als eine Verallgemeinerung von endlichen Automaten auf Graphen
  betrachtet werden.

  Ferner zeigen wir, dass im Gegensatz zu endlichen Automaten auf
  Worten und Bäumen die de\-ter\-mi\-nis\-tisch\-en,
  nichtdeterministischen und alternierenden Varianten unserer
  Automaten eine strikte Hierarchie bezüglich ihrer Ausdrucksstärke
  bilden. Für die schwächeren Varianten ist das Leerheitsproblem
  entscheidbar.  }

\clearpage
\thispagestyle{empty}
\section*{Acknowledgments}
I would like to thank my supervisors, Fabian Kuhn and Andreas
Podelski, for many helpful and pleasant discussions, and their
continuous support of a project whose outcome was unpredictable in its
early stages. Furthermore, I am very grateful to Jan Leike, who kindly
read and commented on drafts of this thesis.

\tableofcontents

\cleardoublepage
\pagenumbering{arabic}
\input{chap/introduction.tex}
\input{chap/graphs.tex}
\input{chap/ADGAs.tex}
\input{chap/MSOL.tex}
\input{chap/NDGAs_DDGAs.tex}
\input{chap/conclusion.tex}



\end{document}

%% file: fig/cover_illustration.tikz
\newsavebox\automatonbox
\begin{lrbox}{\automatonbox}\input{fig/ADGA_concentric_circles.tikz}\end{lrbox}
\newsavebox\graphbox
\begin{lrbox}{\graphbox}\input{fig/graph_labeled_pentagon.tikz}\end{lrbox}
\newsavebox\runbox
\begin{lrbox}{\runbox}\input{fig/run_accepting.tikz}\end{lrbox}
\begin{tikzpicture}
  \node (automaton) {\usebox\automatonbox};
  \node[below=1ex of automaton,xshift=-2ex] (Alabel) {\large\fontfamily{iwona}\selectfont \fade{Automaton}};

  \node[right=9ex of automaton] (graph) {\usebox\graphbox};
  \node[below=1ex of graph] (Glabel) {\large\fontfamily{iwona}\selectfont \fade{Input Graph}};

  \node[below=16ex of automaton,xshift=20ex] (run) {\usebox\runbox};
  \node[above=-9ex of run,xshift=0ex] (Rlabel) {\large\fontfamily{iwona}\selectfont \fade{Run}};

  \begin{scope}[->,>=stealth,line width=2ex,lightgray,line cap=round]
    \draw ([xshift=2ex,yshift=-1ex]Alabel.south east) to [out=-30,in=110] ([xshift=0ex,yshift=2ex]Rlabel.north west);
    \draw ([xshift=-2ex,yshift=-1ex]Glabel.south west) to [out=-150,in=70] ([xshift=0ex,yshift=2ex]Rlabel.north east);
  \end{scope}
\end{tikzpicture}

%% file: fig/ADGA_concentric_circles.tikz
\begin{tikzpicture}[automaton, half row sep]
  \matrix[states] {
        &[4ex]&[7ex] \node[permanent] (q_ap) {$\qap$}; \\
    \node[existential] (q_a) {$\qa$}; & \node[universal] (q_a2) {$\qaprime$};  \\
        &     & \node[permanent] (q_ah) {$\qah$}; \\
        & \node[universal] (q_bkr) {$\qbkr$}; \\
    \node[existential] (q_b) {$\qb$}; &     & \node[permanent] (q_yes) {$\q{yes}$}; \\
        & \node[universal] (q_bk) {$\qbk$}; \\
        \\
    \node[existential] (q_c) {$\qc$}; &     & \node[permanent] (q_no) {$\q{no}$}; \\
  };
  \matrix[symbols] { 
        \\
    \node (a) {$\a$}; \\
        \\
        \\
    \node (b) {$\b$}; \\
        \\
        \\
    \node (c) {$\c$}; \\
  };
  \matrix[accepting sets] {
    \x &    \\
       &    \\
       & \x \\
       &    \\
    \x & \x \\
       &    \\
       &    \\
       &    \\
  };
  \DrawColumnBackground{1}{8}{2}
  \path (a)        edge (q_a)
        (b)        edge (q_b)
        (c)        edge (q_c)
        (q_a)      edge (q_a2)
        (q_b.25)   edge node[above] {$\xnni \qb$} (q_bkr)
        (q_b.0)    edge node[above,xshift=.2ex] {$\xnni \qb$} (q_bk)
        (q_b.330)  edge[bend right=20] node[above=4.2ex,xshift=-9.2ex] {$\xni \qb$} (q_no) 
        (q_c.0)    edge[bend right=25] node[above=-.5ex,xshift=-12ex] {$\xnni \qc ∧ \xnni \qa$} (q_yes)
        (q_c.330)  edge[bend right=22] node[above=.2ex] {$\xni \qc ∨ \xni \qa$} (q_no)
        (q_a2.22)  edge[bend left=8] node[above=.3ex,xshift=.1ex] {$\xeq \{\qbkr,\qbk\}$} (q_ap)
        (q_a2)     edge node[above=.1ex,xshift=.4ex] {$\xeq \{\qbkr,\qbk\}$} (q_ah)
        (q_a2.320) edge[bend right=9] node[above=4.5ex,xshift=1.3ex] {$\xneq \{\qbkr,\qbk\}$} (q_no)
        (q_bkr)    edge (q_yes)
        (q_bk)     edge (q_yes);
\end{tikzpicture}

%% file: fig/graph_labeled_pentagon.tikz
%
\pentagraphPic{input graph}{\lnodedistIG}{$\a$}{$\b$}{$\c$}{$\b$}{$\b$}{$\c$}

%% file: fig/run_accepting.tikz
\begin{tikzpicture}[run or game]
  \matrix {
      & & \node[config,pacc] (c3a)
           {\pentagraphPic{configuration}{\lnodedistC}{$\qap$}{$\q{yes}$}{$\q{yes}$}{$\q{yes}$}{$\q{yes}$}{$\q{yes}$}}; \\
    \node[config,exis] (c1)
     {\pentagraphPic{configuration}{\lnodedistC}{$\qa$}{$\qb$}{$\qc$}{$\qb$}{$\qb$}{$\qc$}}; 
      & \node[config,univ] (c2)
         {\pentagraphPic{configuration}{\lnodedistC}{$\qaprime$}{$\qbkr$}{$\q{yes}$}{$\qbkr$}{$\qbk$}{$\q{yes}$}}; \\
      & & \node[config,pacc] (c3b)
           {\pentagraphPic{configuration}{\lnodedistC}{$\qah$}{$\q{yes}$}{$\q{yes}$}{$\q{yes}$}{$\q{yes}$}{$\q{yes}$}}; \\
  };
  \path (c1) edge (c2)
        (c2) edge (c3a)
             edge (c3b);
\end{tikzpicture}

%% file: fig/uni_freiburg_bars.tikz
\begin{tikzpicture}[baseline={([yshift=.8ex]b.south)}]
  \matrix[row sep=.75ex,
          nodes={minimum height=2.9ex, minimum width=.75ex,
                 inner sep=0ex, text height=0ex}]
  {
    \node[fill=unired]  {}; \\
    \node[fill=uniblue] {}; \\
    \node[fill=uniblue] (b) {}; \\
    \node[fill=unired]  {}; \\
  };
  \pgfresetboundingbox
\end{tikzpicture}

%% file: chap/introduction.tex

\chapter{Introduction}
The research for this thesis started with an open-ended (and perhaps
naive) question: \emph{what can we obtain by connecting finite
  automata in a synchronous distributed setting?} As it turns out, a
possible answer is: a new class of automata that can be seen, to some
extent, as a generalization of finite automata to graphs. In order to
substantiate this claim, we begin by reviewing a fundamental result of
formal language theory, and then use that result as a guide within the
less well-explored world of graph languages.

\section{Background and Related Work}
In the early 1960s, a beautiful connection between automata theory and
formal logic was discovered. Independently of each other, Büchi
\cite{Buc60}, Elgot \cite{Elg61} and Trakhtenbrot \cite{Tra61} showed
that the regular languages, recognized by finite automata, are
precisely the languages defined by a certain class of logical
formulas. This idea might be best understood through a simple
example. The following one is borrowed from Thomas \cite{Tho91}.

\begin{example} \label{ex:NFA-MSO}
  Consider the nondeterministic finite automaton $\dA[no]{\b\b}$
  specified in \cref{fig:NFA}. If we exclude the empty word, this
  automaton accepts a finite word $w$ over the alphabet $Σ=\{\a,\b\}$
  \Iff $w$ does not contain the segment $\b$$\b$ and the last symbol
  of $w$ is an $\a$. We can define the same language by the following
  first-order formula:
  \begin{equation*}
    \dphi[no]{\b\b} \coloneqq \;
    \logic{¬\,∃u,v\Bigl(\lab{\b}u \,∧\, u\!\arr\!v \,∧\, \lab{\b}v\Bigr)
      \,∧\, ∃u\Bigl(\lab{\a}u \,∧\, ¬\,∃v\bigl(u\!\arr\!v\bigr)\Bigr)}
  \end{equation*}
  The idea is that we identify each word with a labeled directed graph
  that consists of a single path. For instance, $\a$$\b$$\a$
  corresponds to \input{fig/graph_word_aba.tikz}\!\!. Such a graph is
  a relational structure over which we can evaluate the truth of the
  formula $\dphi[no]{\b\b}$. Variables like $\lsymb{u}$ and
  $\lsymb{v}$ range over the nodes of the graph, $\lsymb{\arr}$
  represents the edge relation, and the symbols $\lsymb{\lab{\a}}$ and
  $\lsymb{\lab{\b}}$ are to be interpreted as unary relations
  indicating that a node is labeled by an $\a$ and a $\b$,
  respectively.

  The first conjunct of $\dphi[no]{\b\b}$ specifies that no two
  consecutive nodes are both labeled by a $\b$, while the second
  conjunct ensures that the last node is labeled by an $\a$.
\end{example}

\begin{SCfigure}[1.4][h!]
  \alignpic
  \input{fig/NFA.tikz}
  \caption{$\dA[no]{\b\b}$, a nondeterministic finite automaton whose
    language, when restricted to nonempty words, consists of all the
    finite words over the alphabet $Σ=\{\a,\b\}$ that do not contain
    the segment $\b$$\b$ and that end with an $\a$.}
  \label{fig:NFA}
\end{SCfigure}
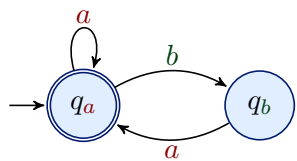

In the formula $\dphi[no]{\b\b}$ of the preceding example, we only
used quantifiers that range over the nodes of a graph. By additionally
allowing quantification over sets of nodes, we reach the full extent
of \emph{monadic second-order} (MSO) logic. The famous result
established by Büchi, Elgot and Trakhtenbrot states that we can
effectively translate every finite automaton to an equivalent
MSO-formula (with relation symbols fixed as in \cref{ex:NFA-MSO}), and
vice versa. (For a proof, see, e.g., \cite[Thm~3.1]{Tho96}.) An
important consequence of this equivalence is that the satisfiability
and validity problems of MSO-logic on words are decidable, because so
are the corresponding problems for finite automata. This application
was the original motivation for establishing a connection between the
two worlds. Nowadays, this and similar connections also play a central
role in model checking, where one needs to decide whether a system,
represented by an automaton, satisfies a given specification,
expressed as a logical formula.

About a decade later, the result was generalized from words to labeled
trees by Thatcher, Wright \cite{TW68} and Doner \cite{Don70} (see,
e.g., \cite[Thm~3.8]{Tho96}). The corresponding tree automata
(which we shall not consider here) can be seen as a natural extension
of finite automata to labeled trees. As far as MSO-logic is concerned,
the generalization to (ordered, directed) labeled trees is
straightforward, since, just like words, these can be regarded as
labeled graphs. We only need to introduce additional edge relation
symbols of the form $\lsymb{\xarr{1}},\lsymb{\xarr{2}},…$, in order to
be able to express that some node is the $i$-th child of another node.

In view of these results, it seems natural to ask whether the bridge
between automata theory and logic persists if we expand our field of
interest from words or trees to arbitrary node labeled graphs
(possibly with multiple edge relations, as for trees). However, the
trouble is that this question is not well-defined. While we can easily
specify what we mean by MSO-logic on graphs, it is not obvious at all
how finite automata should be canonically generalized to graphs that
go beyond trees.

\begin{quote}
  \emph{A result similar to those \emph{[for word and tree languages]}
    does not exist for graph languages, for the trivial reason that
    there is no agreement on what would be the class of “regular graph
    languages”, and, in particular, that there is no accepted notion
    of “finite graph automaton”.} \\
  \hspace*{\fill} (Joost Engelfriet, 1991 \cite[p.~139]{Eng91})
\end{quote}

Nevertheless, graph languages have been an active area of research for
nearly fifty years. In large part, this has been driven by
investigations of generative devices known as \emph{graph grammars}.
The theory of graph grammars is by now well-developed, as can be seen
from the “Handbook of Graph Grammars and Computing by Graph
Transformation” \cite{HGG97-99}, a comprehensive reference consisting
of several volumes. Within this branch of research, MSO-logic has
raised considerable interest. Especially through the work of
Courcelle, MSO-formulas have proven to be particularly useful tools
for obtaining decidability results about graph languages.

\begin{quote}
  \emph{The deep reason why \emph{[monadic second-order]} logic is so
    crucial is that it replaces for graphs \emph{[…\!]} the notion of
    a finite automaton which is very important in the theory of formal
    languages. It “replaces” because no convenient notion of finite
    automaton is known for graphs.} \\
  \hspace*{\fill} (Bruno Courcelle, 1997 \cite[p.~315]{Cou97})
\end{quote}

Supported by the equivalence of regularity and MSO-definability on
words and labeled trees, one might even go as far as referring to the
sets of graphs that can be defined by MSO-formulas as the “regular
graph languages”. Hence, one way to cope with the lack of a canonical
generalization of finite automata to graphs is to search for a model
of computation on graphs that has the same expressive power as
MSO-logic. This is the approach pursued in this thesis.

It must be emphasized that the present work is not, by any means, the
first to investigate graph automata. The definitions suggested in the
literature over the last decades are far too numerous to survey here,
but let us at least mention a small selection. Already in the early
days of graph grammars, mostly in the 1970s, the notion of
graph-accepting machines was explored in parallel to generative
devices. Some examples, among many others, are the models proposed by
Shah, Milgram, Wu and Rosenfeld in \cite{SMR73}, \cite{Mil75} and
\cite{WR79}. However, none of those studies were concerned with
equivalence to MSO-logic, and few of them were pursued much
further. It seems that graph grammars received much more interest than
graph automata. Later, in the early 1990s, Thomas introduced his
\emph{graph acceptors} in \cite{Tho91}, with the explicit goal of an
automata-theoretic investigation of MSO-definable graph properties. It
turned out that Thomas' graph acceptors recognize precisely the graph
languages of bounded degree definable in the existential
fragment\footnote{The existential fragment of MSO-logic consists of
  formulas of the form $\logic{∃U_1,\meta{…}\,,U_n(\meta{φ})}$, where
  $\lsymb{U_1},…,\lsymb{U_n}$ are set variables, and $φ$ is a
  first-order formula (i.e., $φ$ does not contain any set
  quantifiers).} of MSO-logic (see \cite[Thm~3]{Tho97}). This makes
them less expressive than full MSO-logic, and, in particular, their
class of recognizable languages is not closed under
complementation. Also, about a year earlier, Courcelle had introduced
in \cite{Cou90} an algebraic notion of \emph{recognizability}, without
defining any notion of graph automaton. Every MSO-definable graph
language is recognizable in Courcelle's sense, but not vice versa.

The expressiveness of MSO-logic on graphs has thus been approximated
“from below”, by Thomas, and “from above”, by Courcelle, but, to the
author's best knowledge, a perfectly matching automaton model has been
missing so far. Relatively recent remarks by Courcelle and Engelfriet
in \cite{CE12}, as well as the following explicit statement, support
this assumption.

\begin{quote}
  \emph{No existing notion of graph automaton gives an equivalence
    with monadic second-order logic.} \\
  \hspace*{\fill} (Bruno Courcelle, 2008 \cite[p.~8]{Cou08})
\end{quote}

The present work is an attempt to close this gap in the theory of
graph languages. It is successful in the sense that it provides a
class of graph automata equivalent to MSO-logic. However, it must also
be conceded that, up to now, no new results have been inferred from
this alternative characterization. Whether it will prove as fruitful
as classical automata on words and trees remains to be seen.

\section{Structure of this Thesis}
The presentation is organized as follows: After some preliminaries on
graphs in \cref{chap:graphs}, we introduce the alternating variant of
our distributed graph automata (ADGAs) in \cref{chap:adga}, and
discuss some of their properties. Since the capabilities of ADGAs
might not be obvious at first sight, a substantial part of the chapter
is devoted to examples. Then, in \cref{chap:msol}, we review MSO-logic
on graphs, and prove our main result, the equivalence of MSO-logic and
ADGAs. This immediately entails some negative results on ADGAs. We
finish by considering nondeterministic and deterministic variants of
our automata in \cref{chap:ndga_ddga}. Both turn out to be strictly
weaker than ADGAs, and they also form a hierarchy among themselves.
The loss of expressive power is however rewarded by a decidable
emptiness problem.

%% file: fig/graph_word_aba.tikz
\begin{tikzpicture}[lgraph,baseline=(u.text)]
  \node[lnode] (u) {$\a$};
  \node[lnode] (v) at ([shift={(0:\lnodedistIG)}]u) {$\b$};
  \node[lnode] (w) at ([shift={(0:\lnodedistIG)}]v) {$\a$};
  \path (u) edge (v)
        (v) edge (w);
\end{tikzpicture}

%% file: fig/NFA.tikz
\begin{tikzpicture}[automaton]
  \matrix[states] {
    \node[NFA state,initial,accepting] (q_a) {$q_\a$}; & \node[NFA state] (q_b) {$q_\b$}; \\
  };
  \path (q_a) edge[loop above] node {$\a$} (q_a)
              edge[bend left] node {$\b$} (q_b)
        (q_b) edge[bend left] node {$\a$} (q_a);
\end{tikzpicture}

%% file: chap/graphs.tex

\chapter{Preliminaries on Graphs} \label{chap:graphs}
Graphs play a central role in this work. On the one hand, they will
serve as input for automata, and as models for logical formulas. On
the other hand, we will use them to describe the behaviour of
automata, and to represent two-player games. In this chapter, we
provide formal definitions, review some common graph properties, and
discuss the notion of graph minors.

\section{Basic Definitions}
As our most general concept, we consider directed graphs with nodes
labeled by an alphabet $Σ$, and multiple edge relations indexed by an
alphabet $Γ$. While node labels are auxiliary, we regard edge labels
as an integral part of the graph structure.

\begin{definition}[$Γ$-Graph]
  Let $Γ$ be a nonempty finite alphabet (i.e., a set of symbols). A
  \defd{$Γ$-graph} $G$ is a structure
  $\bigl\langle\VG,\,⟨\arrG{γ}⟩_{\:\!γ∈Γ}\bigr\rangle$, where
  \begin{itemize}
  \item $\VG$ is a nonempty finite set of \defd{nodes}, and
  \item each ${\arrG{γ}}⊆\VG×\VG$ is a set of directed \defd{edges}
    labeled by $γ∈Γ$.
  \end{itemize}
\end{definition}

If $Γ$ is understood or irrelevant, we refer to $G$ simply as a
\defd{graph}. Note that self-loops are allowed, and that there can be
multiple edges from one node to another, but at most one for every
edge label $γ∈Γ$. If there are no self-loops, i.e., if $v\arrG{γ}v$
does not hold for any $v∈\VG$ and $γ∈Γ$, then we say that $G$ is
\defd{loop-free}. Furthermore, if there is only a single edge
relation, we call $G$ a \defd{simple} graph. In such a case, we set
$Γ=\{\blank\}$ (singleton consisting of a blank symbol), and omit the
superfluous edge labels.
  
\begin{definition}[$Σ$-Labeled $Γ$-Graph]
  Let $Σ$ and $Γ$ be two nonempty finite alphabets. A
  \defd{$Σ$-labeled $Γ$-graph} is a tuple $⟨G,λ⟩$, denoted as
  $\defd{G_λ}$, where
  \begin{itemize}
  \item $G$ is a $Γ$-graph (referred to as the underlying graph),
    and
  \item $λ\colon \VG→Σ$ is a node \defd{labeling} (function).
  \end{itemize}
\end{definition}

Again, we will often relax our terminology, and refer to $G_λ$ simply
as a \defd{labeled graph}, or even as a graph, if the meaning is clear
from the context.

At this point, it is important to mention that we are only interested
in (labeled) graphs \emph{up to isomorphism}. That is, we consider two
$Σ$-labeled $Γ$-graphs $G_λ$ and $G'_{λ'}$ to be \defd{equal} if there
is a bijection $f\colon \VG→\VGpr$, such that $λ(v)=λ'(f(v))$, and
$u\arrG{γ}v$ \Iff $f(u)\arrGpr{γ}f(v)$, for all $u,v∈\VG$ and
$γ∈Γ$. The reason for this is that our automata and logical formulas
cannot distinguish between isomorphic graphs.

In order to draw an analogy to formal language theory on words, we
introduce an extension to graphs of the well-known notation employing
alphabet exponentiation and the Kleene star. In the context of words,
$Σ^n$ designates the set of all words of length $n$ over the alphabet
$Σ$, and $Σ^*$ the set of all words of arbitrary (finite) length over
$Σ$. Now, a word over $Σ$ can be viewed as a $Σ$-labeled linear graph
with a single edge relation. For example, we can identify
$\a$$\c$$\a$$\b$ with \input{fig/graph_word_acab.tikz}\!\!. From this
point of view, the number $n$ in the expression $Σ^n$ refers to the
underlying linear graph of length $n$, and the Kleene star $*$ in
$Σ^*$ can be seen as a placeholder for any linear graph. We can
generalize this notation to arbitrary finite graphs by replacing $n$
with a $Γ$-graph $G$, and $*$ with the symbol $\clouded{Γ}$ (“clouded
$Γ$”), which serves as a placeholder for any $Γ$-graph.

\begin{definition}[Cloud Notation]
  For any nonempty finite alphabets $Σ$ and $Γ$, and any $Γ$-graph
  $G$, we denote by $Σ^G$ the set of all $Σ$-labeled graphs with
  underlying graph $G$, and by $Σ^{\clouded{Γ}}$~the set of all
  $Σ$-labeled $Γ$-graphs, i.e.,
  \begin{equation*}
    \defd{Σ^G} ≔ \{G_λ \mid λ\colon\VG→Σ\},\!
    \quad \text{and} \quad
    \defd{Σ^{\clouded{Γ}}} ≔ \,\smashoperator{\bigcup_{G∈\mathcal{G}(Γ)}}\, Σ^G,
  \end{equation*} \\[-3ex]
  where $\mathcal{G}(Γ)$ is the set of all $Γ$-graphs.
\end{definition}

Occasionally, the need will arise to consider graphs labeled
heterogeneously with two alphabets $Σ_1$ and $Σ_2$, such that at least
one node is labeled by some symbol from $Σ_1$ that is not contained in
$Σ_2$. In such cases, we shall employ the notational shorthands
\begin{align*}
  \swl{\defd{⟨Σ_1,Σ_2⟩^G}}{⟨Σ_1,Σ_2⟩^{\clouded{Γ}}} &\coloneqq (Σ_1∪Σ_2)^G\setminus (Σ_2)^G,
  \quad \text{and} \\
  \defd{⟨Σ_1,Σ_2⟩^{\clouded{Γ}}} &\coloneqq (Σ_1∪Σ_2)^{\clouded{Γ}}\setminus (Σ_2)^{\clouded{Γ}}.
\end{align*}
 
We will use automata and logical formulas to characterize sets of
labeled graphs. Pursuing the analogy with words, such sets are called
graph languages.

\begin{definition}[Graph Language]
  A \defd{graph language} is a set of labeled graphs. More precisely,
  $L$ is a graph language \Iff there are finite alphabets $Σ$ and $Γ$,
  such that $L⊆Σ^{\clouded{Γ}}$.
\end{definition}

In many examples, we will consider graph languages for which the node
labels are irrelevant. In such cases, we fix the node alphabet to be a
singleton $Σ=\{\blank\}$,\footnote{Since $Σ$ and $Γ$ do not need to be
  disjoint, we can use the same blank symbol $\blank$ for both.} and,
to simplify notation, we identify any labeled graph
$G_λ∈Σ^{\clouded{Γ}}$ with its underlying graph $G$.

Furthermore, when reasoning about graphs as structural objects, we
will follow the usual terminology of graph theory. In particular,
given a $Γ$-graph $G$ and two nodes $u,v∈\VG$, we say that $u$ is an
\defd{incoming neighbor} of $v$, and $v$ an \defd{outgoing neighbor}
of $u$, if $u\arrG{γ}v$ for some $γ∈Γ$. A node without incoming
neighbors is called a \defd{source}, and a node without outgoing
neighbors a \defd{sink}. Without further qualification, the term
\defd{neighbor} refers to both incoming and outgoing neighbors. The
(undirected) \defd{neighborhood} of a node is the set of all of its
neighbors. Accordingly, the incoming and outgoing neighborhoods
contain only the incoming and outgoing neighbors, respectively. If we
additionally qualify a neighborhood of a node $v$ as \defd{closed},
it means that we also include $v$ itself into the set.

A (directed) \defd{path} from $u$ to $v$ is a sequence of nodes,
starting with $u$ and ending with $v$, in which each node but the last
is (directly) followed by one of its outgoing neighbors. If a
subsequent neighbor does not necessarily have to be outgoing, we call
the sequence an \defd{undirected path}.

We say that a $Γ$-graph $H$ is a \defd{subgraph} of another $Γ$-graph
$G$ (or that $G$ contains $H$ as a subgraph) if $\VH⊆\VG$ and
${\arrH{γ}}⊆({\arrG{γ}}∩\VHsquared)$ for all $γ∈Γ$. If a subgraph $H$
contains all the edges between nodes in $\VH$ that occur in $G$, i.e.,
if ${\arrH{γ}}=({\arrG{γ}}∩\VHsquared)$ for all $γ∈Γ$, then we call
$H$ the subgraph of $G$ \defd{induced} by $\VH$, and denote it by
$\defd{G[\VH]}$.

\section{Some Graph Properties} \label{sec:graph-properties}
We now briefly recall some standard graph properties, which will serve
us as examples of graph languages in \cref{chap:adga,chap:msol}. Let
$Σ$ and $Γ$ be two nonempty finite alphabets, and $G_λ$ some
$Σ$-labeled $Γ$-graph. In the following, if node labels are irrelevant
for some graph property, we only refer to the underlying graph $G$,
but, of course, the same properties also apply to labeled graphs.

The node labeling $λ$ constitutes a valid \defd{coloring} of $G$ if no
two adjacent nodes (neighbors) share the same label, i.e.,
$u\arrG{γ}v$ implies $λ(u)≠λ(v)$, for all $u,v∈\VG$ and $γ∈Γ$. If
$\card{Σ}=k$, such a labeling $λ$ is called a $k$-coloring of $G$, and
any $Γ$-graph for which a $k$-coloring exists is said to be
\defd{$k$-colorable}. Note that, by definition, a graph that contains
self-loops is not $k$-colorable for any $k≥1$.

We call $G$ (weakly) \defd{connected} if there is an undirected path
between every two nodes $u,v∈\VG$.

If we want to deal with undirected graphs, we can represent them as
directed graphs with bidirectional edges. Formally, $G$ is
\defd{undirected} if for every $u,v∈\VG$ and $γ∈Γ$, it holds that
$u\arrG{γ}v$ \Iff $v\arrG{γ}u$. For undirected graphs, we shall not
distinguish between the $γ$-edge from $u$ to $v$ and the $γ$-edge from
$v$ to $u$. Instead, we refer collectively to both of them as
\emph{the} $γ$-edge between $u$ and $v$.

A (simple) \defd{cycle} in $G$ is a path that starts and ends at the
same node $v∈\VG$, such that $v$ occurs precisely two times, and every
other node in $\VG$ occurs at most once. If such a cycle contains
every node in $\VG$, it is called a \defd{Hamiltonian cycle}.

We say that a set of edges $M⊆\bigcup_{γ∈Γ}({\arrG{γ}})$ is a
\defd{perfect matching} of $G$ if no two edges in $M$ share a common
node, and every node in $\VG$ is covered by some edge in $M$.

A \defd{nontrivial automorphism} of $G_λ$ is a bijection $f\colon
\VG→\VG$ that is not an identity, such that $λ(v)=λ(f(v))$, and
$u\arrG{γ}v$ \Iff $f(u)\arrG{γ}f(v)$, for all $u,v∈\VG$ and $γ∈Γ$.

The graph $G$ is \defd{planar} if it can be drawn in the plane such
that no two edges intersect each other, i.e., they may only meet at
nodes. We will give a more formal characterization of planarity in the
next section (see \cref{thm:kuratowski-wagner}).

\section{Graph Minors}
The issue with the definition of a planar graph given above is that it
refers to the notion of a plane, an object that is not part of the
graph. In order to make planarity accessible to our automata, we need
a specification that involves only the graph structure itself. We will
exploit an important result in graph theory for this purpose: the
characterization of planarity in terms of forbidden minors. In this
context, we only consider unlabeled, simple, undirected graphs.

Graph minors can be defined by means of edge contractions. Given a
simple undirected graph $G$, \defd{contracting} the edge between two
nodes $u,v∈\VG$ means to remove that edge and merge the nodes $u$ and
$v$, such that every node that was a neighbor of $u$ or $v$ becomes a
neighbor of the merged node.

\begin{definition}[Minor]
  Let $G$ and $H$ be two simple undirected graphs, such that $H$ is
  loop-free. We say that $H$ is a \defd{minor} of $G$ (or that $G$
  contains $H$ as a minor) if we can obtain $H$ by taking a subgraph
  of $G$ and repeatedly contracting edges.
\end{definition}

\begin{example}
  Consider the graph in \cref{fig:graph_5_cycle_as_subgraph}. Removing
  the nodes and edges highlighted in red yields the subgraph shown in
  \cref{fig:graph_5_cycle}. Then, by contracting the edges highlighted
  in green, we obtain $\K{3}$, the complete graph with three nodes
  depicted in \cref{fig:graph_k3}. Hence, the considered graph
  contains $\K{3}$ as a minor.
\end{example}
 
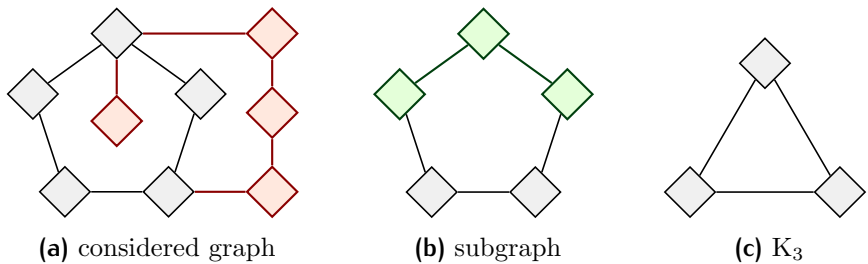
\begin{figure}[h!]
  \alignpic
  \begin{subfigure}[b]{0.35\textwidth}
    \centering
    \input{fig/graph_5_cycle_as_subgraph.tikz}
    \caption{considered graph}
    \label{fig:graph_5_cycle_as_subgraph}
  \end{subfigure}
  \begin{subfigure}[b]{0.3\textwidth}
    \centering
    \input{fig/graph_5_cycle.tikz}
    \caption{subgraph}
    \label{fig:graph_5_cycle}
  \end{subfigure}
  \begin{subfigure}[b]{0.25\textwidth}
    \centering
    \input{fig/graph_k3.tikz}
    \caption{$\K{3}$}
    \label{fig:graph_k3}
  \end{subfigure}
  \caption{A graph that contains the complete graph with three nodes
    ($\K{3}$) as a minor.}
\end{figure}

In the preceding example, we could obtain $\K{3}$ as a minor because
the considered graph contains a sufficiently large cycle. It is easy
to see that this is a necessary and sufficient condition.

\begin{remark} \label{rem:k3-cycle}
  A simple undirected graph contains $\K{3}$ as a minor \Iff it
  contains at least one cycle of three or more nodes.
\end{remark}

While the notion of contracting edges is intuitive, it is not directly
available in the logical and automata-theoretic formalisms that we
will consider. Instead, we will use the following characterization of
minor inclusion given by Courcelle and Engelfriet in \cite[Lemma
1.13]{CE12}.

\begin{lemma}[Minor Inclusion] \label{lem:minor-inclusion}
  Let $G$ and $H$ be two simple undirected graphs, such that $H$ is
  loop-free and $\VH=\{v_1,…,v_n\}$. Then $H$ is a minor of $G$ \Iff
  there exist pairwise disjoint nonempty sets of nodes
  $U_1,…,U_n⊆\VG$, such that 
  \begin{itemize}[topsep=1ex,itemsep=0ex]
  \item each induced subgraph $G[U_i]$ is connected, for $1≤i≤n$, and
  \item for every edge in $H$ between two nodes $v_i,v_j∈\VH$, there
    exists an edge in $G$ between two nodes $u,v∈\VG$ such that
    $u∈U_i$ and $v∈U_j$.
  \end{itemize}
\end{lemma}

We now come back to planarity. An important theorem by Kuratowski
\cite{Kur30} characterizes the planar graphs in terms of two forbidden
graphs: the complete graph with five nodes $\K{5}$, and the complete
bipartite graph with two times three nodes $\K{3,3}$, both depicted in
\cref{fig:graph_k5_and_k3_3}. In Wagner's variant of the theorem
\cite{Wag37}, which we shall use, those forbidden graphs may not occur
as minors (for a proof, see, e.g., \cite[Thm~4.4.6]{Die10}).

\begin{theorem}[Kuratowski-Wagner Theorem] \label{thm:kuratowski-wagner}
  A simple undirected graph is planar \Iff it contains \emph{neither}
  $\K{5}$ \emph{nor} $\K{3,3}$ as a minor.
\end{theorem}

\begin{figure}[h]
  \alignpic
  \begin{subfigure}[b]{0.35\textwidth}
    \centering
    \input{fig/graph_k5.tikz}
    \caption{$\K{5}$}
    \label{fig:graph_k5}
  \end{subfigure}
  \begin{subfigure}[b]{0.35\textwidth}
    \centering
    \input{fig/graph_k3_3.tikz}
    \caption{$\K{3,3}$}
    \label{fig:graph_k3_3}
  \end{subfigure}
  \caption{The complete graph with five nodes $\K{5}$, and the
    complete bipartite graph with two times three nodes $\K{3,3}$.}
  \label{fig:graph_k5_and_k3_3}
\end{figure}
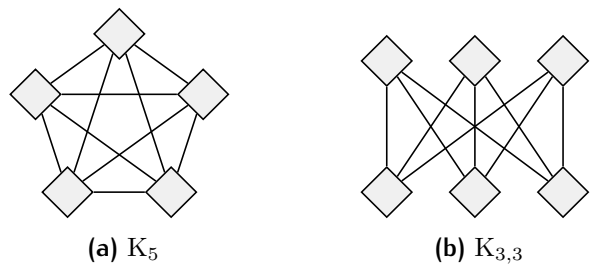

%% file: fig/graph_word_acab.tikz
\begin{tikzpicture}[lgraph,baseline=(u.text)]
  \node[lnode] (u) {$\a$};
  \node[lnode] (v) at ([shift={(0:\lnodedistIG)}]u) {$\c$};
  \node[lnode] (w) at ([shift={(0:\lnodedistIG)}]v) {$\a$};
  \node[lnode] (x) at ([shift={(0:\lnodedistIG)}]w) {$\b$};
  \path (u) edge (v)
        (v) edge (w)
        (w) edge (x);
\end{tikzpicture}

%% file: fig/graph_5_cycle_as_subgraph.tikz
\begin{tikzpicture}[lgraph,undirected]
  \node[lnode,remove] (s) {};
  \node[lnode] (t) at ([shift={(90:\lnodedistIG)}]s) {};
  \node[lnode] (u) at ([shift={(18:\lnodedistIG)}]s) {};
  \node[lnode] (v) at ([shift={(306:\lnodedistIG)}]s) {};
  \node[lnode] (w) at ([shift={(234:\lnodedistIG)}]s) {};
  \node[lnode] (x) at ([shift={(162:\lnodedistIG)}]s) {};
  \node[lnode,remove] (y) at ([shift={(0:1.7634\lnodedistIG)}]t) {};
  \node[lnode,remove] (a) at ([shift={(-90:0.9045\lnodedistIG)}]y) {};
  \node[lnode,remove] (z) at ([shift={(0:1.1756\lnodedistIG)}]v) {};
  \path (t) edge (u)
            edge (x)
        (v) edge (u)
            edge (w)
        (w) edge (x);
  \path [remove] (t) edge (s)
                     edge (y)
                 (a) edge (y)
                     edge (z)
                 (z) edge (v);
\end{tikzpicture}

%% file: fig/graph_5_cycle.tikz
\begin{tikzpicture}[lgraph,undirected]
  \coordinate (s);
  \node[lnode,contract] (t) at ([shift={(90:\lnodedistIG)}]s) {};
  \node[lnode,contract] (u) at ([shift={(18:\lnodedistIG)}]s) {};
  \node[lnode] (v) at ([shift={(306:\lnodedistIG)}]s) {};
  \node[lnode] (w) at ([shift={(234:\lnodedistIG)}]s) {};
  \node[lnode,contract] (x) at ([shift={(162:\lnodedistIG)}]s) {};
  \path (v) edge (u)
            edge (w)
        (w) edge (x);
  \path[contract] (t) edge (u)
                      edge (x);
\end{tikzpicture}

%% file: fig/graph_k3.tikz
\begin{tikzpicture}[lgraph,undirected]
  \node[lnode] (u) {};
  \node[lnode] (v) at ([shift={(-120:1.7\lnodedistIG)}]u) {};
  \node[lnode] (w) at ([shift={(-60:1.7\lnodedistIG)}]u) {};
  \path (u) edge (v)
            edge (w)
        (v) edge (w);
\end{tikzpicture}

%% file: fig/graph_k5.tikz
\begin{tikzpicture}[lgraph,undirected]
  \coordinate (center);
  \node[lnode] (u) at ([shift={(90:\lnodedistIG)}]center) {};
  \node[lnode] (v) at ([shift={(18:\lnodedistIG)}]center) {};
  \node[lnode] (w) at ([shift={(306:\lnodedistIG)}]center) {};
  \node[lnode] (x) at ([shift={(234:\lnodedistIG)}]center) {};
  \node[lnode] (y) at ([shift={(162:\lnodedistIG)}]center) {};
  \path (u) edge (v)
            edge (w)
            edge (x)
            edge (y)
        (v) edge (w)
            edge (x)
            edge (y)
        (w) edge (x)
            edge (y)
        (x) edge (y);
\end{tikzpicture}

%% file: fig/graph_k3_3.tikz
\begin{tikzpicture}[lgraph,undirected]
  \node[lnode] (u) {};
  \node[lnode] (v) at ([shift={(0:\lnodedistIG)}]u) {};
  \node[lnode] (w) at ([shift={(0:\lnodedistIG)}]v) {};
  \node[lnode] (x) at ([shift={(-90:1.5\lnodedistIG)}]u) {};
  \node[lnode] (y) at ([shift={(0:\lnodedistIG)}]x) {};
  \node[lnode] (z) at ([shift={(0:\lnodedistIG)}]y) {};
  \path (u) edge (x)
            edge (y)
            edge (z)
        (v) edge (x)
            edge (y)
            edge (z)
        (w) edge (x)
            edge (y)
            edge (z);
\end{tikzpicture}

%% file: chap/ADGAs.tex

\chapter{Alternating Distributed Graph Automata} \label{chap:adga}
In this chapter, we introduce the main variant of distributed graph
automata investigated in this work, and examine some of their
properties. In particular, we establish a game-theoretic
characterization of their acceptance condition, and derive some
closure properties of their class of recognizable languages, both of
which will be useful for proving our main result in \cref{chap:msol}.

\section{Preview} \label{sec:adga-preview}
We start with an informal description of the automaton model. Formal
definitions follow in \cref{sec:adga-definitions}.

An alternating distributed graph automaton (ADGA) is an abstract
machine that, given a labeled graph as input, can either accept or
reject it, thereby specifying a graph language. Our model of
computation incorporates the following key concepts:
\begin{description}[style=nextline]
\item[Synchronous Distributed Algorithm.] An ADGA operates primarily
  as a distributed algorithm. Each node of the input graph is assigned
  its own local processor, which we shall not distinguish from the
  node itself. Communication takes place in synchronous rounds, in
  which each node receives the current states of its incoming
  neighbors.
\item[Finite-State Machines.] Each local processor is a finite-state
  machine, i.e., an abstract machine that can be in one of a finite
  number of states, and has no additional memory. Its initial state is
  determined by the node label. After each communication round, it
  updates its state according to a nondeterministic transition
  function that depends only on the current state and the states
  received from the incoming neighborhood.
\item[Constant Running Time.] The number of communication rounds is
  limited by a constant. To ensure this, we associate a number, called
  \emph{level}, with every state. In most cases, this number indicates
  the round in which the state may occur. We require that potentially
  initial states are at level $0$, and outgoing transitions from
  states at level $i$ go to states at level $i+1$. There is an
  exception, however: the states at the highest level, called the
  \emph{permanent states}, can also be initial states, and can have
  incoming transitions from any level. Moreover, all their outgoing
  transitions are self-loops. The idea is that, once a node has
  reached a permanent state, it terminates its local computation, and
  waits for the other nodes in the graph to terminate too.
\item[Aggregation of States.] In order to be finitely representable,
  an ADGA treats collections of states as sets, i.e., it abstracts
  away from the multiplicity of states. This aggregation of states
  into sets is applied in two ways:
  \begin{itemize}
  \item First, the information received by the nodes in each round is
    a family of sets of states, indexed by the edge alphabet of the
    graph. That is, for each edge relation, a node knows which states
    occur in its incoming neighborhood, but it cannot distinguish
    between neighbors that are in the same state.
  \item Second, once all the nodes have reached a permanent state, the
    ADGA ceases to operate as a distributed algorithm, and collects
    all the reached permanent states into a set $F$. This set is the
    sole acceptance criterion: if $F$ is part of the ADGA's accepting
    sets, then the input graph is accepted, otherwise it is rejected.
  \end{itemize}
\end{description}

\begin{figure}[h!]
  \alignpic
  \input{fig/ADGA_3_colorable.tikz}
  \caption{$\dA[color]{3}$\!, an ADGA over
    $\bigl\langle\{\blank\},\{\blank\}\bigr\rangle$ whose graph
    language consists of the 3-colorable graphs.}
  \label{fig:ADGA_3_colorable}
\end{figure}

\begin{example*}[3-Colorability] \label{ex:ADGA_3_colorable}
  \Cref{fig:ADGA_3_colorable} shows a simple ADGA $\dA[color]{3}$\!,
  represented as a state diagram. The states are arranged in columns
  corresponding to their levels, ascending from left to right.
  $\dA[color]{3}$ expects a $\{\blank\}$-labeled $\{\blank\}$-graph as
  input, and accepts it \Iff it is 3-colorable. The automaton proceeds
  as follows: All nodes of the input graph are initialized to the
  state $\q{ini}$. In the first round, each node nondeterministically
  chooses to go to one of the states $q_\pik$, $q_\herz$ and
  $q_\kreuz$, which represent the three possible colors. Then, in the
  second round, the nodes verify locally that the chosen coloring is
  valid. If the set received from their incoming neighborhood (only
  one, since there is only a single edge relation) contains their own
  state, they go to $\q{no}$, otherwise to $\q{yes}$. The automaton
  then accepts the input graph \Iff all the nodes are in $\q{yes}$,
  i.e., $\{\q{yes}\}$ is its only accepting set. This is indicated by
  the blue bar to the right of the state diagram. We shall refer to
  such a representation of sets using bars as \emph{barcode}.
\end{example*}

One last key concept that enters into the definition of ADGAs is
\emph{alternation}, a generalization of nondeterminism introduced by
Chandra, Kozen and Stockmeyer in \cite{CKS81} (in their case, for
Turing machines and other types of word automata).
\begin{description}[style=nextline]
\item[Alternating Automaton.] In addition to being able to
  nondeterministically choose between different transitions, nodes can
  also explore several choices in parallel. To this end, the
  nonpermanent states of an ADGA are partitioned into two types,
  \emph{existential} and \emph{universal}, such that states on the
  same level are of the same type. If, in a given round, the nodes are
  in existential states, then they nondeterministically choose a
  single state to go to in the next round, as described above. In
  contrast, if they are in universal states, then the run of the ADGA
  is split into several parallel branches, called universal branches,
  one for each possible combination of choices of the nodes. This
  procedure of splitting is repeated recursively for each round in
  which the nodes are in universal states. The ADGA then accepts the
  input graph \Iff its acceptance condition is satisfied in every
  universal branch of the run.
\end{description}

\begin{figure}[h!]
  \alignpic
  \input{fig/ADGA_not_3_colorable.tikz}
  \caption{$\dcA[color]{3}$\!, an ADGA over
    $\bigl\langle\{\blank\},\{\blank\}\bigr\rangle$ whose graph
    language consists of the graphs that are \emph{not} 3-colorable.}
  \label{fig:ADGA_not_3_colorable}
\end{figure}

\begin{example*}[Non-3-Colorability]
  To illustrate the notion of universal branching, consider the ADGA
  $\dcA[color]{3}$ shown in \cref{fig:ADGA_not_3_colorable}. It is a
  complement automaton of $\dA[color]{3}$ from
  Example~\ref{ex:ADGA_3_colorable}, i.e., it accepts precisely those
  $\{\blank\}$-labeled $\{\blank\}$-graphs that are \emph{not}
  3-colorable. States represented as red triangles are universal
  (whereas the green squares in \cref{fig:ADGA_3_colorable} stand for
  existential states). Given an input graph with $n$ nodes,
  $\dcA[color]{3}$ proceeds as follows: All nodes are initialized to
  $\q{ini}$. In the first round, the run is split into $3^n$ universal
  branches, each of which corresponds to one possible outcome of the
  first round of $\dA[color]{3}$ running on the same input
  graph. Then, in the second round, in each of the $3^n$ universal
  branches, the nodes check whether the coloring chosen in that branch
  is valid. As indicated by the barcode, the acceptance condition of
  $\dcA[color]{3}$ is satisfied \Iff at least one node is in state
  $\q{no}$, i.e., the accepting sets are $\{\q{no}\}$ and
  $\{\q{yes},\q{no}\}$. Hence, the automaton accepts the input graph
  \Iff no valid coloring was found in any universal branch. Note that
  we could also have chosen to make the states $q_\pik$, $q_\herz$ and
  $q_\kreuz$ existential, since their outgoing transitions are
  deterministic. Regardless of their type, there is no branching in
  the second round.
\end{example*}

\section{Formal Definitions} \label{sec:adga-definitions}
We now repeat and clarify the notions from \cref{sec:adga-preview} in
a more formal setting.
\begin{definition}[Alternating Distributed Graph Automaton] \label{def:adga}
  An \defd{alternating distributed graph automaton} (ADGA) $\A$ over
  $⟨Σ,Γ⟩$ is a tuple $⟨Σ,Γ,\Q,\ab σ,δ,\F⟩$, where
  \begin{itemize}
  \item $Σ$ is a finite nonempty alphabet of node labels,
  \item $Γ$ is a finite alphabet of edge labels,
  \item $\Q=⟨Q_\EE,Q_\AA,Q_\P⟩$, where $Q_\EE$, $Q_\AA$ and $Q_\P$ are
    pairwise disjoint finite sets of \defd{existential} states,
    \defd{universal} states and \defd{permanent} states, respectively,
    with $Q_\P≠∅$, and for notational convenience we use the
    abbreviations
    \begin{itemize}
    \item $\swl{Q}{Q_\N} \coloneqq Q_\EE∪Q_\AA∪Q_\P,$\, for the entire
      set of \defd{states},\, and
    \item $Q_\N \coloneqq \swl{Q_\EE∪Q_\AA,}{Q_\EE∪Q_\AA∪Q_\P,}$\, for
      the set of \defd{nonpermanent} states,
    \end{itemize}
  \item $σ\colon Σ→Q$ is an \defd{initialization function},
  \item $δ\colon Q×(2^Q)^Γ→2^Q$ is a (local) \defd{transition
      function} that allows to unambiguously associate a \defd{level}
    $\defd{\lA(q)}∈ℕ$ with every state $q∈Q$, such that
    \begin{itemize}
    \item transitions between nonpermanent states go only from one
      level to the next, where the lowest level consists of the
      nonpermanent states that have no incoming transitions, which are
      also the only nonpermanent states that can be assigned by the
      initialization function, i.e., for every $q∈Q_\N$,
      \vspace{-1ex}
      \begin{gather*}
        \lA(q)=
        \begin{cases}
          0 & \text{\parbox[t]{.55\textwidth}{if for all $p∈Q$ and
              $\S∈(2^Q)^Γ$\!,
              it holds that $q∉δ(p,\S)$,}} \\[2.5ex]
          i+1 & \text{\parbox[t]{.55\textwidth}{if there are $p∈Q_\N$
              and $\S∈(2^Q)^Γ$ with $\lA(p)=i$ and $q∈δ(p,\S)$,}}
        \end{cases} \\[.3ex]
        ∃a∈Σ\colon σ(a)=q \quad \text{implies} \quad \lA(q)=0,
      \end{gather*}
    \item the permanent states are one level higher than the highest
      nonpermanent ones, and have only self-loops as outgoing
      transitions, i.e., for every $q∈Q_\P$,
      \vspace{-1ex}
      \begin{gather*}
        \lA(q)=
        \begin{cases}
          0                          & \text{if\; $Q_\N=∅$}, \\
          \max\{\lA(q)\mid q∈Q_\N\}+1 & \text{otherwise},
        \end{cases} \\
        δ(q,\S)=\{q\} \quad \text{for every $\S∈(2^Q)^Γ$}\!,
      \end{gather*}
    \item states on the same level are in the same component of
      $\Q$, i.e., for every level $i∈ℕ$,\: $\{q∈Q \mid
      \lA(q)=i\}∈(2^{Q_\EE}∪2^{Q_\AA}∪2^{Q_\P})$,\, and
    \end{itemize}
  \item $\F⊆2^{Q_\P}$ is a set of \defd{accepting sets} of permanent
    states.
  \end{itemize}
\end{definition}

When specifying an ADGA formally, we only need to indicate the
outgoing transitions of the nonpermanent states, since permanent
states are, by definition, self-looping. In many cases, however, we
will opt for a more convenient specification through a state diagram,
as we already did in \cref{sec:adga-preview}. Let us clarify the
details of such a representation by means of a slightly more involved
example.

\begin{example}[ADGA Specification through a State Diagram] \label{ex:state_diagram}
  Consider the ADGA $\sA{centric}=⟨Σ,Γ,\Q,σ,δ,\F⟩$ shown in
  \cref{fig:ADGA_concentric_circles}. (For now, we are not interested
  in the graph language that it recognizes. This will be discussed
  later in \cref{ex:A_centric_language}.) As indicated by the caption,
  $Σ=\{\a,\b,\c\}$, and $Γ=\{\blank\}$. Existential states are
  represented by green squares, universal states by red triangles, and
  permanent states by blue circles. The short arrows mapping node
  labels to states indicate the initialization function. For instance,
  $σ(\a)=\qa$.

  As usual, the arrows between states specify the transition
  function. A label on such a transition arrow indicates a condition
  on the states in the incoming neighborhood of a node, which must be
  satisfied in order for the node to be allowed to take that
  transition. If there is no label, any states in the neighborhood
  are permitted. In this example (and any other example that we shall
  consider) the permitted input graphs have a single edge relation,
  which means that the occurrences of states the incoming
  neighborhood of a node are abstracted as a single set of states
  $S$. Now, the labels on the transition arrows are formulas that
  specify conditions on such a set $S$. In those formulas, $S$ remains
  anonymous, and binary relations like $∈$ and $=$ are written as if
  they were unary, but implicitly refer to $S$. For example,
  “$\xni\qb$” means “$\qb∈S$”, and “$\xeq\{\qbkr,\qbk\}$” means
  “$S=\{\qbkr,\qbk\}$”. Consequently, considering the three outgoing
  arrows of $\qb$ tells us, for instance, that
  $δ\bigl(\qb,⟨\{\qa,\qc\}⟩\bigl)=\{\qbkr,\qbk\}$ and
  $δ\bigl(\qb,⟨\{\qa,\qb\}⟩\bigl)=\{\q{no}\}$. Furthermore, we can
  build up more complex formulas using the usual boolean connectives
  $¬$, $∨$, $∧$, etc. Hence, “$\xnni \qc ∧ \xnni \qa$” characterizes
  all the sets of states that contain neither $\qc$ nor $\qa$. We
  shall refer to such formulas as \emph{set formulas}.

  Finally, as already mentioned in \cref{sec:adga-preview}, the
  \emph{barcode} on the far right specifies the accepting sets. The
  blue bars are aligned with the permanent states to which they
  correspond. Each column represents an accepting set, where a bar
  means that the corresponding permanent state is included in the
  set. Thus, $\F=\bigl\{\{\qap,\q{yes}\},\{\qah,\q{yes}\}\bigr\}$.
\end{example}

\begin{figure}
  \alignpic
  \input{fig/ADGA_concentric_circles.tikz}
  \caption{$\sA{centric}$, an ADGA over
    $\bigl\langle\{\a,\b,\c\},\{\blank\}\bigr\rangle$ whose graph
    language consists of the labeled graphs that satisfy the following
    conditions: the labeling constitutes a valid 3-coloring, there is
    precisely one $\a$-labeled node $v_\a$, the undirected
    neighborhood of $v_\a$ contains only $\b$-labeled nodes, and
    $v_\a$ has at least two incoming neighbors.}
  \label{fig:ADGA_concentric_circles}
\end{figure}
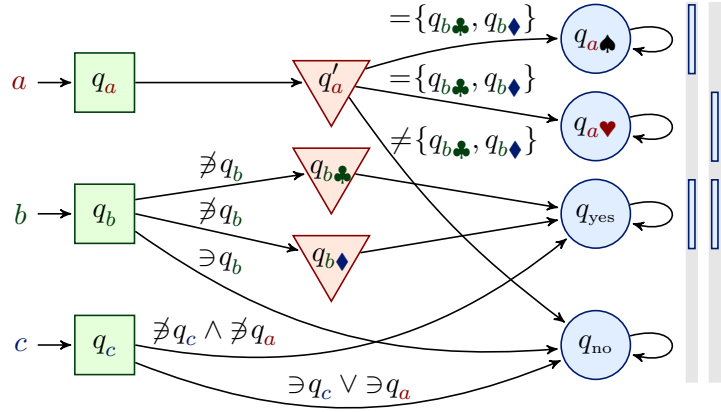

For any ADGA $\A=⟨Σ,Γ,\ab\Q,\ab σ,\ab δ,\F⟩$, we define its
\defd{size} $\siz(\A)$ to be its number of states, i.e.,
$\defd{\siz(\A)}\coloneqq\card{Q}$, and its \defd{length} $\len(\A)$
to be its highest level, i.e., $\defd{\len(\A)} \coloneqq
\max\{\lA(q)\mid q∈Q\}$. For example, we get $\siz(\sA{centric})=10$
and $\len(\sA{centric})=2$, for the automaton from
\cref{fig:ADGA_concentric_circles}.

As a further abbreviation, we shall use $\defd{\lev_i(\A)} \coloneqq
\{q∈Q \mid \lA(q)=i\}$, for the set of states at level $i$, with
$0≤i≤\len(\A)$. For instance,
$\lev_1(\sA{centric})=\{\qaprime,\qbkr,\qbk\}$. We say that level $i$
of $\A$ is existential if \,$\lev_i(\A)⊆Q_\EE$, and analogously for
universal and permanent levels. For $\sA{centric}$, level $0$ is
existential, level $1$ is universal, and level $2$ is permanent.

Next, we want to give a formal definition of a run. For this, we need
the notion of a configuration, which can be seen as the global state
of an ADGA.

\begin{definition}[Configuration]
  Consider an ADGA $\A=⟨Σ,Γ,\ab\Q,\ab σ,\ab δ,\F⟩$. We call any
  $Q$-labeled $Γ$-graph $G_κ∈Q^{\clouded{Γ}}$ a \defd{configuration}
  of $\A$ on $G$. If every node is labeled by a permanent state, i.e.,
  if $G_κ∈(Q_\P)^{\clouded{Γ}}$, we call $G_κ$ a \defd{permanent}
  configuration. Otherwise, if $G_κ$ is a nonpermanent configuration
  whose nodes are labeled exclusively by existential and permanent
  states, i.e., if $G_κ∈⟨Q_\EE,Q_\P⟩^{\clouded{Γ}}$, we say that $G_κ$
  is an \defd{existential} configuration. Analogously, if
  $G_κ∈⟨Q_\AA,Q_\P⟩^{\clouded{Γ}}$, the configuration is called
  \defd{universal}.

  Additionally, we say that a permanent configuration $G_κ$ is
  \defd{accepting} if the set of states occurring in it is accepting,
  i.e., if $\{κ(v) \mid v∈\VG\}∈\F$. Any other permanent configuration
  is called \defd{rejecting}. Nonpermanent configurations are neither
  accepting nor rejecting.
\end{definition}

The (local) transition function of an ADGA specifies for each state a
set of potential successors, for a given family of sets of
states. This can be naturally extended to configurations, which leads
us to the definition of a global transition function.

\begin{definition}[Global Transition Function]
  The \defd{global transition function} $δ^\cloud$ of an ADGA
  $\A=⟨Σ,Γ,\ab\Q,\ab σ,\ab δ,\F⟩$ assigns to each configuration $G_κ$
  of $\A$ the set of all of its \defd{successor configurations} $G_μ$,
  by combining all possible outcomes of local transitions on $G_κ$,
  i.e.,
  \begin{align*}
    δ^\cloud \colon Q^{\clouded{Γ}} &→ 2^{(Q^{\clouded{Γ}})} \\
    G_κ &↦ \biggl\{G_μ \biggm|
    \bigwedge_{v∈\VG} μ(v)∈ δ\Bigl(κ(v),\,\bigl\langle\{κ(u)\mid u\arrG{γ} v\}\bigr\rangle_{γ∈Γ}\Bigr)\biggr\}.
  \end{align*}
\end{definition}

A configuration $G_μ$ that can be obtained from $G_κ$ by iteratively
choosing some successor configuration shall be referred to as a
\defd{descendant configuration} of $G_κ$.

We now have everything at hand to formalize the notion of a run. As
mentioned in \cref{sec:adga-preview}, a run can be split into several
parallel branches whenever the nodes of the input graph are in
universal states. It thus may seem natural to define a run as a tree
whose nodes are labeled by configurations of the automaton. We could
then interpret the branches of such a tree as “non-communicating
parallel universes”. However, since an ADGA has no “memory of the
past” other than its current configuration, there is no need to keep
apart branches that are in the same configuration in a given round. By
merging such branches, we obtain a directed acyclic graph in which
every node is labeled by a \emph{unique} configuration (since a
configuration cannot occur in more than one round). This has the
advantage that we can identify the nodes of a run with the
configurations of an automaton, which will make it easier to refer to
particular nodes and paths of a run in subsequent proofs.

\begin{definition}[Run]
  A \defd{run} of an ADGA $\A=⟨Σ,Γ,\Q,σ,δ,\F⟩$ on a labeled graph
  $G_λ∈Σ^{\clouded{Γ}}$ is a directed acyclic graph $R=⟨K,\arr⟩$ whose
  nodes are configurations of $\A$ on $G$, i.e., $K⊆Q^G$, such that
  \begin{itemize}
  \item the \defd{initial configuration} $G_{σ∘λ}∈K$ is the only
    source,\tablefootnote{Here, the operator $∘$ denotes function
      composition, such that $(σ∘λ)(v)=σ(λ(v))$.}
  \item every existential configuration $G_κ∈(K∩⟨Q_\EE,Q_\P⟩^G)$ has
    exactly one outgoing neighbor $G_μ∈δ^\cloud(G_κ)$,
  \item every universal configuration $G_κ∈(K∩⟨Q_\AA,Q_\P⟩^G)$ with
    $δ^\cloud(G_κ)=\{G_{μ_1},…,G_{μ_m}\}$ has exactly $m$ outgoing
    neighbors $G_{μ_1},…,G_{μ_m}$,\, and
  \item every permanent configuration $G_κ∈(K∩Q_\P^G)$ is a sink.
  \end{itemize}
  Such a run $R$ is \defd{accepting} if every occurring permanent
  configuration is accepting, i.e., if $K$ contains a node
  $G_κ∈Q_\P^G$, then $\{κ(v)\mid v∈\VG\}∈\F$. Otherwise, $R$ is called
  \defd{rejecting}.
\end{definition}

\begin{example}[Two Runs of $\sA{centric}$] \label{ex:A_centric_runs}
  We take up \cref{ex:state_diagram}, and consider again the ADGA
  $\sA{centric}$ from \mbox{\cref{fig:ADGA_concentric_circles}}.

  The graph in \cref{fig:run_accepting} is a run of $\sA{centric}$ on
  the labeled graph $(G_1)_{λ_1}$ shown in
  \cref{fig:graph_labeled_pentagon}. (The figures are also depicted
  together on the \hyperlink{page.a}{cover} of this thesis.) We have
  adopted the same coloring scheme as for (automaton) states, i.e., a
  green configuration is existential, a red one is universal, and a
  blue one is permanent. In the first round, the three nodes that are
  in state $\qb$ have a nondeterministic choice between $\qbkr$ and
  $\qbk$. Hence, the second configuration is one of eight possible
  choices. The branching in the second round is due to the node in
  state $\qaprime$ which goes simultaneously to $\qap$ and $\qah$. In
  both branches, an accepting configuration is reached, since
  $\{\qap,\q{yes}\}$ and $\{\qah,\q{yes}\}$ are both accepting
  sets. This is also visually indicated by the double circles around
  the configurations. We conclude that the run is accepting.

  As an example of a rejecting run, consider \cref{fig:run_rejecting}
  which shows a run of the same automaton $\sA{centric}$ on the
  labeled graph $(G_2)_{λ_2}$ from
  \cref{fig:graph_labeled_square}. Again, the configuration chosen
  during the first round is one of several (four) possibilities. In
  the second round, the run is split into four universal branches,
  corresponding to the four possible combinations of choices of the
  two nodes that are in state $\qaprime$. The permanent configurations
  reached in the two middle branches are rejecting because
  $\{\qap,\qah,\q{yes}\}$ is not an accepting set of $\sA{centric}$.
  The occurrence of these rejecting configurations causes the entire
  run to be rejecting.
\end{example}

\begin{figure}[h!]
  \alignpic
  \begin{subfigure}[b]{0.35\textwidth}
    \centering
    \input{fig/graph_labeled_pentagon.tikz}
    \caption{$(G_1)_{λ_1}$}
    \label{fig:graph_labeled_pentagon}
  \end{subfigure}
  \begin{subfigure}[b]{0.3\textwidth}
    \centering
    \input{fig/graph_labeled_square.tikz}
    \vspace{2.1ex}
    \caption{$(G_2)_{λ_2}$}
    \label{fig:graph_labeled_square}
  \end{subfigure}
  \caption{Two $\{\a,\b,\c\}$-labeled $\{\blank\}$-graphs.}
\end{figure}

\begin{figure}[p]
  \alignpic
  \input{fig/run_accepting.tikz}
  \caption{An accepting run of the ADGA $\sA{centric}$ from
    \cref{fig:ADGA_concentric_circles} on the labeled graph
    $(G_1)_{λ_1}$ from \cref{fig:graph_labeled_pentagon}.}
  \label{fig:run_accepting}
\end{figure}

\begin{figure}[p]
  \alignpic
  \input{fig/run_rejecting.tikz}
  \caption{A rejecting run of the ADGA $\sA{centric}$ from
    \cref{fig:ADGA_concentric_circles} on the labeled graph
    $(G_2)_{λ_2}$ from \cref{fig:graph_labeled_square}.}
  \label{fig:run_rejecting}
\end{figure}

For any ADGA $\A=⟨Σ,Γ,\Q,σ,δ,\F⟩$ and labeled graph
$G_λ∈Σ^{\clouded{Γ}}$, a configuration $G_κ∈Q^{\clouded{Γ}}$ is called
\defd{reachable} by $\A$ on $G_λ$ if either $G_κ=G_{σ∘λ}$ or there is
a configuration $G_μ∈Q^{\clouded{Γ}}$ reachable by $\A$ on $G_λ$ such
that $G_κ∈δ^\cloud(G_μ)$. If $G_λ$ is irrelevant, we simply say that
$G_κ$ is reachable by $\A$. Note that only existential, universal and
permanent configurations can satisfy this property, i.e., “mixed”
configurations with both existential and universal states are never
reachable. Furthermore, any reachable existential configuration
$G_κ∈⟨Q_\EE,Q_\P⟩^{\clouded{Γ}}$ has existential states of uniform
level, i.e., there is a level $i$, such that $\lA(κ(v))=i$ for all
nodes $v∈\VG$ with $κ(v)∈Q_\EE$. The analogous observation holds if
$G_κ$ is a reachable universal configuration.

In the following definition, we transfer the usual terminology of
automata theory to ADGAs.

\begin{definition}[ADGA-Recognizability]
  Let $\A=⟨Σ,Γ,\Q,σ,δ,\F⟩$ be an ADGA. A labeled graph
  $G_λ∈Σ^{\clouded{Γ}}$ is \defd{accepted} by $\A$ \Iff there exists
  an accepting run $R$ of $\A$ on $G_λ$. The graph language
  \defd{recognized} by $\A$ is the set
  \begin{equation*}
    \defd{\L(\A)} \coloneqq \bigl\{ G_λ∈Σ^{\clouded{Γ}} \bigm| \text{$\A$ accepts $G_λ$} \bigr\}.
  \end{equation*}
  Every graph language that is recognized by some ADGA is called
  \defd{ADGA-recognizable}. We denote by \defd{$\LL_\ADGA$} the class
  of all such graph languages.
\end{definition}

If two ADGAs recognize the same graph language, we say that they are
\defd{equivalent}.

\begin{example}[Language Recognized by $\sA{centric}$] \label{ex:A_centric_language}
  We get back to the example automaton $\sA{centric}$ from
  \cref{fig:ADGA_concentric_circles}, this time turning our attention
  to the graph language that it recognizes.

  In the first round, the $\a$-labeled nodes do nothing but update
  their state, while the $\b$- and $\c$-labeled nodes verify that the
  graph coloring is valid from their point of view. The $\c$-labeled
  nodes additionally check that they do not see any $\a$'s, and then
  directly terminate. Meanwhile, the $\b$-labeled nodes
  nondeterministically choose one of the markers $\kreuz$ and
  $\karo$. In the second round, only the $\a$-labeled nodes are
  busy. They verify that their incoming neighborhood consists
  exclusively of $\b$-labeled nodes, and that both of the markers
  $\kreuz$ and $\karo$ occur, thus ensuring that they have at least
  two incoming neighbors. Then they simultaneously pick the markers
  $\pik$ and $\herz$, thereby creating different universal branches,
  and the run of the automaton terminates. Finally, the ADGA checks
  that all the nodes approve of the graph (meaning that none of them
  has reached state $\q{no}$), and that in each universal branch,
  precisely one of the markers $\pik$ and $\herz$ occurs, which
  implies that there is a unique $\a$-labeled node.

  To sum up, the graph language $\L(\sA{centric})$ consists of all the
  $\{\a,\b,\c\}$-labeled $\{\blank\}$-graphs such that
  \begin{itemize}
    \item the labeling constitutes a valid 3-coloring,
    \item there is precisely one $\a$-labeled node $v_\a$, and
    \item $v_\a$ has only $\b$-labeled nodes in its undirected
      neighborhood, and at least two incoming neighbors.
  \end{itemize}
  The name “$\sA{centric}$” refers to the fact that, in the (weakly)
  connected component of $v_\a$, the $\b$- and $\c$-labeled nodes form
  con\emph{centric} circles around $v_\a$, i.e., nodes at distance~1
  of $v_\a$ are labeled with $\b$, nodes at distance~2 (if existent)
  with $\c$, nodes at distance~3 (if existent) with $\b$, and so
  forth.

  An example of a labeled graph that lies in $\L(\sA{centric})$ is the
  graph $(G_1)_{λ_1}$ from \cref{fig:graph_labeled_pentagon}, for
  which we have seen an accepting run in \cref{fig:run_accepting}. On
  the other hand, the labeled graph $(G_2)_{λ_2}$ from
  \cref{fig:graph_labeled_square} is not an element of
  $\L(\sA{centric})$, since it contains two $\a$-labeled nodes.
  Either this fact is detected through the universal branching in the
  second round (as in the run in \cref{fig:run_rejecting}), or the two
  $\b$-labeled nodes fail to choose two different markers in the first
  round, leading to refusal by the $\a$-labeled nodes. In any case,
  the resulting run is rejecting.
\end{example}

\section{Further Examples}
The automaton $\sA{centric}$ that accompanied us through
\cref{sec:adga-definitions} has been useful for illustrating several
features of ADGAs in a reasonably small example, but its recognized
graph language might seem a bit artificial. In this section, we focus
on more natural graph properties like being connected, containing a
cycle, or being planar, in order to give the reader a better idea of
what is possible with ADGAs, and how it can be accomplished.

\subsection*{Connected Graphs}
The fact that ADGAs aggregate the states reached in the last round of
a run gives them the ability to recognize a particular graph property
that a purely distributed algorithm, by its very nature, could not
perceive: the property of being connected. For any node alphabet $Σ$
and edge alphabet $Γ$, we can construct an ADGA $\A_\conn(Σ,Γ)$ that
recognizes the language of all (weakly) connected $Σ$-labeled
$Γ$-graphs. The following example shows this for $Σ=Γ=\{\blank\}$, but
other cases are completely analogous.

\begin{example}[Recognizing Connected Graphs] \label{ex:ADGA_weakly_connected}
  The ADGA $\A_\conn\bigl(\{\blank\},\{\blank\}\bigr)$ is specified in
  \cref{fig:ADGA_weakly_connected}. It proceeds as follows: In the
  first round, the nodes simultaneously pick the markers $\pik$ and
  $\herz$, creating a universal branch for each combination of
  choices. Then, in the second round, they check in each branch
  whether all their incoming neighbors have chosen the same marker as
  themselves. If this is the case, they simply retain their
  marker. Otherwise, they signalize a discordance in the affected
  branch by switching to $\q{acc}$. Afterwards, the automaton accepts
  the input graph if and only if, in each universal branch, either all
  the nodes have chosen the same marker, or a discordance has been
  signalized.

  If the input graph is (weakly) connected, and not all nodes have
  chosen the same marker, then there are always two adjacent nodes with
  different markers, and at least one of them will signalize a
  discordance. However, if the graph is not connected, then nodes in
  different connected components can choose different markers, without
  any of them detecting it, and hence some branches of the run will
  reach rejecting configurations (corresponding to
  $\{\qpikprime,\qherzprime\}$).
\end{example}

\begin{figure}
  \alignpic
  \input{fig/ADGA_weakly_connected.tikz}
  \caption{$\A_\conn\bigl(\{\blank\},\{\blank\}\bigr)$, an ADGA over
    $\bigl\langle\{\blank\},\{\blank\}\bigr\rangle$ whose graph
    language consists of the (weakly) connected graphs.}
  \label{fig:ADGA_weakly_connected}
\end{figure}

In some cases, it is convenient to restrict the allowed input graphs
of an automaton to (weakly) connected graphs. To this end, we define
the \defd{connected graph language} of an ADGA $\A=⟨Σ,Γ,\Q,σ,δ,\F⟩$ as
\begin{equation*}
  \defd{\L_\conn(\A)}\coloneqq\L(\A)∩\L\bigl(\A_\conn(Σ,Γ)\bigr).
\end{equation*}
As we will see in \cref{sec:closure-properties}, ADGA-recognizable
graph languages are effectively closed under intersection, thus we can
always replace $\A$ by an intersection automaton of $\A$ and
$\A_\conn(Σ,Γ)$. Therefore, the definition above is only a
convenience.

\subsection*{Directed Trees}
One situation in which it is helpful to require connected input graphs
is when recognizing the language of directed trees. By a
\defd{directed tree} we mean a connected simple graph, such that there
is a unique source, called the \emph{root}, and any other node has
exactly one incoming neighbor, called its \emph{parent}.

\begin{example}[Recognizing Directed Trees]
  The connected graph language of the ADGA $\sA{tree}$ specified in
  \cref{fig:ADGA_directed_tree} is precisely the set of all directed
  trees. The automaton follows quite naturally from the above
  definition of a directed tree. First, each node checks whether it
  has any incoming neighbors. If not, then it expects to be the
  unique root (indicated by state $q^\notlarr$), otherwise it must
  verify that it has only one parent (indicated by state $q^←$). In
  any case, in the second round, each node simultaneously picks the
  markers $\pik$ and $\herz$, which causes a universal branching of
  the run. Any potential root node terminates at that point. The other
  nodes continue for another round, and verify in each branch that
  they see precisely one state in their incoming neighborhood (the
  symbol “$\#$” in the state diagram refers to the cardinality of the
  received set). If a node has more than one incoming neighbor, it
  will see several states in some of the branches, and signalize this
  error by going to $\q{no}$. Otherwise it goes to $\q{yes}$. Finally,
  the ADGA checks that the last configuration in every branch is
  error-free and contains only one of the markers $\pik$ and $\herz$,
  thereby ensuring that there is precisely one node that claims to be
  the root.

  Note that the entire approach relies heavily on the requirement that
  the input graph be connected: without this condition, there could be
  a directed cycle disconnected from the root.
\end{example}

\begin{figure}
  \alignpic
  \input{fig/ADGA_directed_tree.tikz}
  \caption{$\sA{tree}$, an ADGA over
    $\bigl\langle\{\blank\},\{\blank\}\bigr\rangle$ whose
    \emph{connected} graph language $\L_\conn(\sA{tree})$ consists of
    the rooted directed trees with all edges directed away from the
    root.}
  \label{fig:ADGA_directed_tree}
\end{figure}

\subsection*{Undirected Graphs}
Up to now, our examples have only involved graphs with directed edges.
But, as mentioned in \cref{sec:graph-properties}, we can represent
undirected graphs as directed graphs with bidirectional edges. The
next example shows an ADGA that checks whether a $\{\blank\}$-labeled
$\{\blank\}$-graph is undirected. Again, the principle can be
generalized to any node alphabet $Σ$ and edge alphabet $Γ$, thus
giving us an entire family of ADGAs. We denote by $\A_\undir(Σ,Γ)$ the
automaton that recognizes the language of all undirected $Σ$-labeled
$Γ$-graphs.

\newpage

\begin{example}[Recognizing Undirected Graphs]
  The ADGA $\A_\undir\bigl(\{\blank\},\{\blank\}\bigr)$ is specified
  in \cref{fig:ADGA_undirected}. It uses a universal branching in the
  first round, where each node can either send a message to all of its
  outgoing neighbors (state $\qs$) or remain silent (state
  $\qns$). In the second round, the silent nodes check whether or not
  they have received any message, which they indicate by going to
  $\qr$ or $\qnr$, respectively. In the last round, the nodes that
  have sent a message verify that none of their incoming neighbors
  report not to have received any message. The automaton then accepts
  the input graph \Iff every test turns out positive in each universal
  branch.
\end{example}

Analogously to the connected graph language, we define the
\defd{undirected graph language} of an ADGA $\A=⟨Σ,Γ,\Q,σ,δ,\F⟩$ as
\begin{equation*}
  \defd{\L_\undir(\A)}\coloneqq\L(\A)∩\L\bigl(\A_\undir(Σ,Γ)\bigr).
\end{equation*}

\begin{figure}
  \alignpic
  \input{fig/ADGA_undirected.tikz}
  \caption{$\A_\undir\bigl(\{\blank\},\{\blank\}\bigr)$, an ADGA over
    $\bigl\langle\{\blank\},\{\blank\}\bigr\rangle$ whose graph
    language consists of the undirected graphs.}
  \label{fig:ADGA_undirected}
\end{figure}

\subsection*{Graph Minors}
To finish our series of examples, we show how ADGAs can check their
input graphs for particular minors, and then make use of this to
recognize the graphs that contain a cycle and the planar graphs. The
approach is heavily inspired by the book \cite{CE12}, where it is
shown how planarity can be expressed in monadic second-order logic. In
this context, we only consider unlabeled, simple, undirected graphs.

For any given loop-free graph $H$, we can construct an ADGA
$\dA[minor]{H}$, such that, for every graph $G$, the graph language
$\L_\undir(\dA[minor]{H})$ contains $G$ \Iff $H$ is a minor of
$G$. Our construction follows from the characterization of minor
inclusion given in \cref{lem:minor-inclusion}. The automaton proceeds
as follows: First, it nondeterministically partitions some subset of
$\VG$ into sets $U_1,…,U_n$, corresponding to the $n$ nodes of
$H$. Then, it checks that each induced subgraph $G[U_i]$ is connected,
and that for each edge in $H$ between two nodes $v_i$ and $v_j$, there
is an edge in $G$ connecting the corresponding subgraphs $G[U_i]$ and
$G[U_j]$. In order to verify that a subgraph is connected, we use a
slightly adapted version of the automaton
$\A_\conn\bigl(\{\blank\},\{\blank\}\bigr)$ from
\cref{fig:ADGA_weakly_connected} as a building block.

The following example shows the construction for the complete graph
$\K{3}$, which by \cref{rem:k3-cycle} gives us an ADGA for the
language of all simple undirected graphs that contain at least one
cycle of three or more nodes.

\begin{figure}[p]
  \vspace{-2ex}
  \alignpic
  \input{fig/ADGA_K3_minor.tikz}
  \caption{$\dA[minor]{\K{3}}$\!, an ADGA over
    $\bigl\langle\{\blank\},\{\blank\}\bigr\rangle$ whose
    \emph{undirected} graph language $\L_\undir(\dA[minor]{\K{3}})$
    consists of the graphs that contain $\K{3}$ as a minor, or in
    other words, the graphs that contain at least one cycle of three
    or more nodes.}
  \label{fig:ADGA_K3_minor}
\end{figure}

\begin{example}[Recognizing Graphs with a Cycle]
  The ADGA $\dA[minor]{\K{3}}$ is specified in
  \cref{fig:ADGA_K3_minor}. Since it has too many accepting sets to
  represent them with a barcode, its acceptance condition is given as
  a set formula. Also, for the sake of better readability, there are
  three occurrences of $\q{acc}$ in the state diagram, but they all
  represent the same state.

  Let $U_1$, $U_2$ and $U_3$ be three sets corresponding to the three
  nodes of $\K{3}$. In the first round, each node of the input graph
  nondeterministically decides whether to participate and join one of
  those sets (states $q_1$, $q_2$, $q_3$), or not to interfere at all
  and terminate right away (state $\q{out}$). Then, in the second
  round, there is a universal branching in which each participating
  node simultaneously picks the markers $\pik$ and $\herz$. In the
  third and last round, in every universal branch, each participating
  node checks whether it has any neighbor that is in the same set as
  itself but has chosen a different marker. If this is the case, it
  goes to $\q{acc}$ to signalize that the affected branch of the run
  contains a discordance and must be treated accordingly by the
  automaton. (This is analogous to the behaviour of the ADGA from
  \cref{ex:ADGA_weakly_connected}.) Otherwise, the node assumes that
  the subgraph $G[U_i]$ to which it belongs is connected, and it
  checks whether any of its neighbors is part of the subgraph
  $G[U_{(i\bmod3)+1}]$. If so, it switches to a state with a
  superscript “$←$”, otherwise to a state with a
  “$\notlarr$”. Finally, the automaton decides on acceptance as
  follows: Branches in which a discordance has been signalized (by the
  state $\q{acc}$) are inconclusive, and thus the configurations
  reached in such branches are considered to be accepting. In all the
  other branches, the automaton expects that in each subgraph
  $G[U_i]$, the nodes agree on one of the markers $\pik$ and $\herz$,
  and at least one of them signalizes that it is connected to a node
  in the “next” subgraph $G[U_{(i\bmod3)+1}]$.
\end{example}

We can proceed similarly to construct the automaton $\dA[minor]{H}$
for any other loop-free graph $H$. Levels~0, 1 and 2 are completely
analogous and depend only on the number of nodes of $H$. Level~4, on
the other hand, must be adapted to the structure of the graph. For
each edge in $H$ between two nodes $v_i$ and $v_j$, it must be
verified that there is a corresponding edge in the input graph $G$,
connecting $G[U_i]$ and $G[U_j]$. Either the nodes in $U_i$ or the
nodes in $U_j$ must perform this verification. However, the number of
permanent states required by the nodes in each set $U_i$ grows
exponentially with the number of edges for which those nodes are
responsible. This is because a single node in $U_i$ might have
neighbors in several other sets, and each combination must be encoded
in a separate state. Hence, if we want to keep the total number of
states low, we have to balance the load among the sets as evenly as
possible. The exact specification of $\dA[minor]{H}$ can thus be
optimized for each graph $H$, but the construction principle is always
the same as for $\dA[minor]{\K{3}}$.

The possibility to check for arbitrary minors is a powerful tool. As
another application example, we outline how to recognize the language
of planar graphs.

\begin{example}[Recognizing Planar Graphs]
  By the Kuratowski-Wagner Theorem (\cref{thm:kuratowski-wagner}), a
  graph is planar \Iff it contains neither the complete graph $\K{5}$
  nor the complete bipartite graph $\K{3,3}$ as a minor. Moreover, as
  we will see in \cref{sec:closure-properties}, ADGA-recognizable
  graph languages are effectively closed under boolean set
  operations. Thus, by constructing the union automaton of
  $\dA[minor]{\K{5}}$ and $\dA[minor]{\K{3,3}}$ and then complementing
  it, we obtain an ADGA $\sA{planar}$ over
  $\bigl\langle\{\blank\},\{\blank\}\bigr\rangle$ whose undirected
  graph language is precisely the set of all simple undirected planar
  graphs.
\end{example}

\section{Normal Forms} \label{sec:normal-forms}
In this section, we establish some normal forms of ADGAs, which will
prove helpful for the closure constructions in
\cref{sec:closure-properties}.

The notion of a nonblocking ADGA is analogous to that of a nonblocking
finite automaton on words: it guarantees that the automaton cannot
“get stuck” during execution, which for an ADGA means that all of its
runs eventually reach a permanent configuration in each universal
branch.

\begin{definition}[Nonblocking ADGA]
  An ADGA $\A=⟨Σ,Γ,\Q,σ,δ,\F⟩$ is called \defd{nonblocking} \Iff every
  configuration $G_κ∈Q^{\clouded{Γ}}$ that is reachable by $\A$ has at
  least one successor configuration, i.e., $δ^\cloud(G_κ)≠∅$.
\end{definition}

A sufficient (but not necessary) condition for $\A$ to be nonblocking
is that its transition function is complete, i.e., $δ(q,\S)≠∅$ for
every $q∈Q$ and $\S∈(2^Q)^Γ$\!. This gives us an effective way of
transforming any given ADGA into an equivalent nonblocking one.

\begin{remark} \label{rem:nonblocking}
  For every ADGA $\A$, we can effectively construct an equivalent ADGA
  $\A'$ that is nonblocking. Moreover, $\siz(\A')≤\siz(\A)+\len(\A)$\,
  and \,$\len(\A')=\len(\A)$.
\end{remark}

\begin{proof}
  \newcommand{\qstop}[1]{q_{\,#1}^{\scriptscriptstyle\textnormal{stop}}}
  \newcommand{\Qstop}{Q_{\scriptscriptstyle\textnormal{stop}}}
  Let $\A=⟨Σ,Γ,\Q,σ,δ,\F⟩$. We extend $\A$ such that its transition
  function becomes complete, giving us an equivalent ADGA that is
  guaranteed to be nonblocking. To this end, we introduce an
  additional permanent state $\qstop{i}$ for every nonpermanent level
  $i$. If a node was blocked at level $i$, it now simply moves to
  state $\qstop{i}$, and waits there for the other nodes to
  terminate. A permanent configuration is then accepting if it already
  was so previously, or if it contains states indicating that the
  lowest level at which some node would have been blocked, is
  universal. Formally, we fix the set $\Qstop=\{\qstop{i}\mid
  0≤i<\len(\A)\}$, and construct $\A'=⟨Σ,Γ,\Q',σ,δ',\F'⟩$, with
  \begin{itemize}[topsep=1ex,itemsep=0ex]
  \item $Q'_\EE=Q_\EE$, \quad $Q'_\AA=Q_\AA$, \quad
    $Q'_\P=Q_\P∪\Qstop$,
   \item
    $δ'(q,\S) =
    \begin{cases}
      \,δ(q,\S) & \text{if $\S∈(2^Q)^Γ$ and $δ(q,\S)≠∅$}, \\
      \bigl\{\qstop{\lA(q)}\bigr\} & \text{otherwise},
    \end{cases}$ \\
    for every $q∈Q'_\N$ and $\S∈(2^{Q'})^Γ$\!,\, and 
  \item $\F' = \F ∪ \bigl\{ F⊆Q'_P \bigm| \text{$\min\{i\mid \qstop{i}∈F\}$ is a universal level of $\A$} \bigr\}$.
  \end{itemize}
  As a slight optimization, if $Q_\EE=∅$ or $Q_\AA=∅$, the states in
  $\Qstop$ can be merged into a single state $\q{stop}$.
\end{proof}

Next, again in analogy to finite automata on words, we say that an
ADGA is trim if it does not have any states that are obviously
useless. In the case of (nondeterministic) finite automata on words,
a state is considered useless if it is not reachable from any initial
state or if no accepting state is reachable from it. However, for
ADGAs the notion of reachability of states is more involved, since it
is subject to the reachability of configurations. For our purposes, it
will be enough to consider a necessary condition for reachability,
which can be easily checked for every state: reachability within the
state diagram, ignoring the requirements on the transition
arrows. States that do not satisfy this condition are obviously
useless, since they cannot occur in any run. We can thus safely remove
them from the automaton, without affecting the graph language it
recognizes.

\begin{definition}[Trim ADGA]
  Let $\A=⟨Σ,Γ,\Q,σ,δ,\F⟩$ be an ADGA. We consider a state $q∈Q$
  to be \defd{potentially reachable} if
  \begin{itemize}[topsep=1ex,itemsep=0ex]
    \item $q=σ(a)$ for some $a∈Σ$, or
    \item $q∈δ(p,\S)$ for some $p∈Q$ and $\S=⟨S_γ⟩_{γ∈Γ}∈(2^Q)^Γ$\!,
      such that $p$ and every state $p'∈\bigcup_{γ∈Γ}S_γ$ are
      potentially reachable.
  \end{itemize}
  The automaton $\A$ is said to be \defd{trim} if all of its states
  are potentially reachable.
\end{definition}

\begin{remark} \label{rem:trim}
  For every ADGA $\A$, we obtain an equivalent ADGA $\A'$ that is
  trim, by removing all states from $\A$ that are not potentially
  reachable (and adapting the transition function and acceptance
  condition accordingly). Moreover, if $\A$ is nonblocking, then so is
  $\A'$.
\end{remark}

Last, we introduce the notion of alternating normal form, which
requires that successive nonpermanent levels of an ADGA are
alternately existential and universal. Every ADGA can be transformed
into an equivalent automaton in alternating normal form by inserting
“dummy” levels between any two consecutive levels that are of the same
type.

\begin{definition}[Alternating Normal Form]
  An ADGA $\A=⟨Σ,Γ,\ab \Q,\ab σ,\ab δ,\F⟩$ is in \defd{alternating
    normal form} if for every level $i∈\{0,…,\len(\A)-2\}$,
  \begin{gather*}
    \lev_i(\A)⊆Q_\mEE \quad \text{implies} \quad \lev_{i+1}(\A)⊆Q_\AA, \quad \text{and} \\
    \lev_i(\A)⊆Q_\mAA \quad \text{implies} \quad \lev_{i+1}(\A)⊆Q_\EE.
  \end{gather*}
\end{definition}

\begin{remark} \label{rem:ANF}
  For every ADGA $\A$, we can effectively construct an equivalent ADGA
  $\A'$ that is in alternating normal form. If $\A$ is nonblocking or
  trim, then these properties carry over to $\A'$. Moreover,
  $\siz(\A')<2\siz(\A)$\, and \,$\len(\A')<2\len(\A)$.
\end{remark}

\begin{proof}
  Let $\A=⟨Σ,Γ,\Q,σ,δ,\F⟩$.  We construct $\A'=⟨Σ,Γ,\Q',σ,δ',\F⟩$ as
  an extended version of $\A$, with $Q_\EE⊆Q'_\EE$,\, $Q_\AA⊆Q'_\AA$,
  and $Q_\P=Q'_\P$.

  Suppose that there is some level $i$, with $0≤i≤\len(\A)-2$, such
  that both $i$ and $i+1$ are existential levels in $\A$. We remedy
  this in $\A'$ by inserting a disjoint copy $Q'_{i+1}$ of
  $\lev_{i+1}(\A)$\:\!\! “between” $\lev_i(\A)$ and $\lev_{i+1}(\A)$,
  and defining the states in $Q'_{i+1}$ to be universal, i.e.,
  $Q'_{i+1}⊆Q'_\AA$. Then, we redirect the outgoing transitions of
  states in $\lev_i(\A)$, such that, instead of going to states in
  $\lev_{i+1}(\A)$, they go to the corresponding copies in
  $Q'_{i+1}$. Turning to these copies, we direct all their outgoing
  transitions to the matching original states in $\lev_{i+1}(\A)$.

  More formally, for every state $q'∈Q'_{i+1}$, we denote by $q$ the
  state in $\lev_{i+1}(\A)$ to which it corresponds (i.e., $q'$ is the
  copy of $q$). Now, for every $p∈\lev_i(\A)$ and $\S∈(2^Q)^Γ$\!, we
  define
  \begin{equation*}
    δ'(p,\S) = \Bigl\{ q'∈Q'_{i+1} \Bigm| q∈δ(p,\S) \Bigr\}
    ∪ \Bigl( δ(p,\S)∩Q_\P \Bigr),
  \end{equation*}
  and for every $q'∈Q'_{i+1}$ and $\S∈(2^{Q'})^Γ$\!, we set
  \begin{equation*}
    δ'(q',\S)=\{q\}.
  \end{equation*}

  In the dual case, where the levels $i$ and $i+1$ are both universal
  in $\A$, we proceed analogously. By doing so for every level in $\A$
  that is directly followed by a level of the same type, and otherwise
  retaining the original transitions of $\A$, we achieve that $\A'$ is
  in alternating normal form. Each level $i∈\{1,…,\len(\A)-1\}$ is
  duplicated at most once, thus $\A'$ cannot exceed twice the size and
  length of $\A$.

  Since the additional levels that we have introduced merely cause the
  runs of $\A'$ to be longer than those of $\A$, without affecting
  which permanent configurations are eventually reached, $\A'$ is
  obviously equivalent to $\A$. It is also easy to see that if $\A$ is
  trim, so is $\A'$. Finally, we observe that a configuration can only
  be reachable by $\A'$ if it is also reachable by $\A$, or if it
  comprises only states from $Q'\setminus Q$ (i.e., states from the
  additional levels that we have introduced). By construction, any
  reachable configuration of $\A$ has as many successor configurations
  in $\A$ as in $\A'$. Furthermore, every configuration comprising
  only states from $Q'\setminus Q$ has precisely one successor
  configuration in $\A'$. Hence, if $\A$ is nonblocking, so is $\A'$.
\end{proof}

\section{Game-Theoretic Characterization} \label{sec:game-theo}
In this section, we give an alternative characterization of the
acceptance behaviour of ADGAs, using a game-theoretic approach. This
different point of view will be useful on two occasions: in
\cref{sec:closure-properties}, where we will show that ADGAs can be
easily complemented, and in \cref{sec:adga=mso}, where we will encode
the behaviour of a given ADGA into a logical formula.

The entire approach is heavily inspired by the work of Löding and
Thomas in \cite{LT00}, where they investigated the complementation of
finite automata on infinite words. A simplified variant for automata
on finite words can be found in \cite{Kum06}.

We consider games with two players: the \emph{automaton}
(player~$\EE$), and the \emph{pathfinder} (player~$\AA$).\footnote{The
  custom of calling the players “automaton” and “pathfinder” was
  introduced by Gurevich and Harrington in \cite{GH82}.} Given an ADGA
$\A$ and a labeled graph $G_λ$, the goal of the automaton is to accept
$G_λ$, while the pathfinder tries to reject it. In a way, the
automaton wants to come up with an accepting run of $\A$ on $G_λ$, and
the pathfinder seeks to refute any possible run, by finding (a path
to) a rejecting configuration. Thus, the automaton is responsible for
the nondeterministic (existential) choices, whereas the pathfinder
picks among the universal branches.

The game associated with $\A$ and $G_λ$ is represented by a directed
acyclic graph whose nodes are configurations of $\A$ on $G$. This
graph can be thought of as a superposition of all possible runs of
$\A$ on $G_λ$. We refer to its nodes as (game) \defd{positions}, but
keep the terminology used for configurations (e.g., “existential”,
“permanent”, etc.). The nonpermanent positions are divided among the
two players: existential ones belong to the automaton, universal ones
to the pathfinder. Starting at the initial configuration of $\A$ on
$G_λ$, the two players move through the graph together. At each
position, the player owning that position has to choose the next move
along one of the outgoing edges. This continues until some permanent
position is reached. If that position is accepting, the automaton
wins, otherwise the pathfinder wins. Also, if a nonpermanent position
is reached from which no move is possible, the owner of that position
loses. (This cannot happen if $\A$ is nonblocking.)

\begin{definition}[Game]
  Let $\A=⟨Σ,Γ,\Q,σ,δ,\F⟩$ be an ADGA and $G_λ$ a $Σ$-labeled
  $Γ$-graph. The \defd{game} $\defd{\J(\A,G_λ)}$ associated with $\A$
  and $G_λ$ is the tuple $⟨\wh{K},G_{κ_0},\arr⟩$ defined as follows:
  \begin{itemize}
  \item $\wh{K}=⟨K_\EE,K_\AA,\Kacc,\Krej⟩$, where
    \begin{itemize}
    \item $K_\EE⊆⟨Q_\EE,Q_\P⟩^G$ and $K_\AA⊆⟨Q_\AA,Q_\P⟩^G$ are the
      sets of existential and universal configurations, respectively,
      reachable by $\A$ on $G_λ$,
    \item $\Kacc,\Krej⊆Q_\P^G$ are the sets of accepting and rejecting
      configurations, respectively, reachable by $\A$ on $G_λ$.
    \end{itemize}
    We use the abbreviation $K \coloneqq K_\EE∪K_\AA∪\Kacc∪\Krej$.
  \item $G_{κ_0}=G_{σ∘λ}$ is the \defd{starting position} of the game.
  \item ${\arr}⊆K×K$ is the set of directed edges for which $⟨K,\arr⟩$
    constitutes a directed acyclic graph, such that
    \begin{itemize}
    \item $G_{κ_0}$ is the only source,
    \item every node $G_κ∈(K_\EE∪K_\AA)$ with
      $δ^\cloud(G_κ)=\{G_{μ_1},…,G_{μ_m}\}$ has exactly $m$ outgoing
      neighbors $G_{μ_1},…,G_{μ_m}∈K$,\, and
    \item every node $G_κ∈(\Kacc∪\Krej)$ is a sink.
    \end{itemize}
  \end{itemize}
\end{definition}

If $\A$ and $G_λ$ are not relevant in a given context, we refer to
$\J(\A,G_λ)$ simply as \emph{a} game $J=⟨\wh{K},G_{κ_0},\arr⟩$. For
convenience, we will apply graph-theoretic notions directly to $J$,
referring implicitly to its underlying graph $⟨K,\arr⟩$, e.g., by “a
path in $J$” we mean “a path in $⟨K,\arr⟩$”.

\newpage

\begin{example}[Game associated with $\sA{centric}$]
  \Cref{fig:game} represents the game
  $\J\bigl(\sA{centric},(G_2)_{λ_2}\bigr)$ associated with the ADGA
  from \cref{fig:ADGA_concentric_circles} and the labeled graph from
  \cref{fig:graph_labeled_square}. The green configuration is the
  starting position and belongs to the automaton, whereas the red
  positions belong to the pathfinder. Just as for runs, the blue
  positions with a double circle are accepting, and the other blue
  positions rejecting.

  The game is a superposition of all possible runs of $\sA{centric}$
  on $(G_2)_{λ_2}$, in the sense that it contains all of these runs as
  subgraphs. For instance, we obtain the run from
  \cref{fig:run_rejecting} as the subgraph induced by the green
  position, the third red position from the top, and the four
  bottommost blue positions.
\end{example}

\begin{figure}[htb]
  \alignpic
  \input{fig/game.tikz}
  \caption{The game associated with the ADGA $\sA{centric}$ from
    \cref{fig:ADGA_concentric_circles} and the labeled graph
    $(G_2)_{λ_2}$ from \cref{fig:graph_labeled_square}. The edges
    highlighted in red represent a winning strategy for the
    pathfinder.}
  \label{fig:game}
\end{figure}
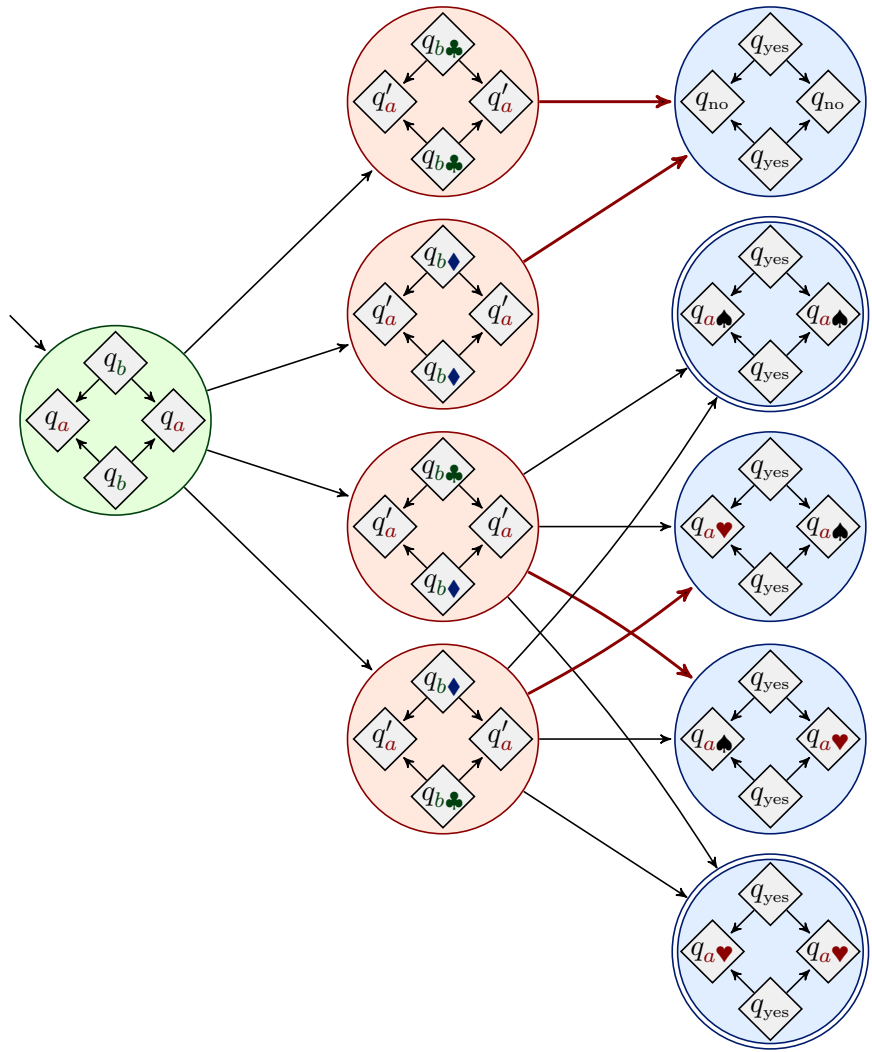

It remains to formalize how a game is played, and how the winner is
determined.

\begin{definition}[Play]
  A \defd{play} $π$ in a game $J$ is a path from the starting position
  to some sink (i.e., a maximal path). The \defd{winner} of the play
  $π=G_{κ_0}\cdots G_{κ_n}$ (with respect to $J$) is
  \begin{itemize}
  \item the automaton, if $G_{κ_n}$ is either accepting or universal,
  \item the pathfinder, if $G_{κ_n}$ is either rejecting or
    existential.
  \end{itemize}
\end{definition}

We will refer to $G_{κ_n}$ as a \defd{goal position} of the automaton
or the pathfinder, depending on which player wins when reaching that
position.

The moves chosen by the two players in a play $π$ are determined by
their respective strategies.

\begin{definition}[Strategy]
  A (memoryless or positional) \defd{strategy} for player
  $\X∈\{\EE,\AA\}$ in a game $J=⟨\wh{K},G_{κ_0},\arr⟩$ is a partial
  function $f_\X\colon K_\X → K$, such that $f_\X(G_κ)$ is an outgoing
  neighbor of $G_κ$, for every position $G_κ$ in $J$ that belongs to
  player~$\X$ and is not a sink. A play $π=G_{κ_0}\cdots G_{κ_n}$ is
  \defd{played according} to $f_\X$ \Iff $f_\X(G_{κ_i})=G_{κ_{i+1}}$
  for every node $G_{κ_i}$ on $π$ that belongs to player~$\X$, where
  $0≤i<n$.

  We say that $f_\X$ is a \defd{winning strategy} for player~$\X$ \Iff
  that player wins every play $π$ in $J$ played according to $f_\X$.
\end{definition}

\begin{example}[Winning Strategy for the Pathfinder]
  We consider again the game $\J\bigl(\sA{centric},(G_2)_{λ_2}\bigr)$
  from \cref{fig:game}. No matter which position the automaton chooses
  in the first round, the pathfinder can always move to a rejecting
  position in the second round. One possible winning strategy for the
  pathfinder is represented by the edges highlighted in red.

  The fact that the pathfinder has a winning strategy in this game is
  essentially a restatement of the observation made in
  \cref{ex:A_centric_language}: every run of $\sA{centric}$ on
  $(G_2)_{λ_2}$ is rejecting.
\end{example}

The previous example already suggests a strong relationship between
the acceptance behaviour of ADGAs and the winning strategies of the
two players. To no great surprise, both concepts turn out to be
equivalent.

\begin{lemma}[Acceptance and Winning Strategy] \label{lem:acc-win}
  Let $\A$ be an ADGA and $G_λ$ a labeled graph. Then $\A$ accepts
  $G_λ$ \Iff the automaton has a winning strategy in the game
  $\J(\A,G_λ)$.
\end{lemma}

\begin{proof} We give a very simple proof, for the sake of
  completeness.
  \begin{itemize}
  \item[($⇒$)] If $\A$ accepts $G_λ$, there is an accepting run $R$ of
    $\A$ on $G_λ$, from which we construct the following strategy
    $f_\EE$ for the automaton in $J=\J(\A,G_λ)$: For every position
    $G_κ$ in $J$ that belongs to the automaton and is not a sink,
    $f_\EE(G_κ)$ is the unique outgoing neighbor of $G_κ$ in $R$,
    provided that $G_κ$ occurs in $R$. Otherwise $f_\EE(G_κ)$ is some
    arbitrary outgoing neighbor of $G_κ$ in $J$. Hence, any play
    $π=G_{σ∘λ}\cdots G_κ$ in $J$ played according to $f_\EE$
    corresponds to a maximal path in $R$ (starting at
    $G_{σ∘λ}$). Since $R$ is a run and $G_κ$ has no outgoing
    neighbors, $G_κ$ is either universal or permanent. Furthermore,
    since $R$ is accepting, if $G_κ$ is permanent, then it is also
    accepting. In any case, the automaton wins the play $π$, i.e., it
    wins every play in $J$ played according to $f_\EE$, which means
    that $f_\EE$ is a winning strategy for that player.
  \item[($⇐$)] If the automaton has a winning strategy $f_\EE$ in
    $J=\J(\A,G_λ)$, we can construct the following run $R$ of $\A$ on
    $G_λ$: Starting at $G_{σ∘λ}$, for every node $G_κ$ in $R$, if
    $G_κ$ belongs to the automaton, then its unique outgoing neighbor
    in $R$ is $f_\EE(G_κ)$\;\!\footnote{The value of $f_\EE(G_κ)$
      cannot be undefined, because otherwise any path in $R$ from
      $G_{σ∘λ}$ to $G_κ$ would be a play played according to $f_\EE$
      that is lost by the automaton (since $G_κ$ is existential),
      which would contradict the assumption that $f_\EE$ is a winning
      strategy for the automaton.}, and if $G_κ$ belongs to the
    pathfinder, then its outgoing neighbors in $R$ are given by
    $δ^\cloud(G_κ)$. Hence, every maximal path $π=G_{σ∘λ}\cdots G_κ$
    in $R$ corresponds to a play in $J$ played according to
    $f_\EE$. Since permanent configurations do not have any outgoing
    neighbors, the only node on $π$ that might be permanent is
    $G_κ$. Furthermore, since $f_\EE$ is a winning strategy for the
    automaton, the configuration $G_κ$ is either accepting or
    universal. Thus every permanent configuration occurring on some
    (maximal) path in $R$ is accepting, which implies that $R$ is an
    accepting run, from which follows that $\A$ accepts $G_λ$.
    \qedhere
  \end{itemize}
\end{proof}

Dually to \cref{lem:acc-win}, an ADGA rejects a labeled graph \Iff the
pathfinder has a winning strategy in the associated game. Instead of
proving this directly, we can infer it from the following determinacy
result.

\begin{lemma}[Determinacy] \label{lem:either-win}
  In every game $J$, either the automaton or the pathfinder has a
  winning strategy.
\end{lemma}

\begin{proof}
  We consider every position $G_κ$ of the game
  $J=⟨\wh{K},G_{κ_0},\arr⟩$ as the starting position of a subgame
  $J_κ$ which is obtained by restricting $J$ to the subgraph induced
  by $G_κ$ and all its descendant configurations. Note that this
  implies that $J_{κ_0}=J$. We show by induction on the structure of
  the game, that for every position $G_κ$ in $J$, either the automaton
  or the pathfinder has a winning strategy in the induced subgame
  $J_κ$.
  \begin{itemize}
  \item[(\texttt{BC})] If $G_κ$ is a sink, then $J_κ$ consists only of
    the single position $G_κ$, and the only possible play in $J_κ$ is
    $π=G_κ$. The automaton wins that play if $G_κ$ is either accepting
    or universal, otherwise the pathfinder wins. In any case, one of
    the two players has a (trivial) winning strategy.
  \item[(\texttt{IS})] Now consider the case that $G_κ$ has $m>0$
    outgoing neighbors $G_{μ_1},…,\ab G_{μ_m}$. Let
    player~$\X∈\{\EE,\AA\}$ be the player who has to make a move at
    position $G_κ$, i.e., $G_κ∈K_\X$, and let player~$\Y$ be the
    opponent. By induction hypothesis, we know that for each of the
    subgames $J_{μ_1},…,J_{μ_m}$, either player~$\X$ or player~$\Y$
    has a winning strategy. There are two possible cases:
    \begin{itemize}
    \item If player~$\X$ has a winning strategy $f_\X$ in some subgame
      $J_{μ_i}$ ($1≤i≤m$), this strategy can be extended to a winning
      strategy $f'_\X$ in $J_κ$, where for every node $G_ν$ in $J_κ$
      that belongs to player~$\X$ and has at least one outgoing
      neighbor,
      \begin{equation*}
        f'_\X(G_ν) =
        \begin{cases}
          G_{μ_i} & \text{if\, $G_ν=G_κ$}, \\
          f_\X(G_ν) & \text{if\, $G_ν$ is in $J_{μ_i}$}, \\
          G_{ν'} & \text{\parbox[t]{.5\textwidth}{otherwise, where
              $G_{ν'}$ is some arbitrary outgoing neighbor of $G_ν$
              in $J_κ$.}}
        \end{cases}
      \end{equation*}
    \item Otherwise, player~$\Y$ has winning strategies $f_{\Y
        1},…,f_{\Y m}$ for each of the subgames $J_{μ_1},…,J_{μ_m}$,
      respectively. These can be combined into a winning strategy
      $f'_\Y$ in $J_κ$, such that for every node $G_ν$ in $J_κ$ that
      belongs to player~$\Y$ and has at least one outgoing neighbor,
      $f'_\Y(G_ν)=f_{\Y i}(G_ν)$, where $i∈\{1,…,m\}$ is the
      smallest\footnote{The subgames $J_{μ_1},…,J_{μ_m}$ are not
        necessarily disjoint. If a position occurs in several
        subgames, player~$\Y$ can arbitrarily choose which winning
        strategy to follow at that position, since any choice will
        lead the play one step closer to some goal position of
        player~$\Y$.} index for which the corresponding subgame
      $J_{μ_i}$ contains $G_ν$.
    \end{itemize}
    In both cases, either the automaton or the pathfinder has a
    winning strategy in $J_κ$.
    \qedhere
  \end{itemize}
\end{proof}

\section{Closure Properties} \label{sec:closure-properties}
Building on the results from \cref{sec:normal-forms,sec:game-theo}, we
can now establish some closure properties of the class of
ADGA-recognizable graph languages.

Complementation can be achieved by a simple dualization construction,
which does not involve any blow-up. We have already used this
implicitly in the examples of \cref{sec:adga-preview}, where the ADGA
$\dA[color]{3}$ from \cref{fig:ADGA_3_colorable} was complemented by
changing its existential states to universal ones and complementing
its acceptance condition. This resulted in the ADGA $\dcA[color]{3}$
from \cref{fig:ADGA_not_3_colorable}. The following definition
generalizes this construction for arbitrary ADGAs.

\begin{definition}[Dual Automaton]
  Let $\A=\bigl\langle Σ,Γ,⟨Q_\EE,Q_\AA,Q_\P⟩,σ,δ,\F \bigr\rangle$ be
  an ADGA. Its \defd{dual automaton} $\cA$ is obtained by swapping the
  existential and universal states, and complementing the set of
  accepting states, i.e.,
  \begin{equation*}
    \cA=\bigl\langle Σ,Γ,⟨Q_\AA,Q_\EE,Q_\P⟩,σ,δ,2^{Q_\P}\setminus\F \bigr\rangle.
  \end{equation*}
\end{definition}

To show that the dual automaton is always a complement automaton, we
first look at this construction from the game-theoretic point of view.

\begin{lemma*} \label{lem:dualization}
  Consider an ADGA $\A$ over $⟨Σ,Γ⟩$ and a labeled graph
  $G_λ∈Σ^{\clouded{Γ}}$. Then the automaton has a winning strategy in
  the game $\J(\A,G_λ)$ \Iff the pathfinder has a winning strategy in
  the dual game $\J(\cA,G_λ)$.
\end{lemma*}

\begin{proof}
  Let
  $J=\J(\A,G_λ)=\bigl\langle⟨K_\EE,K_\AA,\Kacc,\Krej⟩,G_{κ_0},\arr\bigr\rangle$.
  We observe that the dual game $\cJ=\J(\cA,G_λ)$ has the same
  underlying graph and starting position as $J$, only the roles and
  winning conditions of the two players have been interchanged, i.e.,
  $\cJ=\bigl\langle⟨K_\AA,K_\EE,\Krej,\Kacc⟩,G_{κ_0},\arr\bigr\rangle$. Hence,
  every play $π$ in $J$ is also a play in $\cJ$, and vice versa. Due
  to the complementarity of the winning conditions,
  player~$\X∈\{\EE,\AA\}$ wins $π$ in $J$ \Iff its opponent,
  player~$\Y$, wins $π$ in $\cJ$. Moreover, the reversal of roles
  ensures that a strategy $f$ for player~$\X$ in the one game is a
  strategy for player~$\Y$ in the other game. Thus, player~$\X$ wins
  every play played according $f$ in $J$ \Iff player~$\Y$ wins every
  play played according $f$ in $\cJ$.
\end{proof}

It is now straightforward to prove the desired result.

\begin{lemma}[Complementation] \label{lem:complementation}
  For every ADGA $\A$ over $⟨Σ,Γ⟩$, the dual automaton $\cA$
  recognizes the complement language of $\A$, i.e.,
  \begin{equation*}
    \L(\cA) = Σ^{\clouded{Γ}} \setminus \L(\A).
  \end{equation*}
\end{lemma}

\begin{proof}
  Let $G_λ∈Σ^{\clouded{Γ}}$. By \cref{lem:acc-win}, $\A$ accepts $G_λ$
  \Iff the automaton has a winning strategy in the game
  $\J(\A,G_λ)$. By Lemma~\ref{lem:dualization}, this is equivalent to
  the pathfinder having a winning strategy in the dual game
  $\J(\cA,G_λ)$.  By \cref{lem:either-win}, this is the case \Iff the
  automaton does not have a winning strategy in $\J(\cA,G_λ)$. Again
  by \cref{lem:acc-win}, this is equivalent to saying that $\cA$ does
  not accept $G_λ$. Hence, $\A$ accepts $G_λ$ \Iff $\cA$ does not
  accept $G_λ$.
\end{proof}

Next, we prove closure under union and intersection. The following
constructions exploit the power of nondeterminism and universal
branching, and are, in principle, very similar to the corresponding
constructions for alternating automata on words. However, because of
the distributed nature of ADGAs, they are slightly more technical
(local choices must be coordinated).

\begin{lemma}[Union and Intersection] \label{lem:union-intersect}
  For every two ADGAs $\A_1$ and $\A_2$ over $⟨Σ,Γ⟩$, we can
  effectively construct ADGAs $\A_∪$ and $\A_∩$ that recognize the
  union language and intersection language, respectively, of $\A_1$
  and $\A_2$, i.e.,
  \begin{equation*}
    \L(\A_∪) = \L(\A_1)∪\L(\A_2),\!
    \quad \text{and} \quad
    \L(\A_∩) = \L(\A_1)∩\L(\A_2).
  \end{equation*}
  Moreover, \vspace{-1.5ex}
  \begin{align*}
    \siz(\A_∪)&=\swr{\siz}{\len}(\A_∩)=\siz(\A_1)+\siz(\A_2)+\card{Σ}+1,
    \quad \text{and} \\
    \len(\A_∪)&=\len(\A_∩)=\max\bigl\{\len(\A_1),\,\len(\A_2)\bigr\}+1.
  \end{align*}
\end{lemma}

\begin{proof}
  Let $\A_1=⟨Σ,Γ,\Q_1,σ_1,δ_1,\F_1⟩$ and $\A_2=⟨Σ,Γ,\Q_2,\ab
  σ_2,δ_2,\F_2⟩$. Without loss of generality, we may assume that
  \begin{itemize}[topsep=1ex,itemsep=0ex]
  \item both automata are nonblocking and trim (see
    \cref{rem:nonblocking,rem:trim}), and
  \item they agree on the sequence of quantifiers, i.e., for\,
    $0≤i<\min\{\len(\A_1), \ab \len(\A_2)\}$, the state sets
    $\lev_i(\A_1)$ and $\lev_i(\A_2)$ are either both existential or
    both universal, in the respective automata. By \cref{rem:ANF}, a
    simple way to ensure this is to transform both automata into
    alternating normal form, and possibly inserting an additional
    “dummy” level into one of them.
  \end{itemize}
  Further, let $Q_Σ$ be a set of states with the same cardinality as
  $Σ$, where $q_a∈Q_Σ$ denotes the state corresponding to $a∈Σ$, and
  let $\qacc$ and $\qrej$ be two additional states. We assume that
  $Q_1$, $Q_2$, $Q_Σ$ and $\{\qacc,\qrej\}$ are pairwise disjoint.

  First, we construct the union automaton. The idea is that, in the
  first round, each node in the input graph nondeterministically and
  independently decides whether to behave like in $\A_1$ or in
  $\A_2$. If there is a consensus, then the run continues as it would
  in the unanimously chosen automaton $\A_j$, and it is accepting \Iff
  it corresponds to an accepting run of $\A_j$. Otherwise, a conflict
  is detected, either locally by adjacent nodes, or at the latest,
  when acceptance is checked globally, and in either case the run is
  rejecting. Formally, we define $\A_∪=⟨Σ,Γ,\Q_∪,σ_∪,δ_∪,\F_∪⟩$,
  where
  \begin{itemize}[topsep=1ex,itemsep=0ex]
  \item $(Q_∪)_\mEE = (Q_1)_\mEE ∪ (Q_2)_\mEE ∪ Q_Σ$,
  \item $(Q_∪)_\mAA = (Q_1)_\mAA ∪ (Q_2)_\mAA$,
  \item $(Q_∪)_\mP  = (Q_1)_\mP ∪ (Q_2)_\mP ∪ \{\qrej\}$,
  \item $σ_∪(a) = q_a, \quad \text{for every $a∈Σ$}$,
  \item
    $δ_∪(q,\S) =
    \begin{cases}
      \{σ_1(a),\,σ_2(a)\} & \text{if \,$q=q_a∈Q_Σ$\, and \,$\S∈(2^{Q_Σ})^Γ$}\!, \\
      \,δ_1(q,\S) & \text{if \,$q∈Q_1$\, and \,$\S∈(2^{Q_1})^Γ$}\!, \\
      \,δ_2(q,\S) & \text{if \,$q∈Q_2$\, and \,$\S∈(2^{Q_2})^Γ$}\!, \\
      \{\qrej\} & \text{otherwise},
    \end{cases}$ \\
    for every $q∈(Q_∪)_\N$ and $\S∈(2^{Q_∪})^Γ$\!,
  \item $\F_∪ = \F_1 ∪ \F_2$.
  \end{itemize}
  Note that $δ_∪$ satisfies the properties required by the definition
  of an ADGA, in particular, that states on the same level are in the
  same component of $\Q_∪$, which is guaranteed by the assumptions
  that $\A_1$ and $\A_2$ are both trim and agree on the sequence of
  quantifiers.

  Next, we verify that, for any $G_λ∈Σ^{\clouded{Γ}}$,\:\! $\A_∪$
  accepts $G_λ$ \Iff $\A_1$ or $\A_2$ accepts $G_λ$.
  \begin{itemize}
  \item[($⇐$)] Let one of the two automata, say $\A_j$,\, $j∈\{1,2\}$,
    have an accepting run $R=⟨K,\arr⟩$ on $G_λ$. By construction, the
    initial configuration of $\A_∪$ on $G_λ$ is existential, and there
    is a global transition to the initial configuration of $\A_j$ on
    $G_λ$, i.e., $G_{σ_j∘λ}∈δ_∪^\cloud(G_{σ_∪∘λ})$. Moreover, any
    transition of $\A_j$ is also a transition of $\A_∪$, and
    configurations common to $\A_j$ and $\A_∪$ have the same type
    (e.g., existential, etc.) in both automata. Thus,
    \begin{equation*}
      R'=\bigl\langle K∪\{G_{σ_∪∘λ}\}, \;\, {\arr}∪\{⟨G_{σ_∪∘λ},G_{σ_j∘λ}⟩\} \bigr\rangle
    \end{equation*}
    is a run of $\A_∪$ on $G_λ$, and since $\F_j⊆\F_∪$, any permanent
    configuration occurring in $R'$ is accepting, which entails that
    $R'$ is also accepting.
  \item[($⇒$)] Now let $\A_∪$ have an accepting run $R'$ on
    $G_λ$. Since the initial configuration of $\A_∪$ on $G_λ$ is
    existential, it must have exactly one outgoing neighbor $G_κ$ in
    $R'$. By construction of $\A_∪$, every state occurring in $G_κ$
    belongs to either $\A_1$ or $\A_2$, i.e., $G_κ∈(Q_1∪Q_2)^G$.

    Assume that states of both automata occur in $G_κ$. We generalize
    this property by calling any configuration $G_μ$ of $\A_∪$
    \emph{impure} if
    $G_μ∈(Q_1∪Q_2∪\{\qrej\})^G\setminus(Q_1^G∪Q_2^G)$.
    \begin{itemize}[topsep=0ex,itemsep=0ex]
    \item Any successor configuration of an impure configuration is
      also impure, because there are no local transitions from
      $Q_2∪\{\qrej\}$ to $Q_1$, or from $Q_1∪\{\qrej\}$ to $Q_2$.
    \item Further, any impure permanent configuration $G_μ$ is
      rejecting, because $G_μ$ being impure means that $\{μ(v) \mid
      v∈\VG\}∉(2^{Q_1}∪2^{Q_2})$, and consequently the acceptance
      condition given by $\F_1∪\F_2$ cannot be fulfilled.
    \end{itemize}
    Hence, the above assumption implies that all the permanent
    configurations among $G_κ$ and its descendant configurations under
    $(δ_∪)^\cloud$ are rejecting. Since we required $\A_1$ and $\A_2$
    to be nonblocking, such permanent configurations must exist (every
    nonpermanent configuration has at least one successor
    configuration). It follows that $R'$ is not accepting, which is a
    contradiction.

    We conclude that only states of one automaton, say $\A_j$,\,
    $j∈\{1,2\}$, can occur in $G_κ$. More precisely, $G_κ$ is the
    initial configuration of $\A_j$ on $G_λ$, i.e.,
    $G_κ=G_{σ_j∘λ}$. Since $\A_∪$ behaves like $\A_j$ on
    configurations of $\A_j$, this means that $R'$ is exactly of the
    same form as the run constructed in the previous part of this
    proof, i.e.,
    \begin{equation*}
      R'=\bigl\langle \{G_{σ_∪∘λ}\}\djun K, \;\, \{⟨G_{σ_∪∘λ},G_{σ_j∘λ}⟩\}\djun{\arr} \bigr\rangle,
    \end{equation*}
    where $K⊆Q_j^G$ and ${\arr}⊆K×K$. By removing the initial
    configuration (and its outgoing edge), we get a run $R=⟨K,\arr⟩$
    of $\A_j$ on $G_λ$, which is accepting because
    $\F_∪∩2^{Q_j}=\F_j$.
  \end{itemize}
  Finally, we turn our attention to the intersection automaton. By De
  Morgan's law, $\overline{\L(\A_1)∩\L(\A_2)} =
  \overline{\L(\A_1)}∪\overline{\L(\A_2)}$,\:\! hence we can simply
  combine the available constructions for complementation and union,
  which leads to the intersection automaton
  $\A_∩=⟨Σ,Γ,\Q_∩,σ_∩,δ_∩,\F_∩⟩$, where
  \begin{itemize}[topsep=1ex,itemsep=0ex]
  \item $(Q_∩)_\mEE = (Q_1)_\mEE ∪ (Q_2)_\mEE$,
  \item $(Q_∩)_\mAA = (Q_1)_\mAA ∪ (Q_2)_\mAA ∪ Q_Σ$,
  \item $(Q_∩)_\mP  = (Q_1)_\mP ∪ (Q_2)_\mP ∪ \{\qacc\}$,
  \item $σ_∩(a) = q_a, \quad \text{for every $a∈Σ$}$,
  \item
    $δ_∩(q,\S) =
    \begin{cases}
      \{σ_1(a),\,σ_2(a)\} & \text{if \,$q=q_a∈Q_Σ$\, and \,$\S∈(2^{Q_Σ})^Γ$}\!, \\
      \,δ_1(q,\S) & \text{if \,$q∈Q_1$\, and \,$\S∈(2^{Q_1})^Γ$}\!, \\
      \,δ_2(q,\S) & \text{if \,$q∈Q_2$\, and \,$\S∈(2^{Q_2})^Γ$}\!, \\
      \{\qacc\} & \text{otherwise},
    \end{cases}$ \\
    for every $q∈(Q_∩)_\N$ and $\S∈(2^{Q_∩})^Γ$\!,
  \item $\F_∩ = \F_1 ∪ \F_2 ∪ \bigl(2^{(Q_∩)_\P} \setminus
    (2^{(Q_1)_\P}∪2^{(Q_2)_\P})\bigr)$.
    \qedhere
  \end{itemize}
\end{proof}

As a last type of operation on graph languages, we consider uniform
relabelings of nodes, which we call node projections and formally
define as follows.

\begin{definition}[Projection]
  Let $Σ$ and $Σ'$ be two nonempty node alphabets and $Γ$ an edge
  alphabet. A (node) \defd{projection} from $Σ$ to $Σ'$ is a mapping
  $h\colon Σ → Σ'$. With a slight abuse of notation, this mapping is
  extended to labeled graphs by applying it to each node label, and to
  graph languages by applying it to each labeled graph. More
  precisely, for every $G_λ∈Σ^{\clouded{Γ}}$ and $L⊆Σ^{\clouded{Γ}}$,
  \begin{equation*}
    h(G_λ) \coloneqq G_{h∘λ},\!
    \quad \text{and} \quad  
    h(L) \coloneqq \{h(G_λ)\mid G_λ∈L\}.
  \end{equation*}
\end{definition}

Again exploiting the power of nondeterminism, we can easily show that
ADGA-recognizable graph languages are closed under arbitrary
projections.

\begin{lemma}[Projection] \label{lem:projection}
  For every ADGA $\A$ over $⟨Σ,Γ⟩$ and projection\, $h\colon Σ → Σ'$,
  we can effectively construct an ADGA $\A'$ that recognizes the
  projected language of $\A$ through $h$, i.e.,
  \begin{equation*}
    \L(\A') = h(\L(\A)).
  \end{equation*}
  Moreover, \vspace{-1.5ex}
  \begin{equation*}
    \siz(\A')=\siz(\A)+\card{Σ'}
    \quad \text{and} \quad
    \len(\A')=\len(\A)+1.
  \end{equation*}
\end{lemma}

\begin{proof}
  The idea is simple: For every $b∈Σ'$, each node labeled with $b$
  nondeterministically chooses a new label $a∈Σ$, such that
  $h(a)=b$. Then the automaton $\A$ is simulated on that new input.

  Without loss of generality, we may assume that $\A$ is trim (see
  \cref{rem:trim}). Let $\A=⟨Σ,Γ,\Q,σ,δ,\F⟩$, and let $Q'_{Σ'}$ be a
  set of states with the same cardinality as $Σ'$ and disjoint from
  $Q$, where $q_b∈Q'_{Σ'}$ denotes the state corresponding to
  $b∈Σ'$. We construct the projection automaton $\A'=\bigl\langle
  Σ',Γ,⟨Q'_\EE,Q_\AA,Q_\P⟩,σ',δ',\F \bigr\rangle$, where
  \begin{itemize}[topsep=1ex,itemsep=0ex]
  \item $Q'_\EE = Q_\EE ∪ Q'_{Σ'}$,
  \item $σ'(b) = q_b, \quad \text{for every $b∈Σ'$}$,
  \item
    $δ'(q,\S) =
    \begin{cases}
      \{σ(a)\mid h(a)=b\} & \text{if \,$q=q_b∈Q'_{Σ'}$\, and \,$\S∈(2^{Q'_{Σ'}})^Γ$}\!, \\
      \,δ(q,\S) & \text{if \,$q∈Q$\, and \,$\S∈(2^Q)^Γ$}\!, \\
      \,∅ & \text{otherwise},
    \end{cases}$ \\
    for every $q∈(Q_\N∪Q'_{Σ'}\:\!\!)$ and $\S∈(2^{Q\:\!∪\:\!Q'_{Σ'}}\:\!\!)^Γ$\!.
  \end{itemize}
  Note that, as required, states on the same level of $\A'$ are of the
  same type (e.g., existential, etc.), because we have assumed that
  $\A$ is trim. Consider any $G_{λ'}∈Σ^{\prime\:\!\clouded{Γ}}$. The
  initial configuration $G_{σ'∘λ'}$ of $\A'$ on $G_{λ'}$ is
  existential. Its successor configurations are the initial
  configurations of $\A$ on the $Σ$-labeled $Γ$-graphs that are mapped
  to $G_{λ'}$ by $h$, i.e.,
  \begin{equation*}
    δ^{\prime\:\!\cloud}(G_{σ'∘λ'}) = \bigl\{ G_{σ∘λ} \bigm| G_λ∈Σ^G \;\!∧\; h(G_λ)=G_{λ'} \bigr\}.   
  \end{equation*}
  Moreover, the behaviours of $\A$ and $\A'$ on configurations of $\A$
  are the same. Hence, arguing similarly as for the union construction
  (\cref{lem:union-intersect}), we can show that $\A'$ accepts
  $G_{λ'}$ \Iff $\A$ accepts some $G_λ∈Σ^G$ such that $h(G_λ)=G_{λ'}$.
\end{proof}

The following theorem summarizes the closure properties stated in
\cref{lem:complementation,lem:union-intersect,lem:projection}.

\begin{theorem}[Closure Properties] \label{thm:closure}
  The class $\LL_\ADGA$ of ADGA-recognizable graph languages is
  effectively closed under boolean set operations and under
  projection.
\end{theorem}

%% file: fig/ADGA_3_colorable.tikz
\begin{tikzpicture}[automaton, half row sep]
  \matrix[states] {
        & \node[existential] (q_p) {$q_\pik$}; &[6ex] \\
        &     & \node[permanent] (q_yes) {$\q{yes}$}; \\
    \node[initial,existential] (q_ini) {$\q{ini}$}; & \node[existential] (q_h) {$q_\herz$}; \\
        &     & \node[permanent] (q_no) {$\q{no}$}; \\
        & \node[existential] (q_k) {$q_\kreuz$}; \\
  };
  \path[use as bounding box]
        (q_ini)   edge (q_p)
                  edge (q_h)
                  edge (q_k)
        (q_p)     edge[bend left=25] node[above=-.3ex,xshift=.4ex] {$\xnni q_\pik$} (q_yes)
                  edge[bend left=15] node[above] {$\xni q_\pik$} (q_no)
        (q_h.10)  edge[bend right=5] node[above left=-.5ex,xshift=-1.5ex] {$\xnni q_\herz$} (q_yes)
        (q_h.350) edge[bend left=5] node[below left=-.5ex,xshift=-1.5ex] {$\xni q_\herz$} (q_no)
        (q_k)     edge[bend right=15] node[below=.3ex] {$\xnni q_\kreuz$} (q_yes)
                  edge[bend right=25] node[below,xshift=.4ex] {$\xni q_\kreuz$} (q_no);
  \matrix[accepting sets] {
       \\
    \x \\
       \\
       \\
       \\
  };
  \DrawColumnBackground{2}{4}{1}
\end{tikzpicture}

%% file: fig/ADGA_not_3_colorable.tikz
\begin{tikzpicture}[automaton, half row sep]
  \matrix[states] {
        & \node[universal] (q_p) {$q_\pik$}; &[6ex] \\
        &     & \node[permanent] (q_yes) {$\q{yes}$}; \\
    \node[initial,universal] (q_ini) {$\q{ini}$}; & \node[universal] (q_h) {$q_\herz$}; \\
        &     & \node[permanent] (q_no) {$\q{no}$}; \\
        & \node[universal] (q_k) {$q_\kreuz$}; \\
        & \\
  };
  \path[use as bounding box]
        (q_ini.20) edge (q_p)
        (q_ini)    edge (q_h)
                   edge (q_k)
        (q_p)      edge[bend left=25] node[above=-.3ex,xshift=.1ex] {$\xnni q_\pik$} (q_yes)
                   edge[bend left=15] node[above] {$\xni q_\pik$} (q_no)
        (q_h.10)   edge[bend right=5] node[above left=-.5ex,xshift=-1.5ex] {$\xnni q_\herz$} (q_yes)
        (q_h.350)  edge[bend left=5] node[below left=-.5ex,xshift=-1.3ex] {$\xni q_\herz$} (q_no)
        (q_k)      edge[bend right=15] node[below=.4ex,xshift=-.2ex] {$\xnni q_\kreuz$} (q_yes)
                   edge[bend right=25] node[below,xshift=.5ex] {$\xni q_\kreuz$} (q_no);
  \matrix[accepting sets] {
            \\
       & \x \\
            \\
    \x & \x \\
            \\
            \\
  };
  \DrawColumnBackground{2}{4}{2}
\end{tikzpicture}
\vspace{-2ex}

%% file: fig/graph_labeled_square.tikz
%
\squaregraphPic{input graph}{\lnodedistIG}{$\b$}{$\a$}{$\b$}{$\a$}

%% file: fig/run_rejecting.tikz
\begin{tikzpicture}[run or game, small configs]
  \matrix {
      & & \node[config,pacc] (c3a) {\squaregraphPic{configuration}{\lnodedistC}{$\q{yes}$}{$\qap$}{$\q{yes}$}{$\qap$}}; \\ \\
      & & \node[config,prej] (c3b) {\squaregraphPic{configuration}{\lnodedistC}{$\q{yes}$}{$\qap$}{$\q{yes}$}{$\qah$}}; \\
    \node[config,exis] (c1) {\squaregraphPic{configuration}{\lnodedistC}{$\qb$}{$\qa$}{$\qb$}{$\qa$}}; 
      & \node[config,univ] (c2) {\squaregraphPic{configuration}{\lnodedistC}{$\qbkr$}{$\qaprime$}{$\qbk$}{$\qaprime$}}; \\
      & & \node[config,prej] (c3c) {\squaregraphPic{configuration}{\lnodedistC}{$\q{yes}$}{$\qah$}{$\q{yes}$}{$\qap$}}; \\ \\
      & & \node[config,pacc] (c3d) {\squaregraphPic{configuration}{\lnodedistC}{$\q{yes}$}{$\qah$}{$\q{yes}$}{$\qah$}}; \\
  };
  \path (c1) edge (c2)
        (c2) edge (c3a)
             edge (c3b)
             edge (c3c)
             edge (c3d);
\end{tikzpicture}

%% file: fig/ADGA_weakly_connected.tikz
\begin{tikzpicture}[automaton]
  \matrix[states] {
                                                  &  \node[existential] (q_p1) {$q_\pik$};  & \node[permanent] (q_p2) {$\qpikprime$}; \\
    \node[initial,universal] (q_ini) {$\q{ini}$}; &                                          & \node[permanent] (q_acc) {$\qacc$}; \\
                                                  &  \node[existential] (q_h1) {$q_\herz$}; & \node[permanent] (q_h2) {$\qherzprime$}; \\
  };
  \matrix[accepting sets] {
    \x &    &    & \x &    & \x \\
       & \x &    & \x & \x & \x \\
       &    & \x &    & \x & \x \\
  };
  \DrawColumnBackground{1}{3}{6}
  \path (q_ini.15) edge (q_p1)
        (q_ini)    edge (q_h1)
        (q_p1)     edge node[above] {$\xnni q_\herz$} (q_p2)
                   edge node[below left=-.7ex] {$\xni q_\herz$} (q_acc)
        (q_h1)     edge node[above left=-.7ex] {$\xni q_\pik$} (q_acc)
                   edge node[below] {$\xnni q_\pik$} (q_h2);
\end{tikzpicture}

%% file: fig/ADGA_directed_tree.tikz
\begin{tikzpicture}[automaton, half row sep]
  \matrix[states] {
        &     &     &[4ex] \node[permanent] (q_np) {$q^\notlarr_\pik$}; \\
        & \node[universal] (q_n) {$q^\notlarr\!\!$}; \\
        &     &     & \node[permanent] (q_nh) {$q^\notlarr_\herz$}; \\
    \node[initial,existential] (q_ini) {$\q{ini}$}; \\
        &     & \node[existential] (q_rp) {$q^←_\pik$}; & \node[permanent] (q_yes) {$\q{yes}$}; \\
        & \node[universal] (q_r) {$q^←\!\!$}; \\
        &     & \node[existential] (q_rh) {$q^←_\herz$}; & \node[permanent] (q_no) {$\q{no}$}; \\
  };
  \matrix[accepting sets] {
    \x &    \\
       &    \\
       & \x \\
       &    \\
    \x & \x \\
       &    \\
       &    \\
  };
  \DrawColumnBackground{1}{7}{2}
  \path (q_ini)    edge node[above left,xshift=-.7ex,yshift=-.2ex] {$\xeq ∅$} (q_n)
                   edge node[below left] {$\xneq ∅$} (q_r)
        (q_n.15)   edge (q_np)
        (q_n.345)  edge (q_nh)
        (q_r.15)   edge (q_rp)
        (q_r)      edge (q_rh)
        (q_rp.25)  edge[bend left=10] node[above] {$\#=1$} (q_yes)
        (q_rp)     edge[bend left=25] node[below left,yshift=.5ex] {$\#>1$} (q_no)
        (q_rh)     edge[bend right=25] node[above left,yshift=-.7ex] {$\#=1$} (q_yes)
        (q_rh.335) edge[bend right=10] node[below] {$\#>1$} (q_no);
\end{tikzpicture}

%% file: fig/ADGA_undirected.tikz
\begin{tikzpicture}[automaton, half row sep]
  \matrix[states] {
        & \node[existential] (q_s) {$\qs$}; & \node[existential] (q_spr) {$\qspr$}; & \node[permanent] (q_no) {$\qno$}; \\
        \\
    \node[initial,universal] (q_ini) {$\q{ini}$}; \\
        &     & \node[existential] (q_r) {$\qr$}; & \node[permanent] (q_yes) {$\qyes$}; \\
        & \node[existential] (q_ns) {$\qns$}; & &  \\
        &     & \node[existential] (q_nr) {$\qnr$}; \\
  };
  \matrix[accepting sets] {
       \\
       \\
       \\
    \x \\
       \\
       \\
  };
  \DrawColumnBackground{1}{4}{1}
  \path (q_ini.15)  edge (q_s)
        (q_ini)     edge (q_ns)
        (q_s)       edge (q_spr)
        (q_ns)      edge node[above=.5ex] {$\xni\qs$} (q_r)
                    edge node[below] {$\xnni\qs$} (q_nr)
        (q_spr)     edge node[above] {$\xni\qnr$} (q_no)
        (q_spr.335) edge node[above right=-.6ex] {$\xnni\qnr$} (q_yes)
        (q_r)       edge (q_yes)
        (q_nr)      edge (q_yes);
\end{tikzpicture}

%% file: fig/ADGA_K3_minor.tikz
\begin{tikzpicture}[automaton]
  \matrix[states] {
        &     &     &[18ex] \node[permanent] (q_p1r) {$q_{\pik1}^←$}; \\
        &     & \node[existential] (q_p1) {$q_{\pik1}$}; & \node[permanent] (q_p1n) {$q_{\pik1}^\notlarr$}; \\
        & \node[universal] (q_1) {$q_1$}; & & \node[permanent] (q_acc1) {$\q{acc}$}; \\
        &     & \node[existential] (q_h1) {$q_{\herz1}$}; & \node[permanent] (q_h1r) {$q_{\herz1}^←$}; \\
        &     &     & \node[permanent] (q_h1n) {$q_{\herz1}^\notlarr$}; \\[2ex]
        &     &     & \node[permanent] (q_p2r) {$q_{\pik2}^←$}; \\
        &     & \node[existential] (q_p2) {$q_{\pik2}$}; & \node[permanent] (q_p2n) {$q_{\pik2}^\notlarr$}; \\
    \node[initial,existential] (q_ini) {$\q{ini}$}; & \node[universal] (q_2) {$q_2$}; & & \node[permanent redundant] (q_acc2) {$\q{acc}$}; \\
        &     & \node[existential] (q_h2) {$q_{\herz2}$}; & \node[permanent] (q_h2r) {$q_{\herz2}^←$}; \\
        &     &     & \node[permanent] (q_h2n) {$q_{\herz2}^\notlarr$}; \\[2ex]
        &     &     & \node[permanent] (q_p3r) {$q_{\pik3}^←$}; \\
        &     & \node[existential] (q_p3) {$q_{\pik3}$}; & \node[permanent] (q_p3n) {$q_{\pik3}^\notlarr$}; \\
        & \node[universal] (q_3) {$q_3$}; & & \node[permanent redundant] (q_acc3) {$\q{acc}$}; \\
        &     & \node[existential] (q_h3) {$q_{\herz3}$}; & \node[permanent] (q_h3r) {$q_{\herz3}^←$}; \\
        &     &     & \node[permanent] (q_h3n) {$q_{\herz3}^\notlarr$}; \\[2ex]
        &     &     & \node[permanent] (q_out) {$\q{out}$}; \\
  };
  \node[accepting formula] {
    $\xni \q{acc} ∨\:
     \smashoperator{\bigwedge_{1≤j≤3}}\,
     \Bigl(\bigl(\xni q_{\pik j}^← ∧ \xnni q_{\herz j}^← ∧ \xnni q_{\herz j}^\notlarr\bigr) ∨
     \bigl(\xni q_{\herz j}^← ∧ \xnni q_{\pik j}^← ∧ \xnni q_{\pik j}^\notlarr\bigr)\Bigr)$
  };
  \path[every node/.style={sloped}]
      (q_ini)    edge (q_1)
                 edge (q_2)
                 edge (q_3)
                 edge[bend right=56] (q_out)
      (q_1.15)   edge (q_p1)
      (q_1)      edge (q_h1)
      (q_2.15)   edge (q_p2)
      (q_2)      edge (q_h2)
      (q_3.15)   edge (q_p3)
      (q_3)      edge (q_h3)
      (q_p1.40)  edge node[above=.2ex] {$\xnni q_{\herz1}∧(\xni q_{\pik2}∨\xni q_{\herz2})$} (q_p1r)
      (q_p1)     edge node[above=.2ex,xshift=.8ex] {$\xnni q_{\herz1}∧¬(\xni q_{\pik2}∨\xni q_{\herz2})$} (q_p1n)
      (q_p1.320) edge node[above=.2ex,sloped=false] {$\xni q_{\herz1}$} (q_acc1)
      (q_h1.40)  edge node[above=.4ex,sloped=false] {$\xni q_{\pik1}$} (q_acc1)
      (q_h1)     edge node[above=.2ex] {$\xnni q_{\pik1}∧(\xni q_{\pik2}∨\xni q_{\herz2})$} (q_h1r)
      (q_h1.320) edge node[above=.2ex,xshift=.5ex] {$\xnni q_{\pik1}∧¬(\xni q_{\pik2}∨\xni q_{\herz2})$} (q_h1n)
      (q_p2.40)  edge node[above=.2ex] {$\xnni q_{\herz2}∧(\xni q_{\pik3}∨\xni q_{\herz3})$} (q_p2r)
      (q_p2)     edge node[above=.2ex,xshift=.8ex] {$\xnni q_{\herz2}∧¬(\xni q_{\pik3}∨\xni q_{\herz3})$} (q_p2n)
      (q_p2.320) edge node[above=.2ex,sloped=false] {$\xni q_{\herz2}$} (q_acc2)
      (q_h2.40)  edge node[above=.4ex,sloped=false] {$\xni q_{\pik2}$} (q_acc2)
      (q_h2)     edge node[above=.2ex] {$\xnni q_{\pik2}∧(\xni q_{\pik3}∨\xni q_{\herz3})$} (q_h2r)
      (q_h2.320) edge node[above=.2ex,xshift=.5ex] {$\xnni q_{\pik2}∧¬(\xni q_{\pik3}∨\xni q_{\herz3})$} (q_h2n)
      (q_p3.40)  edge node[above=.2ex] {$\xnni q_{\herz3}∧(\xni q_{\pik1}∨\xni q_{\herz1})$} (q_p3r)
      (q_p3)     edge node[above=.2ex,xshift=.8ex] {$\xnni q_{\herz3}∧¬(\xni q_{\pik1}∨\xni q_{\herz1})$} (q_p3n)
      (q_p3.320) edge node[above=.2ex,sloped=false] {$\xni q_{\herz3}$} (q_acc3)
      (q_h3.40)  edge node[above=.4ex,sloped=false] {$\xni q_{\pik3}$} (q_acc3)
      (q_h3)     edge node[above=.2ex] {$\xnni q_{\pik3}∧(\xni q_{\pik1}∨\xni q_{\herz1})$} (q_h3r)
      (q_h3.320) edge node[above=.2ex,xshift=.5ex] {$\xnni q_{\pik3}∧¬(\xni q_{\pik1}∨\xni q_{\herz1})$} (q_h3n);
\end{tikzpicture}

%% file: fig/game.tikz
\begin{tikzpicture}[run or game, small configs]
  \matrix {
      & \node[config,univ] (c2a) {\squaregraphPic{configuration}{\lnodedistC}{$\qbkr$}{$\qaprime$}{$\qbkr$}{$\qaprime$}};
        & \node[config,prej] (c3a) {\squaregraphPic{configuration}{\lnodedistC}{$\q{yes}$}{$\q{no}$}{$\q{yes}$}{$\q{no}$}}; \\ \\
      & \node[config,univ] (c2b) {\squaregraphPic{configuration}{\lnodedistC}{$\qbk$}{$\qaprime$}{$\qbk$}{$\qaprime$}};
        & \node[config,pacc] (c3b) {\squaregraphPic{configuration}{\lnodedistC}{$\q{yes}$}{$\qap$}{$\q{yes}$}{$\qap$}}; \\
    \node[config,exis] (c1) {\squaregraphPic{configuration}{\lnodedistC}{$\qb$}{$\qa$}{$\qb$}{$\qa$}}; \\
      & \node[config,univ] (c2c) {\squaregraphPic{configuration}{\lnodedistC}{$\qbkr$}{$\qaprime$}{$\qbk$}{$\qaprime$}};
        & \node[config,prej] (c3c) {\squaregraphPic{configuration}{\lnodedistC}{$\q{yes}$}{$\qap$}{$\q{yes}$}{$\qah$}}; \\ \\
      & \node[config,univ] (c2d) {\squaregraphPic{configuration}{\lnodedistC}{$\qbk$}{$\qaprime$}{$\qbkr$}{$\qaprime$}};
        & \node[config,prej] (c3d) {\squaregraphPic{configuration}{\lnodedistC}{$\q{yes}$}{$\qah$}{$\q{yes}$}{$\qap$}}; \\ \\
      & & \node[config,pacc] (c3e) {\squaregraphPic{configuration}{\lnodedistC}{$\q{yes}$}{$\qah$}{$\q{yes}$}{$\qah$}}; \\
  };
  \coordinate (init) at ([xshift=-3ex,yshift=3ex]c1.north west);
  \path (init) edge (c1)
        (c1)   edge (c2a)
               edge (c2b)
               edge (c2c)
               edge (c2d)
        (c2a)  edge[pathfinder strat] (c3a)
        (c2b)  edge[pathfinder strat] (c3a)
        (c2c)  edge (c3b)
               edge (c3c)
               edge[pathfinder strat,bend left=5] (c3d)
               edge[bend left=5] (c3e)
        (c2d)  edge[bend right=5] (c3b)
               edge[pathfinder strat,bend right=5] (c3c)
               edge (c3d)
               edge (c3e);
\end{tikzpicture}

%% file: chap/MSOL.tex

\chapter{Monadic Second-Order Logic on Graphs} \label{chap:msol}
In this chapter, we review monadic second-order (MSO) logic on labeled
graphs. Then, building on the results of \cref{chap:adga}, we prove
our main result: the ADGA-recognizable graph languages are precisely
the MSO-definable ones. This, in turn, allows us to infer some
negative properties of ADGAs.

\section{Basic Definitions}
Throughout this work, we fix two disjoint, countably infinite sets of
(object language) variables: the supply of node variables
$\Vnode=\{\lsymb{u},\lsymb{v},…,\lsymb{u_1},…\}$, and the supply of
set variables $\Vset=\{\lsymb{U},\lsymb{V},…,\lsymb{U_1},…\}$. Node
variables will always be represented by lower-case letters, and set
variables by upper-case ones (sometimes with
subscripts).\footnote{Concrete instances of such object language
  variables will be typeset in a bright blue sans-serif font, to
  better distinguish them from meta-language variables, which can
  refer (amongst others) to arbitrary object language variables.}

\begin{definition}[MSO-Logic: Syntax] \label{def:mso-syntax}
  Let $Σ$ be a node alphabet and $Γ$ an edge alphabet. The set
  \defd{$\MSO(Σ,Γ)$} of \defd{monadic second-order formulas} (on
  graphs) over $⟨Σ,Γ⟩$ is built up from the \defd{atomic} formulas
  \begin{itemize}[topsep=1ex,itemsep=0ex]
  \item $\swl{\logic{\lab{\meta{a}}\meta{x}}}{\logic{\meta{x}\xarr{\meta{γ}}\meta{y}}}$
    \quad (“$x$ has label $a$”),
  \item $\logic{\meta{x}\xarr{\meta{γ}}\meta{y}}$
    \quad (“$x$ has a $γ$-edge to $y$”),
  \item $\swl{\logic{\meta{x}=\meta{y}}}{\logic{\meta{x}\xarr{\meta{γ}}\meta{y}}}$
    \quad (“$x$ is equal to $y$”),
  \item $\swl{\logic{\meta{x}∈\meta{X}}}{\logic{\meta{x}\xarr{\meta{γ}}\meta{y}}}$
    \quad (“$x$ is an element of $X$”),
  \end{itemize}
  for all\, $x,y\!∈\!\Vnode$,\, $X\!∈\!\Vset$,\, $a\!∈\!Σ$,\, and
  $γ\!∈\!Γ$, using the usual propositional connectives and
  quantifiers, which can be applied to both node and set
  variables. More precisely, if $φ$ and $ψ$ are $\MSO(Σ,Γ)$-formulas,
  then so are \noheight{\mbox{$\logic{¬\meta{φ}}$,\,
      $\logic{\meta{φ}∨\meta{ψ}}$,\, $\logic{\meta{φ}∧\meta{ψ}}$,\,
      $\logic{\meta{φ}⇒\meta{ψ}}$,\, $\logic{\meta{φ}⇔\meta{ψ}}$,\,
      $\logic{∃\meta{x}(\meta{φ})}$,\, $\logic{∀\meta{x}(\meta{φ})}$,\,
      $\logic{∃\meta{X}(\meta{φ})}$, and $\logic{∀\meta{X}(\meta{φ})}$}}, \\
  for all $x∈\Vnode$ and $X∈\Vset$.
\end{definition}

We will consistently represent MSO-formulas in the typographic style
used above, to distinguish object language from
meta-language.\footnote{ In order to make a clear distinction between
  MSO-formulas and formal statements at the meta-level (where some of
  the same symbols are used), the former will always be represented on
  a light blue background, using a bright blue font for symbols that
  directly occur in the considered formula. In contrast to this, other
  symbols must be interpreted at the meta-level to get the intended
  MSO-formulas. As usual, notations like
  $\logic{\bigwedge_{\meta{\one≤i≤n}}\meta{φ_i}}$ and
  $\logic{\bigexists_{\meta{\one≤i≤n}}\meta{x_i}(\meta{φ})}$ are used
  to represent $\logic{\meta{φ_\one}∧\meta{\cdots}∧\meta{φ_n}}$ and
  $\logic{∃\meta{x_\one}(\meta{\cdots}\,∃\meta{x_n}(\meta{φ}))}$,
  respectively.}

An occurrence of a variable $x∈\Vnode$ or $X∈\Vset$ in a formula $φ$
is said to be \defd{free} if it is not within the scope of a
quantifier. We denote by $\defd{\free(φ)}$ the set of variables that
occur freely in $φ$. If $φ$ has no free occurrences of variables,
i.e., if\, $\free(φ)=∅$, we say that $φ$ is a
\defd{sentence}. Moreover, we will use the notation
\defd{$φ[x_1,…,x_m,X_1,…,X_n]$} to indicate that at most the variables
given in brackets occur freely in $φ$, i.e.,
$\free(φ)⊆\{x_1,…,x_m,X_1,…,X_n\}$. This notation will also
occasionally be abused to instantiate a formula (schema) with concrete
variables.\footnote{Strictly speaking, if $x∈\Vnode$ and $X∈\Vset$ are
  unspecified, an object like $φ[x,X]=\logic{\meta{x}∈\meta{X}}$ is a
  formula schema. We can instantiate it with concrete object language
  variables, for instance $\lsymb{v}$ and $\lsymb{U}$, to obtain the
  formula $\logic{v∈U}$, which (by slight abuse of notation) will be
  denoted by $φ[\lsymb{v},\lsymb{U}]$. To simplify matters, we shall
  henceforth not explicitly distinguish formula schemata from
  formulas.}

\begin{definition}[MSO-Logic: Semantics]
  The truth of an $\MSO(Σ,Γ)$-formula $φ$ is evaluated with respect to
  a labeled graph $G_λ∈Σ^{\clouded{Γ}}$ and a variable assignment\,
  $α\colon \free(φ)→\VG ∪ 2^\VG$ that assigns a node $v∈\VG$ to each
  node variable in $\free(φ)$, and a set of nodes $S⊆\VG$ to each set
  variable in $\free(φ)$. We write $\defd{⟨G_λ,α⟩⊨φ}$ to denote that
  $G_λ$ and $α$ \defd{satisfy} $φ$. If $φ$ is a sentence, the variable
  assignment is superfluous, and we simply write $\defd{G_λ⊨φ}$ if
  $G_λ$ satisfies $φ$. The meaning of the atomic formulas is as hinted
  informally in \cref{def:mso-syntax}, i.e.,
  \begin{itemize}[topsep=1ex,itemsep=0ex]
  \item
    $⟨G_λ,α⟩⊨\swl{\logic{\lab{\meta{a}}\meta{x}}}{\logic{\meta{x}\xarr{\meta{γ}}\meta{y}}}
    \quad⇔\quad λ(α(x))=a$,
  \item $⟨G_λ,α⟩⊨\logic{\meta{x}\xarr{\meta{γ}}\meta{y}}
    \quad⇔\quad α(x)\arrG{γ}α(y)$,
  \item
    $⟨G_λ,α⟩⊨\swl{\logic{\meta{x}=\meta{y}}}{\logic{\meta{x}\xarr{\meta{γ}}\meta{y}}}
    \quad⇔\quad α(x)=α(y)$,
  \item
    $⟨G_λ,α⟩⊨\swl{\logic{\meta{x}∈\meta{X}}}{\logic{\meta{x}\xarr{\meta{γ}}\meta{y}}}
    \quad⇔\quad α(x)∈α(X)$.
  \end{itemize}
  for all\, $x,y\!∈\!\Vnode$,\, $X\!∈\!\Vset$,\, $a\!∈\!Σ$,\, and
  $γ\!∈\!Γ$. Composed formulas are interpreted according to the usual
  semantics of second-order logic, i.e.,
  \begin{itemize}[topsep=1ex,itemsep=0ex]
  \item $⟨G_λ,α⟩⊨\swl{\logic{¬\meta{φ}}}{\logic{\meta{φ}⇔\meta{ψ}}}
    \quad⇔\quad \text{$⟨G_λ,α⟩⊭φ$}$,
  \item $⟨G_λ,α⟩⊨\swl{\logic{\meta{φ}∨\meta{ψ}}}{\logic{\meta{φ}⇔\meta{ψ}}}
    \quad⇔\quad \text{$⟨G_λ,α⟩⊨φ$ \,or\, $⟨G_λ,α⟩⊨ψ$}$,
  \item $⟨G_λ,α⟩⊨\swl{\logic{\meta{φ}∧\meta{ψ}}}{\logic{\meta{φ}⇔\meta{ψ}}}
    \quad⇔\quad \text{$⟨G_λ,α⟩⊨φ$ \,and\, $⟨G_λ,α⟩⊨ψ$}$,
  \item $⟨G_λ,α⟩⊨\logic{\meta{φ}⇒\meta{ψ}}
    \quad⇔\quad \text{$⟨G_λ,α⟩⊭φ$ \,or\, $⟨G_λ,α⟩⊨ψ$}$,
  \item $⟨G_λ,α⟩⊨\logic{\meta{φ}⇔\meta{ψ}}
    \quad⇔\quad \text{$⟨G_λ,α⟩⊨\logic{\meta{φ}⇒\meta{ψ}}$ \,and\, $⟨G_λ,α⟩⊨\logic{\meta{ψ}⇒\meta{φ}}$}$,
  \item $⟨G_λ,α⟩⊨\swl{\logic{∃\meta{x}(\meta{φ})}}{\logic{\meta{φ}⇔\meta{ψ}}}
    \quad⇔\quad
    \text{$⟨G_λ,α[x\!↦\!v]⟩⊨φ$ \,for some\, $v∈\VG$}$,
  \item $⟨G_λ,α⟩⊨\swl{\logic{∀\meta{x}(\meta{φ})}}{\logic{\meta{φ}⇔\meta{ψ}}}
    \quad⇔\quad
    \text{$⟨G_λ,α[x\!↦\!v]⟩⊨φ$ \,for all\, $v∈\VG$}$,
  \item $⟨G_λ,α⟩⊨\swl{\logic{∃\meta{X}(\meta{φ})}}{\logic{\meta{φ}⇔\meta{ψ}}}
    \quad⇔\quad
    \text{$⟨G_λ,α[X\!↦\!U]⟩⊨φ$ \,for some\, $U⊆\VG$}$,
  \item $⟨G_λ,α⟩⊨\swl{\logic{∀\meta{X}(\meta{φ})}}{\logic{\meta{φ}⇔\meta{ψ}}}
    \quad⇔\quad
    \text{$⟨G_λ,α[X\!↦\!U]⟩⊨φ$ \,for all\, $U⊆\VG$}$,
  \end{itemize}
  for all\, $φ,ψ∈\MSO(Σ,Γ)$, $x∈\Vnode$ and $X∈\Vset$. Here,
  $α[x\!↦\!v]$ designates the variable assignment that coincides with
  $α$ except for $x$, which is mapped to $v$, and analogously,
  $α[X\!↦\!U]$ coincides with $α$ except for $X$, which is mapped to
  $U$.
\end{definition}

We will omit some unnecessary parentheses by following some of the
usual precedence rules for propositional connectives: $\lsymb{¬}$
binds stronger than $\lsymb{∨}$ and $\lsymb{∧}$, which in turn bind
stronger than $\lsymb{⇒}$ and $\lsymb{⇔}$.

\begin{definition}[MSO-Definability]
  The graph language $\Lf{Σ,Γ}(φ)$ \defd{defined} by an
  $\MSO(Σ,Γ)$-sentence $φ$, with respect to $⟨Σ,Γ⟩$, is the set of all
  $Σ$-labeled $Γ$-graphs for which the sentence is satisfied, i.e.,
  \begin{equation*}
    \defd{\Lf{Σ,Γ}(φ)} \coloneqq \bigl\{ G_λ∈Σ^{\clouded{Γ}} \bigm| G_λ⊨φ \bigr\}.
  \end{equation*}
  Every graph language that is defined by some MSO-sentence is called
  \defd{MSO-definable}. We denote by \defd{$\LL_\MSO$} the class of
  all such graph languages.
\end{definition}

An $\MSO(Σ,Γ)$-sentence $φ$ is \defd{equivalent} to an ADGA $\A$ over
$⟨Σ,Γ⟩$ if it defines the same graph language as $\A$ recognizes,
i.e., $\Lf{Σ,Γ}(φ)=\L(\A)$.

We now revisit two of the graph languages considered in
\cref{chap:adga}, and define them by MSO-sentences.

\begin{example}[Translation of $\sA{centric}$ to MSO-Logic] \label{ex:A_centric_MSO}
  We fix $Σ=\{\a,\b,\c\}$ and $Γ=\{\blank\}$. The following
  $\MSO(Σ,Γ)$-sentence $\sphi{centric}$ is equivalent to the ADGA
  $\sA{centric}$ from \cref{fig:ADGA_concentric_circles} (see
  \cref{ex:A_centric_language} for a discussion of the recognized
  graph language).
  \newcommand{\va}{v_{\hspace{-.2ex}\meta{\a}}}
  \begin{align*}
    \sphi{centric} \coloneqq \;
    &\logic{∀u,v\biggl( u\!\arr\!v
      \;\,⇒\; ¬\bigl(\lab{\b}u∧\lab{\b}v\bigr) ∧ 
      ¬\bigl(\lab{\c}u∧\lab{\c}v\bigr) \!\biggr) \;∧} \\[-.9ex]
    &\logic{∃\va\biggl( ∀u\Bigl(\bigl(\lab{\a}u \:\!⇔\:\! u\!=\!\va\bigr)
      \,∧\,\bigl(u\!\arr\!\va ∨ \va\!\arr\!u \,⇒\, \lab{\b}u\bigr) \Bigr) \;∧} \\[-1.7ex]
    &\hspace{20ex}\logic{∃u_1,u_2\Bigl(
      u_1\!\arr\!\va \,∧\, u_2\!\arr\!\va \,∧\, ¬\,u_1\:\!\!\!=\!u_2 \Bigr) \biggr)}
  \end{align*}
  The first line ensures that no two adjacent nodes are both
  $\b$-labeled or both $\c$-labeled. The other two lines specify the
  existence of the “center”, a node $v_\a$ such that
  \begin{itemize}[topsep=1ex,itemsep=0ex]
  \item $v_\a$ is the only $\a$-labeled node in the graph and has only
    $\b$-labeled nodes in its undirected neighborhood (second line),
    and
  \item $v_\a$ has at least two distinct incoming neighbors (third
    line).
  \end{itemize}
\end{example}

In the preceding example, we did not exploit the possibility of
quantifying over set variables. This makes $\sphi{centric}$ a
\emph{first}-order formula. The next example (slightly adapted from
\cite{CE12}) shows an application that requires second-order
quantification.

\begin{example}[3-Colorability] \label{ex:MSO_3_colorable}
  Let $Σ=Γ=\{\blank\}$. The following $\MSO(Σ,Γ)$-sentence
  $\dphi[color]{3}$ defines (with respect to $⟨Σ,Γ⟩$) the language of
  3-colorable graphs. It is thus equivalent to the ADGA
  $\dA[color]{3}$ from \cref{fig:ADGA_3_colorable}.
  \begin{align*}
    \hspace{-\mathindent}
    \dphi[color]{3} \coloneqq \logic{∃U_\mpik,U_\mherz,U_\mkreuz\biggl(}
    &\logic{∀u\Bigl(\bigl(u∈U_\mpik ∨ u∈U_\mherz ∨ u∈U_\mkreuz\bigr)
      \,∧\, ¬\bigl(u∈U_\mpik ∧ u∈U_\mherz\bigr) \,∧ \vphantom{\biggl(} } \hspace{-5ex} \\[-1.5ex]
    &\hspace{5ex}\logic{¬\bigl(u∈U_\mpik ∧ u∈U_\mkreuz\bigr)
      \,∧\, ¬\bigl(u∈U_\mherz ∧ u∈U_\mkreuz\bigr)\:\!\Big) \;∧} \\[-.9ex]
    &\logic{∀u,v\Big(u\!\arr\!v
      \;\,⇒\;\, ¬\bigl(u∈U_\mpik∧v∈U_\mpik\bigr) \,∧} \\[-1.5ex]
    &\hspace{11.2ex}\logic{¬\bigl(u∈U_\mherz∧v∈U_\mherz\bigr)
      \,∧\, ¬\bigl(u∈U_\mkreuz∧v∈U_\mkreuz\bigr)\:\!\Bigr)\,\biggr)} \hspace{-5ex}
  \end{align*}
  The existentially quantified set variables $\lsymb{U_\mpik}$,
  $\lsymb{U_\mherz}$ and $\lsymb{U_\mkreuz}$ represent the three
  possible colors. In the first two lines, we specify that the sets
  assigned to these variables form a partition of the set of nodes
  (possibly with empty components). The remaining two lines constitute
  the actual definition of a valid coloring: no two adjacent nodes
  share the same color, which means that adjacent nodes are in
  different sets.
\end{example}

Similarly to \cref{ex:A_centric_MSO,ex:MSO_3_colorable}, we could
translate all the ADGAs seen in the examples of \cref{chap:adga} to
equivalent MSO-sentences. In the next section, we generalize this for
arbitrary ADGAs.

\section{Equivalence of MSO-Logic and ADGAs} \label{sec:adga=mso}
We can now formally state and prove our main theorem.

\begin{theorem}[$\LL_\ADGA=\LL_\MSO$] \label{thm:adga=mso}
  A graph language is ADGA-recognizable \Iff it is
  MSO-definable. There are effective translations in both directions.
\end{theorem}

We divide the proof of \cref{thm:adga=mso} into
\cref{lem:adga<mso,lem:adga>mso}, which correspond to the two
directions of translation.

\begin{lemma}[$\LL_\ADGA⊆\LL_\MSO$] \label{lem:adga<mso}
  For every ADGA $\A$ over $⟨Σ,Γ⟩$, we can effectively construct an
  $\MSO(Σ,Γ)$-sentence $\phiA$ that is equivalent to $\A$, i.e.,
  \begin{equation*}
    \Lf{Σ,Γ}(\phiA)=\L(\A).
  \end{equation*}
\end{lemma}

\begin{proof}
  Let $\A=⟨Σ,Γ,\Q,σ,δ,\F⟩$. We have to construct an
  $\MSO(Σ,Γ)$-sentence $\phiA$, such that any $G_λ∈Σ^{\clouded{Γ}}$
  satisfies $\phiA$ \Iff it is accepted by $\A$. Thus, $\phiA$ must
  somehow encode the existence of an accepting run of $\A$ on $G_λ$. A
  direct approach might seem tricky at first, since a run is a
  nontrivial graph itself, whose nodes are not in the domain of
  discourse that is referred to from within $\phiA$. We can simplify
  the problem by taking again the game-theoretic point of view
  introduced in \cref{sec:game-theo}. By \cref{lem:acc-win}, $\A$
  accepts $G_λ$ \Iff the automaton has a winning strategy in the
  associated game $J=\J(\A,G_λ)$. We will exploit this equivalence,
  and construct an MSO-sentence expressing that the automaton can win
  the game, no matter how the pathfinder chooses to play.

  Throughout this proof, we use the abbreviations
  $n\coloneqq\len(\A)$, for the length of the automaton, and
  $Q_i\coloneqq\lev_i(\A)∪Q_\P$, for the set of states that may occur
  in a configuration reachable by $\A$ in round $i$, where
  $0≤i≤n$. Our sentence $\phiA$ will contain the set variables
  $\lsymb{U_{\meta{i},\meta{q}}}∈\Vset$, for every round $i≥1$ and
  state $q∈Q_i$. The intended meaning of the subformula
  \noheight{$\logic{v∈U_{\meta{i},\meta{q}}}$} is the following:
  given a prefix $π=G_{κ_0}\cdots G_{κ_i}$ of a play in $J$, the node
  $v∈\VG$ assigned to the variable $\lsymb{v}∈\Vnode$ is in state $q$
  in round $i$, i.e., $κ_i(v)=q$. An entire play will thus be
  represented by an assignment to the set variables
  $\lsymb{U_{\meta{i},\meta{q}}}$. Note that we do not need set
  variables for round $0$, since the initial configuration
  $G_{κ_0}=G_{σ∘λ}$ is always the same. We abbreviate by $\wh{X}_i$
  the list of set variables for round $i$, i.e., $\wh{X}_0 \coloneqq
  ⟨\,⟩$ and $\wh{X}_i \coloneqq
  ⟨\lsymb{U_{\meta{i},\meta{q}}}⟩_{q∈Q_i}$, for $1≤i≤n$.

  \newcommand{\Fstate}[2]{\text{$\varphi_{#1\colon\! #2 \vphantom{i}}^{\hspace{.05ex}\textnormal{state\vphantom{g}}}$}}
  \newcommand{\Fneigh}[2]{\text{$\varphi_{#1\colon\! #2 \vphantom{i}}^{\hspace{.05ex}\textnormal{neigh\vphantom{g}}}$}\hspace{.05ex}}
  \newcommand{\Flegal}[1]{\text{$\varphi_{#1 \vphantom{i}}^{\hspace{.05ex}\textnormal{legal\vphantom{g}}}$}\hspace{-.1ex}}
  \newcommand{\Fwin}[1]{\text{$\varphi_{#1 \vphantom{i}}^{\hspace{.05ex}\textnormal{win\vphantom{g}}}$}\hspace{-.1ex}}
  \newcommand{\Shiddenhat}{\raisebox{0ex}[1.3ex][0ex]{$\S$}}
  \newcommand{\ShiddenhatScr}{\raisebox{0ex}[1.3ex][0ex]{$\scriptstyle \S$}}

  We now construct $\phiA$ bottom-up, starting with the simple
  building block mentioned above. For $0\!≤\!i\!≤\!n$,\;
  $q\!∈\!Q_i$,\:\! and $x\!∈\!\Vnode$, the subformulas
  $\Fstate{i}{q}[x,\wh{X}_i]$ express that in round $i$, the node
  assigned to $x$ is in state $q$. They are
  defined by \\[-1.5ex]
  \begin{equation*}
    \Fstate{i}{q}[x,\wh{X}_i] \;\coloneqq\;
    \begin{cases}
      \logic{\displaystyle \hspace{-2ex}\bigvee_{\hspace{2ex}\meta{a∈Σ\colon\! σ(a)=q}} \hspace{-5ex}\lab{\meta{a}}\meta{x}\hspace{1.3ex}}
      & \text{if\, $i=0$}, \\[4ex]
      \logic{\:\meta{x}∈U_{\meta{i},\meta{q}} \hspace{2.95ex}} & \text{if\, $1≤i≤n$}.
    \end{cases}
  \end{equation*}

  Next, for\:\! $0\!≤\!i\!≤\!n$,\; $q\!∈\!Q_i$,\:\! and\,
  $\Shiddenhat\!=\!⟨S_γ⟩_{γ∈Γ}\!∈\!(2^{Q_i})^Γ$\!, the
  \mbox{subformulas}
  $\Fneigh{i}{q,\ShiddenhatScr}[\lsymb{v},\wh{X}_i]$ express that in
  round $i$, the node assigned to $\lsymb{v}$ receives the information
  $⟨q,\Shiddenhat⟩$ from its closed incoming neighborhood. This is
  ensured by checking that the local state is $q$, and for each $γ∈Γ$,
  every state in $S_γ$ is seen on some $γ$-edge, and any state seen on
  a $γ$-edge is in $S_γ$:
  \begin{align*}
    \Fneigh{i}{q,\ShiddenhatScr}[\lsymb{v},\wh{X}_i] \;\coloneqq\;\;
    &\logic{\meta{\Fstate{i}{q}[\lsymb{v},\wh{X}_i]} \;\,
      ∧\; \smashoperator{\bigwedge_{\meta{γ∈Γ,\:p∈S_γ}}}\, ∃u\Bigl(\meta{\Fstate{i}{p}[\lsymb{u},\wh{X}_i]}
      \,∧\, u\xarr{\meta{γ}}v\Bigr) \vphantom{\biggr)} } \\[-1ex]
    &\phantom{\Fstate{i}{q}[\lsymb{v},\wh{X}_i] \;\,}
      \logic{\:∧\;\, ∀u \biggl(\, \displaystyle\bigwedge_{\meta{γ∈Γ}} \!\Bigl( u\xarr{\meta{γ}}v \,⇒\:\!
      \smashoperator{\bigvee_{\meta{p∈S_γ}}} \meta{\Fstate{i}{p}[\lsymb{u},\wh{X}_i]} \Bigr)\biggr)}.
  \end{align*}

  With building blocks for these local properties at our disposal, we
  can now proceed to more global statements.  For the remainder of
  this proof, we set $G_{κ_0}=G_{σ∘λ}$, and for $1≤i≤n$, we refer by
  $G_{κ_i}$ to the configuration represented by some given assignment
  to the set variables in $\wh{X}_i$. 

  The meaning of the subformula $\Flegal{i}[\wh{X}_{i-1},\wh{X}_i]$ is
  that $G_{κ_i}$ is a legal successor configuration of
  $G_{κ_{i-1}}$. Two properties have to be checked: On the one hand,
  it must hold that $G_{κ_i}∈δ^\cloud(G_{κ_{i-1}})$, or equivalently,
  that for every node $v∈\VG$, if $v$ receives the information
  $⟨p,\Shiddenhat⟩$ in round $i-1$, then it is in some state
  $q∈δ(p,\Shiddenhat)$ in round $i$. On the other hand, the given
  assignment to the set variables in $\wh{X}_i$ must indeed represent
  a valid configuration, which in particular means that a node cannot
  be in several states at once. This leads to the definition
  \newcommand{\PossibleNeighborhoods}{ \hspace{3ex}
    \begin{subarray}{l}
      \swr{\scriptstyle p\:\!}{\S}∈Q_{i-\one}\:\!, \\
      \S∈(\two^{Q_{i-\one}})^Γ
    \end{subarray}
  }
  \begin{multline*}
    \Flegal{i}[\wh{X}_{i-1},\wh{X}_i] \;\coloneqq\;\;
    \logic{∀v \Biggl( \quad\:\!
      \smashoperator{\bigwedge_{\meta{\PossibleNeighborhoods}}}\: \Bigl( \meta{\Fneigh{i-1}{p,\ShiddenhatScr}[\lsymb{v},\wh{X}_{i-\one}]} \,⇒\:\!
      \smashoperator{\bigvee_{\meta{q∈δ(p,\S)}}} \meta{\Fstate{i}{q}[\lsymb{v},\wh{X}_i]} \Bigr)} \\[-1.3ex]
    \hspace{31ex}
    \logic{\,\,∧\;\, \smashoperator{\bigwedge_{\meta{q,r∈Q_i\colon q≠r}}}
      ¬\Bigl(\meta{\Fstate{i}{q}[\lsymb{v},\wh{X}_i]}\,∧\,\meta{\Fstate{i}{r}[\lsymb{v},\wh{X}_i]} \Bigr)\; \Biggr)}.
  \end{multline*}

  We now come to our goal of expressing that the automaton has a
  winning strategy in $J$. As in the proof of \cref{lem:either-win},
  we consider every position $G_κ$ in $J$ as the starting position of
  a subgame $J_κ$, consisting of $G_κ$ and all its descendant
  configurations. For every round $i$, we construct a subformula
  $\Fwin{i}\:\![\wh{X}_i]$ expressing that the automaton has a winning
  strategy in the subgame $J_{κ_i}$.

  In the last round $n$, the reached configuration can only be
  permanent, i.e., $G_{κ_n}∈(Q_\P)^{\clouded{Γ}}$. Hence, the
  automaton has a winning strategy in $J_{κ_n}$ \Iff $G_{κ_n}$ is
  accepting. We check that there is an accepting set $F∈\F$, such that
  each state $q∈F$ occurs in $G_{κ_n}$, and only such states occur:
  \begin{equation*}
    \Fwin{n}[\wh{X}_n] \;\coloneqq\; \logic{\displaystyle \bigvee_{\meta{F∈\F}}\! \Biggl( \, \smashoperator[r]{\bigwedge_{\meta{q∈F}}}
      ∃v\Bigl(\meta{\Fstate{n}{q}[\lsymb{v},\wh{X}_n]}\Bigr)
      \:∧\; ∀v\Bigl(\smashoperator[r]{\bigvee_{\meta{q∈F}}}\meta{\Fstate{n}{q}[\lsymb{v},\wh{X}_n]} \Bigr) \Biggr)}.
  \end{equation*}

  Working our way backwards, we recursively define the formulas for
  previous rounds $i-1$, where $n≥i≥1$. We have to distinguish two
  cases. If level $i-1$ is existential in $\A$, then the automaton is
  the player who has to make a move from position $G_{κ_{i-1}}$. Thus,
  it has a winning strategy in $J_{κ_{i-1}}$ \Iff there is a legal
  successor configuration $G_{κ_i}$ of $G_{κ_{i-1}}$, for which the
  automaton has a winning strategy in the corresponding subgame
  $J_{κ_i}$. This is expressed by
  \begin{equation*}
    \Fwin{i-1}\:\![\wh{X}_{i-\one}] \;\coloneqq\; \logic{\displaystyle\bigexists \meta{\wh{X}_i}
      \biggl(\meta{\Flegal{i}[\wh{X}_{i-\one},\wh{X}_i]}
      \;∧\; \meta{\Fwin{i}[\wh{X}_i]} \biggr)}.
  \end{equation*}
  Otherwise, level $i-1$ is universal, which means that the pathfinder
  has to make a move. Then the automaton has a winning strategy in
  $J_{κ_{i-1}}$ \Iff it has a winning strategy in every subgame that
  starts at a legal successor configuration of $G_{κ_{i-1}}$. The
  corresponding formula is analogous to the previous one:
  \begin{equation*}
    \Fwin{i-1}\:\![\wh{X}_{i-\one}] \;\coloneqq\; \logic{\displaystyle\bigforall \meta{\wh{X}_i}
      \biggl(\meta{\Flegal{i}[\wh{X}_{i-\one},\wh{X}_i]}
      \:⇒\: \meta{\Fwin{i}[\wh{X}_i]} \biggr)}.
  \end{equation*}
  Note that these formulas also cover the cases where a permanent
  configuration is reached earlier than round $n$. If $G_{κ_i}$ is
  permanent, for some $i<n$, then $δ^\cloud(G_{κ_i})=\{G_{κ_i}\}$,
  which means that $\Flegal{i+1}[\wh{X}_i,\wh{X}_{i+1}]$ is satisfied
  precisely when $κ_{i+1}=κ_i$. Proceeding inductively, we get that,
  regardless of whether level $i$ is existential or universal,
  $\Fwin{i}[\wh{X}_i]$ is satisfied \Iff $\Fwin{n}[\wh{X}_n]$ is
  satisfied when interpreting the set variables in $\wh{X}_n$ such
  that $κ_n=κ_i$. In other words, plays of length less than $n$ are
  implicitly extended to length $n$ by repeating the last
  configuration, and consequently the acceptance condition is always
  checked using the subformula $\Fwin{n}[\wh{X}_n]$.

  We have thus achieved our goal. Since the subgame $J_{κ_0}$ is equal
  to $\J(\A,G_λ)$, the desired MSO-sentence is
  \begin{equation*}
    \phiA \coloneqq\, \Fwin{0}[\wh{X}_0] = \Fwin{0}[\:].
    \qedhere
  \end{equation*}
\end{proof}

\hypertarget{engelfriet-charact}{} 
In \cite{Eng91}, Engelfriet characterized the class of MSO-definable
graph languages as the smallest class that contains certain elementary
graph languages\footnote{According to Engelfriet's definition, a graph
  language $L$ is elementary \Iff there are nonempty finite alphabets
  $Σ$ and $Γ$, such that either $L=Σ^{\clouded{Γ}}$, or $L=\bigl\{
  G_λ∈Σ^{\clouded{Γ}} \bigm| ∃u,v∈\VG\colon λ(u)=a \,∧\, u\arrG{γ}v
  \,∧\, λ(v)=b \bigr\}$, for some fixed $a,b∈Σ$ and $γ∈Γ$.} and is
closed under boolean set operations and under projection. His
elementary graph languages can be easily recognized by ADGAs, and
thus, together with \cref{thm:closure} and \cref{lem:adga<mso}, this
characterization of $\LL_\MSO$ implies our main theorem.
Nevertheless, we give a self-contained proof of the following lemma,
in order to provide a direct translation from MSO-formulas to
ADGAs. Some of the ideas are adapted from \cite{Eng91}.

\begin{lemma}[$\LL_\ADGA⊇\LL_\MSO$] \label{lem:adga>mso}
  For every $\MSO(Σ,Γ)$-sentence $φ$, we can effectively construct an
  ADGA $\A_φ$ over $⟨Σ,Γ⟩$ that is equivalent to $φ$, i.e.,
  \begin{equation*}
    \L(\A_φ)=\Lf{Σ,Γ}(φ).
  \end{equation*}
\end{lemma}

\begin{proof}
  It seems natural to prove the claim by induction on the structure of
  $\MSO(Σ,Γ)$-formulas. This forces us to deal with formulas
  containing free occurrences of variables. The truth of such a
  formula $φ$ is evaluated with respect to a labeled graph
  $G_λ∈Σ^{\clouded{Γ}}$ and a variable assignment\, $α\colon
  \free(φ)→\VG ∪ 2^\VG$. We have thus to represent $⟨G_λ,α⟩$ as a
  valid input for an ADGA. This can be done by encoding $α$ into the
  node labels. To this end, we define the inverse function $α^{-1}$ as
  the labeling that assigns to each node $v∈\VG$ the set of variables
  that $α$ associates with $v$, i.e.,
  \begin{align*}
    α^{-1} \colon\: \VG &→ 2^{\free(φ)} \\
    v &↦ \bigl\{x∈\Vnode \bigm| v=α(x)\bigr\} ∪ \bigl\{X∈\Vset \bigm| v∈α(X)\bigr\}.
  \end{align*}
  With this, $⟨G_λ,α⟩$ can be represented as the labeled graph
  $G_{λ×α^{-1}}$ whose labeling is given by
  \begin{align*}
    λ\!×\!α^{-1} \colon\: \VG &→ Σ×2^{\free(φ)} \\
    v &↦ \bigl\langle λ(v),\:α^{-1}(v) \bigr\rangle.
  \end{align*}
  Using this encoding, we generalize the claim of the lemma as
  follows: For any $\MSO(Σ,Γ)$-formula $φ$, there is an effectively
  constructible ADGA $\A_φ$, such that for every $G_λ∈Σ^{\clouded{Γ}}$
  and variable assignment\, $α\colon\!\free(φ)→\VG ∪ 2^\VG$,
  \begin{equation*}
    G_{λ×α^{-1}}∈\L(\A_φ) \quad \text{\Iff} \quad\! ⟨G_λ,α⟩ ⊨ φ.
  \end{equation*}
  If $φ$ is a sentence, i.e., if\, $\free(φ)=∅$, we identify
  $G_{λ×α^{-1}}$ and $⟨G_λ,α⟩$ with $G_λ$. Hence, the statement above
  does indeed imply the lemma.

  We now prove the generalized claim by structural induction on
  $φ$. In each case, we construct a suitable ADGA
  $\A_φ=⟨Σ\!×\!2^{\free(φ)},Γ,\Q,σ,δ,\F⟩$.
  \begin{itemize}
  \item[(\texttt{BC})] We start with the base cases, in which $φ$ is
    an atomic formula.

    For $b∈Σ$,\: $x,y∈\Vnode$ and $X∈\Vset$, the truth of the formulas
    \noheight{$\logic{\lab{\meta{b}}\meta{x}}$},\:
    $\logic{\meta{x}=\meta{y}}$ and $\logic{\meta{x}∈\meta{X}}$
    can be evaluated locally by an ADGA, i.e., without communication
    between the nodes. If $φ$ is equal to such a formula, we define
    the states of $\A_φ$ as
    \begin{equation*}
      Q_\EE=∅, \quad Q_\AA=∅ \quad \text{and} \quad Q_\P=\{\qyes,\qno,\qmaybe\}.
    \end{equation*}
    The intention here is that the node assigned to $x$ (or $y$) will
    answer by “yes” or “no”, while the other nodes remain
    undecided. The automaton then accepts the input \Iff the affected
    node answers positively, i.e.,
    \begin{equation*}
      \F = \bigl\{\{\qyes\},\{\qyes,\qmaybe\}\bigr\}.
    \end{equation*}
    Since all the states are permanent, the transition function is
    already defined implicitly. It remains to specify, for each case,
    the initialization function which directly computes the answer of
    each node. For every $⟨a,M⟩ ∈ Σ×2^{\free(φ)}$,
    \begin{itemize}[topsep=1ex,itemsep=1ex]
    \item if\, $φ =\, \logic{\lab{\meta{b}}\meta{x}}$,\, then \\
      \phantom{\quad} $σ(⟨a,M⟩)=
      \begin{cases}
        \qyes & \text{if  $a=b$ and $M=\{x\}$}, \\
        \qno & \text{if $a≠b$ and $M=\{x\}$}, \\
        \qmaybe & \text{otherwise},
      \end{cases}$
    \item if\, $φ =\, \logic{\meta{x}=\meta{y}}$,\, then \\
      \phantom{\quad} $σ(⟨a,M⟩)=
      \begin{cases}
        \qyes & \text{if $M=\{x,y\}$}, \\
        \qno & \text{if $M=\{x\}$ or $M=\{y\}$,\: with $x≠y$}, \\
        \qmaybe & \text{otherwise},
      \end{cases}$
    \item if\, $φ =\, \logic{\meta{x}∈\meta{X}}$,\, then \\
      \phantom{\quad} $σ(⟨a,M⟩)=
      \begin{cases}
        \qyes & \text{if $M=\{x,X\}$}, \\
        \qno & \text{if $M=\{x\}$}, \\
        \qmaybe & \text{otherwise}.
      \end{cases}$
    \end{itemize}

    The last possible base case is when $φ =
    \logic{\meta{x}\xarr{\meta{τ}}\meta{y}}$,\, with $x,y∈\Vnode$ and
    $τ∈Γ$. To evaluate the truth of such a formula, an ADGA needs one
    communication round, after which the node assigned to $y$ can
    check whether it has received a message from the node assigned to
    $x$ through a $τ$-edge. Then, each node gives a local answer, and
    acceptance is decided as in the previous cases. Accordingly, we
    define the components of $\A_φ$ as follows:
    \begin{itemize}[topsep=1ex,itemsep=.2ex]
    \item $Q_\EE=\{q_x,q_y,q_{x,y}\}$, \quad $Q_\AA=∅$, \quad $Q_\P=\{\qyes,\qno,\qmaybe\}$,
    \item $σ(⟨a,M⟩)=
      \begin{cases}
        q_x & \text{if $M=\{x\}$}, \\
        q_y & \text{if $M=\{y\}$}, \\
        q_{x,y} & \text{if $M=\{x,y\}$,\: with $x≠y$}, \\
        \qmaybe & \text{otherwise},
      \end{cases}$ \\
      for every $⟨a,M⟩ ∈ Σ\!×\!2^{\{x,y\}}$,
    \item $δ(q,\S) =
      \begin{cases}
        \{\qmaybe\} & \text{if $q=q_x$}, \\
        \{\qyes\} & \text{\parbox[t]{.5\textwidth}{\swl{\text{if}}{\text{or}} $q=q_{y\phantom{,x}}$ and $\swr{q_x}{q_{x,y}}∈S_τ$, \\[-.4ex]
            or $q=q_{x,y}$ and $q_{x,y}∈S_τ$,}} \\[2.8ex]
        \{\qno\} & \text{otherwise}
      \end{cases}$ \\
      \quad for every $q∈Q_\N$ and $\S=⟨S_γ⟩_{γ∈Γ}∈(2^Q)^Γ$\!,
    \item $\F = \bigl\{\{\qyes\},\{\qyes,\qmaybe\}\bigr\}$.
    \end{itemize}
    Note that the transition function is deterministic, hence the
    choice of whether a nonpermanent state is existential or universal
    is arbitrary.

  \item[(\texttt{IS})] We now turn to the induction step, for which
    most of the work has already been done by proving the closure
    properties of ADGA-recogniz\-able graph languages
    (\cref{thm:closure}). In the following, let $ψ$, $ψ_1$ and $ψ_2$
    be $\MSO(Σ,Γ)$-formulas that satisfy the induction hypothesis with
    the ADGAs $\A_ψ$, $\A_{ψ_1}$ and $\A_{ψ_2}$, respectively.

    If $φ=\logic{¬\meta{ψ}}$, by \cref{lem:complementation}, it
    suffices to define $\A_φ=\cA_ψ$.

    Similarly, if $φ=\logic{\meta{ψ_\one}∨\meta{ψ_\two}}$, we can use
    the union construction from \cref{lem:union-intersect}. However,
    we must be careful because that construction can only be applied
    on automata that share the same node alphabet. If
    $\free(ψ_1)≠\free(ψ_2)$, we have to extend the node alphabets and
    initialization functions of $\A_{ψ_1}$ and $\A_{ψ_2}$, such that
    each automaton ignores the MSO-variables that are only relevant to
    the other one (as opposed to simply rejecting any input graph that
    contains unknown symbols). For instance, if $⟨a,M⟩$ is a node
    label, and $x∈\free(ψ_2)\setminus\free(ψ_1)$, then the extended
    version of $\A_{ψ_1}$ will initialize a node labeled with
    $⟨a,M\!∪\!\{x\}⟩$ to the same state as one labeled with
    $⟨a,M⟩$. The automaton $\A_φ$ is then obtained by applying the
    union construction on the extended versions of $\A_{ψ_1}$ and
    $\A_{ψ_2}$. We proceed analogously for the case where
    $φ=\logic{\meta{ψ_\one}∧\meta{ψ_\two}}$ (using the intersection
    construction from \cref{lem:union-intersect}), and we reduce
    cases with other logical connectives to the previous ones.

    Next, if $φ=\logic{∃\meta{X}(\meta{ψ})}$, with $X∈\Vset$, we can
    take advantage of the projection construction from
    \cref{lem:projection}. An ADGA can evaluate the truth of $φ$
    by nondeterministically choosing which nodes are in the set
    assigned to $X$, and subsequently simulating $\A_ψ$. We thus
    construct $\A_φ$ by applying the projection construction on
    $\A_ψ$, using the mapping
    \begin{align*}
      h\colon Σ×2^{\free(ψ)} &→ Σ×2^{\free(φ)\:\!\setminus\:\!\{X\}} \\
      ⟨a,M⟩ &↦ ⟨a,M\!\setminus\!\{X\}⟩.
    \end{align*}
    Note that this also works if $X∉\free(ψ)$, since then $h$ is an
    identity function, and consequently $\L(\A_φ)=\L(\A_ψ)$, as
    required.

    Some additional work is needed for the related case in which
    $φ=\logic{∃\meta{x}(\meta{ψ})}$, with $x∈\Vnode$. Like in the
    previous case, a corresponding ADGA can nondeterministically
    choose for each node whether or not it is assigned to $x$. But
    afterwards, it must check that precisely one node has been
    assigned to that variable. We construct a separate ADGA
    $\dA[one]{x}$, specifically for the latter task, and then use it
    as a building block for $\A_φ$. The idea is that any node assigned
    to $x$ can universally choose between two colors. The automaton
    then accepts \Iff exactly one color has been chosen in each
    universal branch. Formally, we define
    $\dA[one]{x}=⟨Σ\!×\!2^{\free(ψ)},Γ,\Q_1,σ_1,δ_1,\F_1⟩$, where
    \begin{itemize}[topsep=1ex,itemsep=.2ex]
    \item $(Q_1)_\EE=∅$, \quad $(Q_1)_\AA=\{q_x\}$, \quad $(Q_1)_\P=\{q_{¬x},\qblack,\qwhite\}$,
    \item $σ_1(⟨a,M⟩)=
      \begin{cases}
        q_x & \text{if $x∈M$}, \\
        q_{¬x} & \text{otherwise},
      \end{cases}$ \quad for every $⟨a,M⟩ ∈ Σ\!×\!2^{\free(ψ)}$,
    \item $δ_1(q_x,\S) = \{\qblack,\qwhite\}$, \quad for every $\S∈(2^{Q_1})^Γ$\!,
    \item $\F_1 = \bigl\{F⊆(Q_1)_\P \bigm| \qblack∈F ⇔ \qwhite∉F\bigr\}$.
    \end{itemize}
    We can now assemble $\A_φ$ by first applying the intersection
    construction on $\A_ψ$ and $\dA[one]{x}$, and then the projection
    construction on the resulting automaton, just as in the previous
    case, with $x$ taking the role of $X$.

    Finally, quantifier duality obviously covers the two cases with
    universal quantifiers, e.g., the formula
    $\logic{∀\meta{x}(\meta{ψ})}$ can be replaced by
    $\logic{¬∃\meta{x}(¬\meta{ψ})}$. It is worth mentioning, however,
    that this indirect approach does not even involve a blow-up of the
    resulting automata, since complementation leaves states and
    transitions unchanged.
    \qedhere
  \end{itemize}
\end{proof}

This concludes the proof of \cref{thm:adga=mso}.

\section{Negative Implications for ADGAs} \label{sec:neg-impl-adga}
We can now take advantage of the equivalence between MSO-logic and
ADGAs to infer some negative results on ADGAs.

The \defd{satisfiability problem} of MSO-logic is the question
whether, for a given $\MSO(Σ,Γ)$-sentence $φ$, there is a labeled
graph $G_λ∈Σ^{\clouded{Γ}}$ that satisfies $φ$, i.e., whether
$\Lf{Σ,Γ}(φ)≠∅$. As remarked in \cite{Cou97,CE12}, the undecidability
of this problem follows directly from Trakhtenbrot's Theorem
\cite{Tra50}, which states that it is undecidable whether a
first-order sentence over a relational vocabulary with at least one
binary relation symbol is \emph{finitely} satisfiable (see, e.g.,
\cite[Thm~9.2]{Lib04}).

\begin{theorem}[Satisfiability Problem]
  The satisfiability problem of MSO-logic (on finite graphs) is
  undecidable.
\end{theorem}

Together with \cref{thm:adga=mso}, we directly obtain the following
corollary concerning the \defd{emptiness problem} of ADGAs. This
problem is the question whether the graph language $\L(\A)$ of a given
ADGA $\A$ is empty.

\begin{corollary}[Emptiness Problem] \label{cor:adga-emptiness}
  The emptiness problem of ADGAs is undecidable.
\end{corollary}

Furthermore, \cref{thm:adga=mso} allows us to state some graph
properties that cannot be recognized by ADGAs. This is a restatement
of a result on MSO-logic proven by Courcelle and Engelfriet in
\cite[Prp~5.13]{CE12}.

\begin{lemma}[ADGA-Unrecognizable Languages]
  Let $Σ=Γ=\{\blank\}$. The following graph languages are \emph{not}
  ADGA-recognizable:
  \begin{itemize}
  \item $\swr{\sL{Ham}}{\sL{morph}} = \bigl\{G∈Σ^{\clouded{Γ}} \bigm|
    \text{$G$ has a Hamiltonian cycle} \bigr\}$,
  \item $\swr{\sL{match}}{\sL{morph}} = \bigl\{G∈Σ^{\clouded{Γ}} \bigm|
    \text{$G$ has a perfect matching} \bigr\}$,\, and
  \item $\sL{morph} = \bigl\{G∈Σ^{\clouded{Γ}} \bigm| \text{$G$ has a
      nontrivial automorphism} \bigr\}$.
  \end{itemize}
\end{lemma}

%% file: chap/NDGAs_DDGAs.tex

\chapter{Nondeterministic and Deterministic DGAs} \label{chap:ndga_ddga}
In this chapter, we consider restrictions on the definition of
ADGAs. It turns out that forbidding universal branchings results in a
loss of expressive power, and additionally forbidding nondeterministic
choices leads to an even weaker class of graph automata. On the other
hand, the emptiness problem becomes decidable, and some closure
properties still hold. Furthermore, as a byproduct of our
investigation, we obtain necessary conditions for recognizability by
those weaker classes of graph automata, loosely similar to pumping
lemmas.

\section{Nondeterministic Distributed Graph Automata}
We start by removing the possibility of universal branching.

\begin{definition}[Nondeterministic Distributed Graph Automaton]
  A \defd{nondeterministic distributed graph automaton} (NDGA) is an
  ADGA $\A=⟨Σ,Γ,\Q,\ab σ,δ,\F⟩$ that has no universal states, i.e.,
  $Q_\AA=∅$. We denote by $\defd{\LL_\NDGA}$ the class of all
  \defd{NDGA-recognizable} graph languages.
\end{definition}

Since the runs of NDGAs do not branch, we can represent them simply as
sequences of configurations of the form $R=G_{κ_0}\!\cdots G_{κ_n}$,
with $n≤\len(\A)$, where $G_{κ_0}$ is the initial configuration, each
$G_{κ_{i+1}}$ is a successor configuration of $G_{κ_i}$, for
$0≤i≤n-1$, and $G_{κ_n}$ is a permanent configuration.

We will compare such sequences from the local point of view of
individual nodes.

\begin{definition*}[Local View]
  Consider a sequence $R=G_{κ_0}\!\cdots G_{κ_n}$ of configurations of
  some NDGA $\A$ on an underlying graph $G$. For each node $v∈\VG$, we
  define the \defd{local view} $R|_v$ of $v$ in $R$ as the sequence of
  informations that $v$ receives from its closed incoming
  neighborhood at each position of $R$, i.e.,
  \begin{equation*}
    \defd{R|_v} \coloneqq ⟨q_0,\S_0⟩\cdots⟨q_n,\S_n⟩,
  \end{equation*}
  where\, $q_i=κ_i(v)$,\, and\, $\S_i=\bigl\langle\{κ_i(u)\mid
  u\arrG{γ}v\}\bigr\rangle_{γ∈Γ}$\:\!,\, for $0≤i≤n$.
\end{definition*}

On several occasions, we will construct a new run from a given one, by
ensuring that every local view in the new run also occurs in the old
run. The following remark formalizes this approach.

\begin{remark} \label{rem:local_view}
  Let $\A=⟨Σ,Γ,\Q,\ab σ,δ,\F⟩$ be an NDGA, $G_λ$ a labeled graph in
  $Σ^{\clouded{Γ}}$, and $R=G_{κ_0}\!\cdots G_{κ_n}$ a run of $\A$ on
  $G_λ$. Consider another labeled graph
  $G'_{λ'}∈Σ^{\clouded{Γ}}$. Suppose we can construct a sequence
  $R'=G'_{κ'_0}\!\cdots G'_{κ'_n}$ of configurations of $\A$ on $G'$,
  such that, for every node $v'∈\VGpr$, there is a node $v∈\VG$ with
  the same label and local view, i.e.,
  \begin{itemize}
  \item $λ(v)=λ'(v')$,\: and
  \item $\swr{R|_v}{λ(v)}=R'|_{v'}$.
  \end{itemize}
  Then $R'$ is a legal run of $\A$ on $G'_{λ'}$. Furthermore, if $R$
  is accepting and the states occurring in $G_{κ_n}$ and $G'_{κ'_n}$
  are the same, i.e., $\{κ_n(v)\mid v∈\VG\} = \{κ'_n(v')\mid
  v'∈\VGpr\}$, then $R'$ is also accepting.
\end{remark}

\begin{proof}
  It is easy to see that $G_{κ_0}=G_{σ∘λ}$ implies
  $G'_{κ'_0}=G'_{σ∘λ'}$, and $G_{κ_{i+1}}∈δ^\cloud(G_{κ_i})$ implies
  $G'_{κ'_{i+1}}\!∈δ^\cloud(G'_{κ'_i})$, for $0≤i≤n-1$, and
  $G_{κ_n}\!∈Q_\P^G$ implies $G'_{κ'_n}\!∈Q_\P^{G'}$. Thus, the
  sequence $R'$ is a legal run of $\A$ on $G'_{λ'}$.
\end{proof}

Next, we want to show that NDGAs are, to a certain extent, blind to
symmetry. To this end, we define a mirroring operation, which
introduces symmetry into a (labeled) graph by duplicating a given
subgraph, together with its connections to the rest of the graph.

\begin{definition}[Graph Mirroring]
  Consider a labeled graph $G_λ∈Σ^{\clouded{Γ}}$ and a subset of nodes
  $U⊆\VG$. Further, let $U'$ be a copy of $U$ that is disjoint from
  $\VG$, and let $f\colon U→U'$ be some bijection. The graph obtained
  by \defd{mirroring} $U$ in $G_λ$, denoted $\defd{\mir(G_λ,U)}$, is
  defined as the labeled graph $G'_{λ'}$, such that
  \begin{itemize}[topsep=1ex,itemsep=0ex]
  \item $\swl{\VGpr}{\arrGpr{γ}} = \VG ∪ U'$,
  \item
    ${\arrGpr{γ}} =
    \begin{aligned}[t]
      \arrG{γ} &∪ \bigl\{⟨u,f(v)⟩ \bigm| u∈(\VG\setminus U) \;∧\; v∈U \;∧\; u\arrG{γ}v \bigr\} \\
               &∪ \bigl\{⟨f(u),v⟩ \bigm| u∈U \;∧\; v∈(\VG\setminus U) \;∧\; u\arrG{γ}v \bigr\} \\
               &∪ \bigl\{⟨f(u),f(v)⟩ \bigm| u,v∈U \;∧\; u\arrG{γ}v \bigr\},
    \end{aligned}$ \\[.5ex]
    for every $γ∈Γ$,
  \item $λ'(v)=λ(v)$\, for every $v∈\VG$,\; and $λ'(f(v))=λ(v)$\, for
    every $v∈U$.
  \end{itemize}
  We call $f$ a \defd{mirroring bijection} between $U$ and $U'$ in
  $G'_{λ'}$, and $f(v)$ a \defd{mirror image} of $v$ in $G'_{λ'}$, for
  every node $v∈U$.
\end{definition}

Note that graph mirroring is well-defined because we consider graphs
only up to isomorphism.

\newpage

\begin{example}
  Let $Σ=\{\a,\b,\c\}$ and $Γ=\{\blank\}$. Mirroring $\{s,t,u\}$ in
  the graph $(G_1)_{λ_1}$ from \cref{fig:graph_1_before_mirroring}
  yields the graph $G_λ$ depicted in
  \cref{fig:graph_1_after_mirroring}. The function $f_1$ indicated in
  that figure is the only possible mirroring bijection between
  $\{s,t,u\}$ and $\{x,y,z\}$ in $G_λ$. We can obtain the same graph
  $G_λ$ by mirroring $\{t,v,x\}$ in the graph $(G_2)_{λ_2}$ from
  \cref{fig:graph_2_before_mirroring}, as shown in
  \cref{fig:graph_2_after_mirroring}. The function $f_2$ is one of two
  possible mirroring bijections between $\{t,v,x\}$ and $\{u,w,y\}$ in
  $G_λ$ (for the other one, the images of $t$ and $x$ are
  swapped). Since there are several ways of obtaining $G_λ$ through
  mirroring, some of the nodes have several mirror images in
  $G_λ$. For instance, $t$ has three of them: $u$, $x$ and $y$.
\end{example}

\begin{figure}
  \alignpic
  \begin{subfigure}{0.38\textwidth}
    \centering
    \input{fig/graph_1_before_mirroring.tikz}
    \caption{\;$(G_1)_{λ_1}$}
    \label{fig:graph_1_before_mirroring}
  \end{subfigure}
  \begin{subfigure}{0.51\textwidth}
    \centering
    \input{fig/graph_1_after_mirroring.tikz}
    \caption{\;$G_λ=\mir\bigl((G_1)_{λ_1},\:\!\{s,t,u\}\bigr)$}
    \label{fig:graph_1_after_mirroring}
  \end{subfigure}

  \vspace{4ex}

  \begin{subfigure}{0.38\textwidth}
    \centering
    \vspace{-9.7ex}
    \input{fig/graph_2_before_mirroring.tikz}
    \caption{\;$(G_2)_{λ_2}$}
    \label{fig:graph_2_before_mirroring}
  \end{subfigure}
  \begin{subfigure}{0.51\textwidth}
    \centering
    \input{fig/graph_2_after_mirroring.tikz}
    \vspace{.5ex}
    \caption{\;$G_λ=\mir\bigl((G_2)_{λ_2},\:\!\{t,v,x\}\bigr)$}
    \label{fig:graph_2_after_mirroring}
  \end{subfigure}
  \caption{Mirroring in $\{\a,\b,\c\}$-labeled $\{\blank\}$-graphs.}
\end{figure}

We can now use the notion of graph mirroring to establish a necessary
condition for NDGA-recognizability.

\begin{lemma}[Mirroring Lemma] \label{lem:weak-mirroring}
  Every NDGA-recognizable graph language $L$ is closed under
  mirroring, i.e., for every labeled graph $G_λ$ and subset of nodes
  $U⊆\VG$,
  \begin{equation*}
    G_λ∈L \quad \text{implies} \quad \mir(G_λ,U)∈L.
  \end{equation*}
\end{lemma}

\begin{proof}
  Let $\A=⟨Σ,Γ,\Q,\ab σ,δ,\F⟩$ be an NDGA. Consider any
  $G_λ∈Σ^{\clouded{Γ}}$ and $U⊆\VG$. We set $G'_{λ'}=\mir(G_λ,U)$, and
  fix a mirroring bijection $f\colon U→(\VGpr\setminus \VG)$ in
  $G'_{λ'}$. If $G_λ∈\L(\A)$, then there must be an accepting run
  $R=G_{κ_0}\!\cdots G_{κ_n}$ of $\A$ on $G_λ$. By
  \cref{rem:local_view}, we can derive from it an accepting run
  $R'=G'_{κ'_0}\!\cdots G'_{κ'_n}$ of $\A$ on $G'_{λ'}$, in which the
  behaviour of every node $v∈\VG$ remains the same, i.e.,
  $κ'_i(v)=κ_i(v)$, for $0≤i≤n$, and every node $v∈U$ is imitated by
  its mirror image under $f$, i.e., $κ'_i(f(v))=κ_i(v)$, for
  $0≤i≤n$. (The local view of the nodes in $(\VG\setminus U)$ does not
  change because they cannot distinguish between nodes that are in the
  same state.) Consequently, $G'_{λ'}∈\L(\A)$.
\end{proof}

\Cref{lem:weak-mirroring} directly yields the following corollary.

\begin{corollary} \label{cor:ndga_infinite}
  Every nonempty NDGA-recognizable graph language is (countably)
  infinite.
\end{corollary}

This implies that we have lost some expressive power by forbidding
universal branchings.

\begin{lemma}[$\LL_\NDGA⊂\LL_\ADGA$] \label{lem:ndga<adga}
  There are (infinitely many) ADGA-recognizable graph languages that
  are not NDGA-recognizable.
\end{lemma}

\begin{proof}
  Let $Σ=Γ=\{\blank\}$. We consider the language $\dL[max]{2}$ of all
  graphs that have at most two nodes, i.e.,
  $\dL[max]{2}=\bigl\{G∈Σ^{\clouded{Γ}} \bigm| \card{\VG}≤2
  \bigr\}$. This language is nonempty and finite, which by
  \cref{cor:ndga_infinite} implies that it is not
  NDGA-recognizable. On the other hand, it is clearly
  ADGA-recognizable, since it is recognized by the ADGA $\dA[max]{2}$,
  specified in \cref{fig:ADGA_max_two}. (The accepting configurations
  of that automaton are those that comprise at most two different
  states.)

  We can apply the same reasoning to any graph language
  $\dL[max]{c}⊆Σ^{\clouded{Γ}}$ in which the number of nodes is
  limited by a constant $c≥1$.
\end{proof}

\begin{figure}[h!]
  \alignpic
  \begin{subfigure}{0.5\textwidth}
    \centering
    \input{fig/ADGA_max_two.tikz}
    \caption{$\dA[max]{2}$}
    \label{fig:ADGA_max_two}
  \end{subfigure}
  \begin{subfigure}{0.45\textwidth}
    \centering
    \input{fig/ADGA_min_three.tikz}
    \caption{$\dA[min]{3}=\dcA[max]{2}$}
    \label{fig:ADGA_min_three}
  \end{subfigure}
  \caption{$\dA[max]{2}$ and $\dA[min]{3}$\!, two ADGAs over
    $\bigl\langle\{\blank\},\{\blank\}\bigr\rangle$ whose graph
    languages consist of the graphs that have at most two nodes, and
    at least three nodes, respectively.}
  \label{fig:ADGA_max_and_min}
\end{figure}

Following this line of thought, we also get that NDGAs cannot, in
general, be complemented.

\begin{lemma}[Complementation] \label{lem:ndga-complementation}
  The class $\LL_\NDGA$ of NDGA-recognizable graph languages is
  \emph{not} closed under complementation.
\end{lemma}

\begin{proof}
  As mentioned in the proof of \cref{lem:ndga<adga}, the language
  $\dL[max]{2}$ of all $\{\blank\}$-labeled $\{\blank\}$-graphs that
  have at most two nodes is not NDGA-recognizable. However, its
  complement, the language $\dL[min]{3}$ of all graphs that have at
  least three nodes, is recognized by the NDGA $\dA[min]{3}$ specified
  in \cref{fig:ADGA_min_three}.
\end{proof}

However, all the other closure properties of ADGAs mentioned in
\cref{sec:closure-properties} are preserved.

\begin{lemma}[Closure Properties] \label{lem:ndga-closure}
  The class $\LL_\NDGA$ of NDGA-recognizable graph languages is
  effectively closed under union, intersection and projection.
\end{lemma}

\begin{proof}
  The union construction from \cref{lem:union-intersect} and the
  projection construction from \cref{lem:projection} do not introduce
  any universal states, and thus yield NDGAs when applied on
  NDGAs.\footnote{Note that those constructions require the input
    NDGAs to satisfy certain properties. Fortunately, they can be
    assumed to hold without loss of generality: The constructions for
    nonblocking and trim ADGAs (see \cref{rem:nonblocking,rem:trim})
    remain valid when restricted to NDGAs. Furthermore, the
    requirement that two automata agree on the sequence of quantifiers
    is trivially fulfilled for NDGAs.}

  It remains to show closure under intersection. This can be done
  using a simple product construction, similar to the one for finite
  automata on words. Consider two NDGAs
  $\A_1=⟨Σ,Γ,\Q_1,σ_1,δ_1,\F_1⟩$ and $\A_2=⟨Σ,Γ,\ab\Q_2,\ab
  σ_2,δ_2,\F_2⟩$. Without loss of generality, we may assume that they
  share the same node and edge alphabets, and that they are both
  nonblocking (see \cref{rem:nonblocking}). For any set $S⊆Q_1×Q_2$,
  we define the projection on the first component as
  \begin{equation*}
    \prj_1(S) \coloneqq \bigl\{ q_1∈Q_1 \bigm| ∃q_2∈Q_2\colon\!⟨q_1,q_2⟩∈S \bigr\},
  \end{equation*}
  and analogously for the projection $\prj_2(S)$ on the second
  component. With this, we construct the product NDGA
  $\A_⊗=⟨Σ,Γ,\Q_⊗,σ_⊗,δ_⊗,\F_⊗⟩$, where
  \begin{itemize}
  \item $(Q_⊗)_\P=(Q_1)_\P×(Q_2)_\P$, \quad $(Q_⊗)_\EE=\bigl(Q_1×Q_2\bigr)\setminus(Q_⊗)_\P$,
  \item
    $\swl{σ_⊗(a)}{(Q_⊗)_\P} =
    \bigl\langle σ_1(a),σ_2(a) \bigr\rangle,
    \quad \text{for every $a∈Σ$}$,
  \item
    $δ_⊗\bigl(⟨q_1,q_2⟩,\:\! \S \:\!\bigr) =\,
    δ_1\bigl( q_1,\:\! \bigl\langle\prj_1(S_γ)\bigr\rangle_{γ∈Γ} \:\!\!\bigr) ×\,
    δ_2\bigl( q_2,\:\! \bigl\langle\prj_2(S_γ)\bigr\rangle_{γ∈Γ} \:\!\!\bigr)$, \\
    for every $⟨q_1,q_2⟩∈Q_⊗$ and $\S=⟨S_γ⟩_{γ∈Γ}∈(2^{Q_⊗})^Γ$\!,
  \item $\F_⊗ = \bigl\{ F⊆(Q_⊗)_\P \bigm| \prj_1(F)∈\F_1 \,∧\, \prj_2(F)∈\F_2 \bigr\}$.
  \end{itemize}
  It is easy to see that $\A_⊗$ recognizes $\L(\A_1)∩\L(\A_2)$.

  Additionally, we also get an alternative construction for the union
  by changing the definition of $\F_⊗$ to
  \begin{equation*}
    \F_⊗ = \bigl\{ F⊆(Q_⊗)_\P \bigm| \prj_1(F)∈\F_1 \,∨\, \prj_2(F)∈\F_2 \bigr\}.
  \end{equation*}
  While this is significantly less efficient than the union
  construction from \cref{lem:union-intersect}, in terms of number of
  states, it has the advantage of not relying on nondeterminism. It
  thus remains applicable when we restrict ourselves to deterministic
  automata in the next section.
\end{proof}

As we have already seen in the Mirroring Lemma
(\cref{lem:weak-mirroring}), all NDGAs have some runs containing
redundancies that prevent them from distinguishing between some nodes
of the input graph. In fact, we can show that \emph{every} run on a
sufficiently large input graph will contain such redundancies. In
order to formally express this idea, we first define two node merging
operations.

Given a graph and two of its nodes $w$ and $w'$, asymmetrically
merging $w'$ into $w$ means to remove $w'$, together with its adjacent
edges, and add outgoing edges from $w$ to all of the former
\emph{outgoing} neighbors of $w'$. If we additionally replicate the
incoming edges of $w'$, the merging is called symmetric.

\begin{definition}[Node Merging]
  Consider a labeled graph $G_λ∈Σ^{\clouded{Γ}}$ and two nodes
  $w,w'∈\VG$. We say that a labeled graph $G'_{λ'}$ is obtained by
  \defd{asymmetric merging} of $w'$ into $w$ in $G_λ$, and denote it
  by $\defd{\amrg(G_λ,w,w')}$, if
  \begin{itemize}[topsep=1ex,itemsep=0ex]
  \item $\swl{\VGpr}{λ'(v)} = \VG \setminus \{w'\}$,
  \item $\swl{\arrGpr{γ}}{λ'(v)} = \bigl({\arrG{γ}}∩(\VGpr\!×\!\VGpr)\bigr)
    \,∪\, \bigl\{⟨w,v⟩\bigm|w'\arrG{γ}v\bigr\}$, \quad for every $γ∈Γ$,
  \item $λ'(v) = λ(v)$, \quad for every $v∈\VGpr$.
  \end{itemize}
  Similarly, if instead of the second condition it holds that
  \begin{equation*}
    {\arrGpr{γ}} = \bigl({\arrG{γ}}∩(\VGpr\!×\!\VGpr)\bigr)
    \,∪\, \bigl\{⟨w,v⟩\bigm|w'\arrG{γ}v\bigr\} \,∪\, \bigl\{⟨v,w⟩\bigm|v\arrG{γ}w'\bigr\},
  \end{equation*}
  for every $γ∈Γ$, then $G'_{λ'}$ is said to be obtained by
  \defd{symmetric merging} of $w$ and $w'$ in $G_λ$, and is denoted by
  $\defd{\smrg(G_λ,w,w')}$.
\end{definition}

With the notation defined, we can now derive another necessary
condition for NDGA-recognizability. The following result reminds
strongly of the Pumping Lemma for regular word languages, and its
proof is somewhat similar in spirit (also based on the Pigeonhole
Principle).

\begin{lemma}[Merging Lemma] \label{lem:merging}
  For every NDGA-recognizable graph language $L$ there exist natural
  numbers $n$ and $m$ (with $n<m$) such that every labeled graph $G_λ$
  satisfies the following node merging properties:
  \begin{itemize}
  \item If $G_λ$ has at least $n$ nodes, then there exist nodes
    $w,w'∈\VG$ such that \\[-2ex]
    \begin{equation*}
      G_λ∈L \quad \text{implies} \quad \amrg(G_λ,w,w')∈L.
    \end{equation*}
  \item If $G_λ$ has at least $m$ nodes, then there exist nodes
    $w,w'∈\VG$ such that \\[-2ex]
    \begin{equation*}
      G_λ∈L \quad \text{implies} \quad \smrg(G_λ,w,w')∈L.
    \end{equation*}
  \end{itemize}
  Moreover, if $\A=⟨Σ,Γ,\Q,\ab σ,δ,\F⟩$ is an NDGA that recognizes
  $L$, then we have $n ≤ \card{Q}^{\len(\A)+1}$,\, and\, $m ≤
  \bigl(\card{Q}·2^{\card{Γ}·\card{Q}}\bigr)^{\!\len(\A)+1}$.
\end{lemma}

\begin{proof}
  We fix an NDGA $\A=⟨Σ,Γ,\Q,\ab σ,δ,\F⟩$ that recognizes $L$ and use
  the abbreviations $g=\card{Γ}$,\, $s=\card{Q}$, and $ℓ=\len(\A)$.

  Note that there cannot be more than $s^{ℓ+1}$ different local
  sequences of states in a run of $\A$. Consider a labeled graph
  $G_λ\!∈\!Σ^{\clouded{Γ}}$ that has more than $s^{ℓ+1}$ nodes, and
  some run $R=G_{κ_0}\!\cdots G_{κ_m}$ of $\A$ on $G_λ$, where
  $m≤ℓ$. By the Pigeonhole Principle, there must be two distinct nodes
  $w,w'∈\VG$ that have the same local sequence of states in $R$. If we
  asymmetrically merge $w'$ into $w$, i.e., if we construct
  $G'_{λ'}=\amrg(G_λ,w,w')$, the remaining nodes in $G'_{λ'}$ will not
  see the difference if they maintain their behaviour from $R$, since
  their local views will remain the same. More formally, by
  \cref{rem:local_view}, we can derive from $R$ a run
  $R'=G'_{κ'_0}\!\cdots G'_{κ'_m}$ of $\A$ on $G'_{λ'}$, such that
  $κ'_i(v)=κ_i(v)$ for every node $v∈\VGpr$ and $0≤i≤m$. If $R$ is
  accepting, so is $R'$, hence $G_λ∈\L(\A)$ implies $G'_{λ'}∈\L(\A)$.

  If we want to symmetrically merge two nodes, the reasoning is very
  similar, but slightly more involved because the merged node inherits
  the unified incoming neighborhood of the original nodes, and
  consequently would get a new local view if the two local views of
  the original nodes were different. A simple solution is to require a
  larger minimum number of nodes. Altogether, there cannot be more
  than $(s\,2^{gs})^{ℓ+1}$ different local views in any run of
  $\A$. Again by the Pigeonhole Principle, if a graph has more than
  $(s\,2^{gs})^{ℓ+1}$ nodes, there must be two distinct nodes that
  have the same local view. The rest of the argument is analogous to
  the previous scenario.
\end{proof}

While the previously seen Mirroring Lemma (\cref{lem:weak-mirroring})
allows us to enlarge (“pump up”) graphs without leaving a given
NDGA-recognizable graph language, the Merging Lemma allows us to
shrink (“pump down”) some of them. The combination of both could thus
be considered as some sort of “graph pumping lemma”.

The Merging Lemma provides further evidence of the expressive weakness
of NDGAs, but, perhaps more importantly, it also tells us that their
emptiness problem is decidable. The daunting time complexities
indicated in the following lemma are only rough upper bounds. Better
estimates and algorithms can hopefully be found through further
investigation.

\begin{lemma}[Emptiness Problem] \label{lem:ndga-emptiness}
  The emptiness problem of NDGAs is decidable in doubly-exponential
  time. More precisely, for every NDGA $\A=⟨Σ,Γ,\Q,\ab σ,δ,\F⟩$,
  \begin{itemize}
  \item whether its recognized graph language $\L(\A)$ is empty or not
    can be decided in time $2^k$\!, where $k∈\O\bigl(\:\!
    \card{Γ}·\card{Q}^{4\len(\A)}·\len(\A) \bigr)$,\, and
  \item whether its undirected graph language $\L_\undir(\A)$ is empty
    or not can be decided in time $2^{2^{k'}}$\!\!, where
    $k'∈\O\bigl(\:\!  \card{Γ}·\card{Q}·\len(\A) \bigr)$.
  \end{itemize}
\end{lemma}

\begin{proof}
  We use again the abbreviations $g=\card{Γ}$,\, $s=\card{Q}$, and
  $ℓ=\len(\A)$. By applying the first part of the Merging Lemma
  (\cref{lem:merging}) recursively, we conclude that if $\L(\A)$ is
  not empty, then it contains a labeled graph that has at most
  $s^{ℓ+1}$ nodes. Similarly, by the second part of the Merging Lemma,
  if $\L_\undir(\A)$ is not empty, then it contains an undirected
  labeled graph that has at most $(s\,2^{gs})^{ℓ+1}$ nodes. (For
  undirected graphs, asymmetric merging is, in general, not
  applicable.) Hence, the emptiness problem is decidable because the
  search space is finite.

  We now derive a rough asymptotic upper bound on the time
  complexities of the naive approaches that check every (directed)
  graph that has at most $s^{ℓ+1}$ nodes, and every undirected graph
  that has at most $(s\,2^{gs})^{ℓ+1}$ nodes, respectively.
  \begin{itemize}
  \item The maximum numbers of nodes can be over-approximated by
    $\O(s^{2ℓ})$ and $\O(2^{4gsℓ})$, respectively.
  \item Given a natural number $n$, there are $\O(2^{gn^2})$
    $Γ$-graphs with precisely $n$ nodes. (This is only an upper bound
    because we consider isomorphic graphs to be equal.)
  \item Given a $Γ$-graph $G$ with $n$ nodes, we can decide whether
    $G_λ∈\L(\A)$ for some labeling $λ\colon \VG→Σ$, by checking every
    possible run of $\A$ that starts with a configuration
    $G_{κ_0}∈\bigl(σ(Σ)\bigr)^G$. This can be done in time
    $\O(s^{n(ℓ+1)})⊆\O(s^{2nℓ})$.
  \end{itemize}
  Hence, the total time complexities are bounded by
  \begin{align*}
    &\O\Biggl(\: \sum_{n=1}^{s^{2ℓ}} 2^{gn^2} \!· s^{2nℓ} \Biggr)
    \,⊆\; \O\biggl(\! 2^{8gs^{4ℓ}ℓ} \biggr), \quad \text{and} \\
    &\O\Biggl(\: \sum_{n=1}^{2^{4gsℓ}} 2^{gn^2} \!· s^{2nℓ} \Biggr)
    \,⊆\; \O\biggl(\! 2^{2^{16gsℓ\strut}} \biggr),
  \end{align*}
  respectively.
\end{proof}

\section{Deterministic Distributed Graph Automata}
As a further restriction, we now forbid nondeterministic choices.

\begin{definition}[Deterministic Distributed Graph Automaton]
  A \defd{deterministic distributed graph automaton} (DDGA) is a
  \emph{nonblocking} NDGA $\A=⟨Σ,Γ,\Q,\ab σ,δ,\F⟩$ in which every
  state $q∈Q$ has at most one outgoing transition for every
  $\S∈(2^Q)^Γ$\!, i.e., $\card{δ(q,\S)}≤1$. We denote by
  $\defd{\LL_\DDGA}$ the class of all \defd{DDGA-recognizable} graph
  languages.
\end{definition}

The transition function being deterministic forces every node of an
input graph to behave like its mirror images. This allows us to state
a stronger Mirroring Lemma for DDGAs.

\begin{lemma}[Strong Mirroring Lemma] \label{lem:strong-mirroring}
  Every DDGA-recognizable graph language $L$ is closed under both
  mirroring and the reversal of mirroring, i.e., for every labeled
  graph $G_λ$ and subset of nodes $U⊆\VG$,
  \begin{equation*}
    G_λ∈L \quad \text{\Iff} \quad\! \mir(G_λ,U)∈L.
  \end{equation*}
\end{lemma}

\begin{proof}~
  \vspace{-1ex}
  \begin{itemize}
  \item[($⇒$)] The “only if” direction is a specialization of the
    weaker Mirroring Lemma (\cref{lem:weak-mirroring}) to DDGAs.
  \item[($⇐$)] Let $\A=⟨Σ,Γ,\Q,\ab σ,δ,\F⟩$ be a DDGA. Consider any
    $G_λ∈Σ^{\clouded{Γ}}$ and $U⊆\VG$. We set $G'_{λ'}=\mir(G_λ,U)$,
    and fix a mirroring bijection $f\colon U→(\VGpr\setminus \VG)$ in
    $G'_{λ'}$. If $G'_{λ'}∈\L(\A)$, then the (unique) run
    $R'=G'_{κ'_0}\!\cdots G'_{κ'_n}$ of $\A$ on $G'_{λ'}$ is
    accepting. Since $\A$ is deterministic, it can be shown
    inductively that every node $v∈U$ behaves like its mirror image
    $f(v)$, i.e., $κ'_i(v)=κ'_i(f(v))$, for $0≤i≤n$. If we remove
    $f(v)$ for every $v∈U$, we do not change the local view of any
    node in $\VG$. Hence, by \cref{rem:local_view}, we can derive from
    $R'$ the accepting run $R=G_{κ_0}\!\cdots G_{κ_n}$ of $\A$ on
    $G_λ$, where $κ_i(v)=κ'_i(v)$, for every $v∈\VG$ and
    $0≤i≤n$. Consequently, $G_λ∈\L(\A)$.  \qedhere
  \end{itemize}
\end{proof}

The Strong Mirroring Lemma implies that DDGAs are utterly incapable of
breaking symmetry. This makes them strictly weaker than NDGAs.

\begin{lemma}[$\LL_\DDGA⊂\LL_\NDGA$] \label{lem:ddga<ndga}
  There are (infinitely many) NDGA-recognizable graph languages that
  are not DDGA-recognizable.
\end{lemma}

\begin{proof}
  Let $Σ=Γ=\{\blank\}$. We consider again the language
  $\dL[min]{3}=\bigl\{G∈Σ^{\clouded{Γ}} \bigm| \card{\VG}≥3 \bigr\}$
  of all graphs that have at least three nodes. As already mentioned
  in the proof of \cref{lem:ndga-complementation}, this language is
  NDGA-recognizable because it is recognized by the NDGA $\dA[min]{3}$
  from \cref{fig:ADGA_min_three}. Now, assume that $\dL[min]{3}$ is
  also DDGA-recognizable and consider the graphs $G$ and
  $G'=\mir(G,\{u\})$ from
  \cref{fig:graphs_two_and_three_nodes}. Clearly $G'∈\dL[min]{3}$, and
  by \cref{lem:strong-mirroring}, it follows that $G∈\dL[min]{3}$,
  which is a contradiction. Hence, $\dL[min]{3}$ is not
  DDGA-recognizable.

  We can apply a similar reasoning to any graph language
  $\dL[min]{c}⊆Σ^{\clouded{Γ}}$ in which the number of nodes is
  bounded from below by a constant $c≥2$.
\end{proof}

\begin{figure}[h!]
  \alignpic
  \begin{subfigure}{0.35\textwidth}
    \centering
    \input{fig/graph_two_nodes.tikz}
    \caption{$G$}
    \label{fig:graph_two_nodes}
  \end{subfigure}
  \begin{subfigure}{0.35\textwidth}
    \centering
    \input{fig/graph_three_nodes.tikz}
    \caption{$G'=\mir(G,\{u\})$}
    \label{fig:graph_three_nodes}
  \end{subfigure}
  \caption{Two $\{\blank\}$-labeled $\{\blank\}$-graphs related by
    mirroring.}
  \label{fig:graphs_two_and_three_nodes}
\end{figure}

Instead of expecting the automaton to break symmetry through
nondeterminism, we could also rely on additional information provided
by the node labels of the input graph. Well-chosen labels can make the
structure of the input graph visible to a deterministic
automaton. This strong dependence on the node labeling implies the
following negative result.

\begin{lemma}[Projection] \label{lem:ddga-projection}
  The class $\LL_\DDGA$ of DDGA-recognizable graph languages is
  \emph{not} closed under projection.
\end{lemma}

\begin{proof}
  Let $Σ=\{\a,\b,\c\}$ and $Γ=\{\blank\}$. We consider the language
  $\dL[occur]{\a,\b,\c}$ of all $Σ$-labeled graphs in which every node
  label occurs at least once, i.e.,
  \begin{equation*}
    \dL[occur]{\a,\b,\c} = \bigl\{G_λ∈Σ^{\clouded{Γ}} \bigm|
    ∃u,v,w∈\VG\colon λ(u)=\a \,∧\, λ(v)=\b \,∧\, λ(w)=\c \bigr\}.
  \end{equation*}
  This language is recognized by the DDGA $\dA[occur]{\a,\b,\c}$
  specified in \cref{fig:ADGA_occur_abc}. Now, consider the projection
  $h\colon Σ→\{\blank\}$ with $h(\a)=h(\b)=h(\c)=\blank$. Clearly,
  $h(\dL[occur]{\a,\b,\c})$ is equal to the language $\dL[min]{3}$ of
  all $\{\blank\}$-labeled graphs that have at least three nodes. As
  shown in the proof of \cref{lem:ddga<ndga}, this language is not
  DDGA-recognizable.
\end{proof}

\begin{SCfigure}[1.8][h!]
  \alignpic
  \input{fig/ADGA_occur_abc.tikz}
  \caption{$\dA[occur]{\a,\b,\c}$, a DDGA over
    $\bigl\langle\{\a,\b,\c\},\{\blank\}\bigr\rangle$ whose graph
    language consists of the labeled graphs in which each of the three
    node labels occurs at least once.}
  \label{fig:ADGA_occur_abc}
\end{SCfigure}

On the positive side, the union and intersection constructions for
NDGAs remain valid, and complementation becomes trivial.

\begin{lemma}[Closure Properties] \label{lem:ddga-closure}
  The class $\LL_\DDGA$ of DDGA-recognizable graph languages is
  effectively closed under boolean set operations.
\end{lemma}

\begin{proof}
  The product constructions for union and intersection specified in
  the proof of \cref{lem:ndga-closure} yield DDGAs when applied on
  DDGAs.

  It remains to show closure under complementation. Consider any DDGA
  $\A=⟨Σ,Γ,\Q,\ab σ,δ,\F⟩$. Since its transition function is
  deterministic, $\A$ has precisely one run on each labeled graph
  $G_λ∈Σ^{\clouded{Γ}}$. Hence, we obtain a complement automaton $\cA$
  by simply complementing the acceptance condition, i.e.,
  $\cA=⟨Σ,Γ,\Q,\ab σ,δ,2^{Q_\P}\setminus\F⟩$.
\end{proof}

%% file: fig/graph_1_before_mirroring.tikz
\begin{tikzpicture}[lgraph]
  \node[lnode,label=left:$s$] (s) {$\a$};
  \node[lnode,label=above:$t$] (t) at ([shift={(30:\lnodedistIG)}]s) {$\b$};
  \node[lnode,label=below:$u$] (u) at ([shift={(330:\lnodedistIG)}]s) {$\b$};
  \node[lnode] (v) at ([shift={(30:\lnodedistIG)}]t) {$\c$};
  \node[lnode] (noname) at ([shift={(330:\lnodedistIG)}]t) {$\a$};
  \node[lnode] (w) at ([shift={(330:\lnodedistIG)}]u) {$\c$};
  \path[use as bounding box]
    (s)      edge[loop above] (s)
    (t)      edge (v)
    (u)      edge (w)
    (v)      edge[bend left=20] (noname)
    (noname) edge (t)
             edge (u)
             edge[bend left=20] (v)
             edge[bend right=20] (w)
    (w)      edge[bend right=20] (noname);
  \begin{scope}[on background layer, every path/.style={fill=lightblue}]
    \fill[rounded corners=8ex]
      ([xshift=-8.2ex]s.center) -- ([xshift=3.2ex,yshift=7.5ex]t.center) --
      ([xshift=3.2ex,yshift=-7.5ex]u.center) -- cycle;
  \end{scope}
\end{tikzpicture}

%% file: fig/graph_1_after_mirroring.tikz
\begin{tikzpicture}[lgraph, trim right=(z)]
  \node[lnode,label=left:$s$] (s) {$\a$};
  \node[lnode,label=above:$t$] (t) at ([shift={(30:\lnodedistIG)}]s) {$\b$};
  \node[lnode,label=below:$u$] (u) at ([shift={(330:\lnodedistIG)}]s) {$\b$};
  \node[lnode] (v) at ([shift={(30:\lnodedistIG)}]t) {$\c$};
  \node[lnode] (noname) at ([shift={(330:\lnodedistIG)}]t) {$\a$};
  \node[lnode] (w) at ([shift={(330:\lnodedistIG)}]u) {$\c$};
  \node[lnode,label={[yshift=-.65ex]above:$\phantom{\;=f_1(t)}\fade{\unfade{x}=f_1(t)}$}] (x) at ([shift={(330:\lnodedistIG)}]v) {$\b$};
  \node[lnode,label={[yshift=.8ex]below:$\phantom{\;=f_1(u)}\fade{\unfade{y}=f_1(u)}$}] (y) at ([shift={(330:\lnodedistIG)}]noname) {$\b$};
  \node[lnode,label={right:$\fade{\unfade{z}=f_1(s)}$}] (z) at ([shift={(330:\lnodedistIG)}]x) {$\a$};
  \path[use as bounding box]
     (s)      edge[loop above] (s)
     (t)      edge (v)
     (u)      edge (w)
     (v)      edge[bend left=20] (noname)
     (noname) edge (t)
              edge (x)
              edge (u)
              edge (y)
              edge[bend left=20] (v)
              edge[bend right=20] (w)
     (w)      edge[bend right=20] (noname)
     (z)      edge[loop above, counter clockwise] (z)
     (x)      edge (v)
     (y)      edge (w);
  \begin{scope}[on background layer, every path/.style={fill=lightblue}]
    \fill[rounded corners=8ex]
      ([xshift=-8.2ex]s.center) -- ([xshift=3.2ex,yshift=7.5ex]t.center) --
      ([xshift=3.2ex,yshift=-7.5ex]u.center) -- cycle;
    \fill[rounded corners=8ex]
      ([xshift=8.2ex]z.center) -- ([xshift=-3.2ex,yshift=7.5ex]x.center) --
      ([xshift=-3.2ex,yshift=-7.5ex]y.center) -- cycle;
  \end{scope}
\end{tikzpicture}

%% file: fig/graph_2_before_mirroring.tikz
\begin{tikzpicture}[lgraph]
  \node[lnode] (s) {$\a$};
  \node[lnode,label=above:$t$] (t) at ([shift={(30:\lnodedistIG)}]s) {$\b$};
  \node[lnode,label=above:$v$] (v) at ([shift={(30:\lnodedistIG)}]t) {$\c$};
  \node[lnode] (noname) at ([shift={(330:\lnodedistIG)}]t) {$\a$};
  \node[lnode,label=above:$x$] (x) at ([shift={(330:\lnodedistIG)}]v) {$\b$};
  \node[lnode] (z) at ([shift={(330:\lnodedistIG)}]x) {$\a$};
  \path[use as bounding box]
     (s)      edge[loop above] (s)
     (t)      edge (v)
     (v)      edge[bend left=20] (noname)
     (noname) edge (t)
              edge (x)
              edge[bend left=20] (v)
     (z)      edge[loop above, counter clockwise] (z)
     (x)      edge (v);
  \begin{scope}[on background layer, every path/.style={fill=lightblue}]
    \fill[rounded corners=8ex]
      ([yshift=7.5ex]v.center) -- ([xshift=-7.2ex,yshift=-2.3ex]t.center) --
      ([xshift=7.2ex,yshift=-2.3ex]x.center) -- cycle;
  \end{scope}
\end{tikzpicture}

%% file: fig/graph_2_after_mirroring.tikz
\begin{tikzpicture}[lgraph, trim right=(z)]
  \node[lnode,label=left:\phantom{$s$}] (s) {$\a$};
  \node[lnode,label=above:$t$] (t) at ([shift={(30:\lnodedistIG)}]s) {$\b$};
  \node[lnode,label={[yshift=.8ex]below:$\fade{f_2(t)=\unfade{u}}\phantom{;f_2(t)=}$}] (u) at ([shift={(330:\lnodedistIG)}]s) {$\b$};
  \node[lnode,label=above:$v$] (v) at ([shift={(30:\lnodedistIG)}]t) {$\c$};
  \node[lnode] (noname) at ([shift={(330:\lnodedistIG)}]t) {$\a$};
  \node[lnode,label={[yshift=.8ex]below:$\phantom{;=f_2(v)}\fade{\unfade{w}=f_2(v)}$}] (w) at ([shift={(330:\lnodedistIG)}]u) {$\c$};
  \node[lnode,label=above:$x$] (x) at ([shift={(330:\lnodedistIG)}]v) {$\b$};
  \node[lnode,label={[yshift=.8ex]below:$\phantom{;=f_2(x)}\fade{\unfade{y}=f_2(x)}$}] (y) at ([shift={(330:\lnodedistIG)}]noname) {$\b$};
  \node[lnode,label=right:\phantom{$z=f_1(s)$}] (z) at ([shift={(330:\lnodedistIG)}]x) {$\a$};
  \path[use as bounding box]
     (s)      edge[loop above] (s)
     (t)      edge (v)
     (u)      edge (w)
     (v)      edge[bend left=20] (noname)
     (noname) edge (t)
              edge (x)
              edge (u)
              edge (y)
              edge[bend left=20] (v)
              edge[bend right=20] (w)
     (w)      edge[bend right=20] (noname)
     (z)      edge[loop above, counter clockwise] (z)
     (x)      edge (v)
     (y)      edge (w);
  \begin{scope}[on background layer, every path/.style={fill=lightblue}]
    \fill[rounded corners=8ex]
      ([yshift=7.5ex]v.center) -- ([xshift=-7.2ex,yshift=-2.3ex]t.center) --
      ([xshift=7.2ex,yshift=-2.3ex]x.center) -- cycle;
    \fill[rounded corners=8ex]
      ([yshift=-7.5ex]w.center) -- ([xshift=-7.2ex,yshift=2.3ex]u.center) --
      ([xshift=7.2ex,yshift=2.3ex]y.center) -- cycle;
  \end{scope}
\end{tikzpicture}

%% file: fig/ADGA_max_two.tikz
\begin{tikzpicture}[automaton]
  \matrix[states] {
                                                  &  \node[permanent] (q_1) {$q_1$}; \\
    \node[initial,universal] (q_ini) {$\q{ini}$}; &  \node[permanent] (q_2) {$q_2$}; \\
                                                  &  \node[permanent] (q_3) {$q_3$}; \\
  };
  \matrix[accepting sets] {
    \x &    &    & \x & \x &    \\
       & \x &    & \x &    & \x \\
       &    & \x &    & \x & \x \\
  };
  \DrawColumnBackground{1}{3}{6}
  \path (q_ini.15) edge (q_1)
        (q_ini)    edge (q_2)
                   edge (q_3);
\end{tikzpicture}

%% file: fig/ADGA_min_three.tikz
\begin{tikzpicture}[automaton]
  \matrix[states] {
                                                    &  \node[permanent] (q_1) {$q_1$}; \\
    \node[initial,existential] (q_ini) {$\q{ini}$}; &  \node[permanent] (q_2) {$q_2$}; \\
                                                    &  \node[permanent] (q_3) {$q_3$}; \\
  };
  \matrix[accepting sets] {
    \x \\
    \x \\
    \x \\
  };
  \DrawColumnBackground{1}{3}{1}
  \path (q_ini) edge (q_1)
                edge (q_2)
                edge (q_3);
\end{tikzpicture}

%% file: fig/graph_two_nodes.tikz
\begin{tikzpicture}[lgraph]
  \node[lnode,label=below:$u$] (u) {};
  \node[lnode,label=below:\phantom{$f(u)$}] (v) at ([shift={(0:\lnodedistIG)}]u) {};
  \path (u) edge (v);
\end{tikzpicture}

%% file: fig/graph_three_nodes.tikz
\begin{tikzpicture}[lgraph]
  \node[lnode,label=below:$u$] (u) {};
  \node[lnode] (v) at ([shift={(0:\lnodedistIG)}]u) {};
  \node[lnode,label=below:$f(u)$] (w) at ([shift={(0:\lnodedistIG)}]v) {};
  \path (u) edge (v)
        (w) edge (v);
\end{tikzpicture}

%% file: fig/ADGA_occur_abc.tikz
\begin{tikzpicture}[automaton]
  \matrix[states] {
    \node[permanent] (q_a) {$q_\a$}; \\
    \node[permanent] (q_b) {$q_\b$}; \\
    \node[permanent] (q_c) {$q_\c$}; \\
  };
  \matrix[symbols] { 
    \node (a) {$\a$}; \\
    \node (b) {$\b$}; \\
    \node (c) {$\c$}; \\
  };
  \matrix[accepting sets] {
    \x \\
    \x \\
    \x \\
  };
  \DrawColumnBackground{1}{3}{1}
  \path (a) edge (q_a)
        (b) edge (q_b)
        (c) edge (q_c);
\end{tikzpicture}

%% file: chap/conclusion.tex

\chapter{Conclusion}
We first summarize and comment upon the results obtained in this work,
and then conclude with a small selection of open questions that seem
worth pursuing.

\section{Commented Summary}
We have introduced ADGAs, a new class of finite graph
automata. Although many graph automaton models have been defined over
the last decades, ADGAs are probably the first automata to be
equivalent to MSO-logic on graphs. In this regard, it seems remarkable
that the individual ingredients of this model of computation are not
spectacular at all, and, for the most part, well-established.

To a certain extent, ADGAs can be considered as synchronous
distributed algorithms, where each node is limited to a finite-state
machine. The model is further restricted by a constant running time
and the fact that nodes see only an abstract representation of their
neighborhood, in the form of sets of states, which drastically limits
their ability to distinguish between different neighbors. The latter
restriction allows ADGAs to operate on graphs of unbounded degree, a
feature that sets them apart from many other types of finite graph
automata defined in the literature (for instance in \cite{WR79} and
\cite{Tho91}).

Besides this distributed character, there is also a centralized aspect
to ADGAs, which contributes greatly to their expressive power. On the
one hand, acceptance is decided on a global level, based on the set of
states reached by the local processors. This combines and generalizes
two decision-making approaches that may appear more natural in a
distributed setting: decision by a unique leader and decision by
unanimous agreement of all the nodes. On the other hand, ADGAs
implement the powerful concept of alternation, a kind of
parallelization, which is also expressed in terms of the global
configuration of the entire system.

As already mentioned above, taken in isolation these concepts
represent nothing new. The contribution of the present thesis is
mainly to combine them into a model of computation that balances
between distribution and centralization in a way that matches
precisely MSO-logic.

For finite automata on words and (bottom-up) tree automata,
alternation and nondeterminism do not increase expressiveness (see
\cite[Thm~5.2]{CKS81} and \cite[Thm~7.4.1]{TATA08}). For our graph
automata, on the other hand, they are essential ingredients that
cannot be eliminated without losing expressive power. We have seen
that the deterministic, nondeterministic and alternating variants of
distributed graph automata form a strict hierarchy, i.e.,
\begin{equation*}
  \LL_\DDGA \hyperref[lem:ddga<ndga]{⊂} \LL_\NDGA 
  \hyperref[lem:ndga<adga]{⊂} \LL_\ADGA.
\end{equation*}
On an intuitive level, this is not very surprising, since
nondeterministic choices and universal branchings are closely related
to existential and universal quantification in MSO-logic, and removing
one type of quantifier (without allowing to negate the other)
drastically diminishes expressiveness.

Another way to look at this is from the perspective of the closure
properties of the three variants of automata, which are summarized in
\cref{tab:closure-decidability}. As already mentioned
\hyperlink{engelfriet-charact}{in-between} the two proofs of
\cref{sec:adga=mso}, Engelfriet has characterized the class of
MSO-definable graph languages as the smallest class that contains
certain elementary graph languages and is closed under boolean set
operations and under projection. To achieve closure under projection
with our distributed automata, we need nondeterminism so that the
nodes can guess which label they might have had without application of
the projection function. But if the additional expressive power
introduced by nondeterministic choices is not matched by the
corresponding dual, namely universal branchings, then there is an
asymmetry that makes us lose closure under complementation.

\begin{table}[h!]
  \alignpic
  \begin{tabular}{lccccc}
    \toprule
         & \multicolumn{4}{c}{Closure Properties} & Decidability \\
    \cmidrule(rl){2-5} \cmidrule(rl){6-6}
         & Complement & Union & Intersection & Projection & Emptiness \\
    \addlinespace
    ADGA & \hyperref[lem:complementation]{\cmark} & \hyperref[lem:union-intersect]{\cmark}
         & \hyperref[lem:union-intersect]{\cmark} & \hyperref[lem:projection]{\cmark} 
         & \hyperref[cor:adga-emptiness]{\xmark} \\
    \addlinespace
    NDGA & \hyperref[lem:ndga-complementation]{\xmark} & \hyperref[lem:ndga-closure]{\cmark}
         & \hyperref[lem:ndga-closure]{\cmark} & \hyperref[lem:ndga-closure]{\cmark}
         & \hyperref[lem:ndga-emptiness]{\cmark} \\
    \addlinespace
    DDGA & \hyperref[lem:ddga-closure]{\cmark} & \hyperref[lem:ddga-closure]{\cmark}
         & \hyperref[lem:ddga-closure]{\cmark} & \hyperref[lem:ddga-projection]{\xmark}
         & \hyperref[lem:ndga-emptiness]{\cmark} \\
    \bottomrule
  \end{tabular}
  \caption{Closure and decidability properties of alternating,
    nondeterministic, and deterministic distributed graph automata.}
  \label{tab:closure-decidability}
\end{table}

However, as also indicated in \cref{tab:closure-decidability}, if we
do not enable universal branchings, it has the positive effect that
emptiness remains decidable, which is not possible anymore once we
have reached the expressive power of MSO-logic on graphs. The reason
for this positive decidability result is that, for NDGAs, any run on a
sufficiently large input graph contains redundancies. This allows us
to narrow down the search space to a finite number of graphs.

\section{Open Questions}
In the present work, we have not gone much further than introducing
new definitions. Whether these definitions are sensible largely
depends on the insights that we might gain through them. In this
regard, it would be very interesting to obtain answers to the
following questions.

\begin{description}[style=nextline]
\item[Logics Equivalent to NDGAs and DDGAs.] Although we have been
  mostly concerned with ADGAs here, the fact that emptiness is
  decidable for NDGAs and DDGAs might make their corresponding classes
  of graph languages attractive, despite their less robust closure
  properties. As mentioned in \cref{sec:neg-impl-adga}, Trakhtenbrot’s
  Theorem states that, even when restricting formulas to first-order
  logic, satisfiability is undecidable on graphs. This tells us that
  alternation in our graph automata is required to even cover the
  expressiveness of first-order logic, but also arouses curiosity
  regarding the logical equivalent of the weaker variants. \emph{What
    would be logical formalisms on graphs that precisely define the
    NDGA- and DDGA-recognizable graph languages?}
\item[Alternative Definitions of ADGAs.] The definition of ADGAs given
  in this thesis (\cref{def:adga}) seems quite involved, and to
  loosely paraphrase a famous French aviator, as long as there remains
  something nonessential to remove, there is also room for
  improvement. On some classes of graphs, the necessary running time
  of ADGAs can be bounded by a fixed constant. For instance, if we
  only consider (graphs representing) words, simulating the classical
  finite automata shows us that every regular language is recognizable
  by some ADGA of length~2. (Nondeterministically guess a run of the
  word automaton in the first round, then check that it is legal and
  accepting in the second round.) But on general graphs, there is no
  such fixed constant, as can be easily inferred from the infinity of
  the MSO quantifier alternation hierarchy investigated by Matz and
  Thomas in \cite{MT97}. The question remains: \emph{can we impose
    simplifications or restrictions on the definition of ADGAs without
    sacrificing expressive power?}
\item[Impact on other Research.] Since ADGAs arose from an open-ended
  question, it seems only fitting to conclude this thesis with other
  questions of that type.
  \begin{itemize}
  \item On words and trees, the equivalence between finite automata
    and MSO-logic led to the decidability of the satisfiability and
    validity problems of MSO-logic. Unfortunately, this is not
    extendable to graphs. But maybe ADGAs can help us finding
    necessary conditions for MSO-definability, similar to what we got
    for NDGA-recognizability in the Mirroring Lemma
    (\cref{lem:weak-mirroring}). More generally, we might ask:
    \emph{what can ADGAs tell us about the class of MSO-definable
      graph languages?}
  \item Conversely, we could also use MSO-logic as a means to an
    end. As already mentioned, ADGAs can be considered, to some
    extent, as distributed algorithms. \emph{What can the connection
      to MSO-logic tell us about distributed algorithms?}
  \end{itemize}
\end{description}